\pdfoutput=1

\documentclass[letter,journal,draftcls,onecolumn,12pt]{IEEEtran}
\IEEEoverridecommandlockouts

\usepackage[utf8]{inputenc}
\usepackage[T1]{fontenc}
\usepackage{url}
\usepackage{ifthen}
\usepackage[cmex10]{amsmath} 
\usepackage{cite}
\usepackage{psfrag,graphicx}
\usepackage[tight,footnotesize]{subfigure}
\usepackage[cmex10]{amsmath}
\usepackage{amsthm}
\usepackage{bbm}
\usepackage{vmr-symbols-vecbold}
\usepackage{standard-macros}
\usepackage{subfigure}

\definecolor{Red}{rgb}{1.0, 0, 0}

\newtheorem{theorem}{Theorem}
\newtheorem{corollary}{Corollary}

\newtheorem{definition}{Definition}

\begin{document}
\title{Testing Optimality of\\Sequential Decision-Making}

\author{Meik~D\"orpinghaus,~\IEEEmembership{Member,~IEEE}, Izaak Neri, \'{E}dgar Rold\'{a}n, Heinrich~Meyr,~\IEEEmembership{Life~Fellow,~IEEE}, and Frank J\"ulicher
\thanks{This work has been partly supported by the German Research Foundation (DFG) within the Cluster of Excellence EXC 1056 'Center for Advancing Electronics Dresden (cfaed)' and within the CRC 912 'Highly Adaptive Energy-Efficient Computing (HAEC)'. The material in this paper has been presented in part at the IEEE International Symposium on Information Theory (ISIT), Aachen, Germany, June 2017 \cite{dorpinghaus2017information}.}
\thanks{M.~D\"orpinghaus is with the Vodafone Chair Mobile Communications Systems and with the Center for Advancing Electronics Dresden (cfaed), Technische Universit\"at Dresden, 01062 Dresden, Germany (e-mail: meik.doerpinghaus@tu-dresden.de).}
\thanks{Izaak Neri is with the Max-Planck-Institute for the Physics of Complex Systems, Dresden, Germany, with the Max-Planck-Institute of Molecular Cell Biology and Genetics, Dresden, Germany, and with the Center for Advancing Electronics Dresden (cfaed), Technische Universit\"at Dresden, 01062 Dresden, Germany (e-mail: izaakneri@posteo.net).}
\thanks{\'{E}dgar Rold\'{a}n is with the Max-Planck-Institute for the Physics of Complex Systems, Dresden, Germany,  with the Center for Advancing Electronics Dresden (cfaed), Technische Universit\"at Dresden, 01062 Dresden, Germany, and with GISC -- Grupo Interdisciplinar de Sistemas Complejos, Madrid, Spain (e-mail: edgar@edgarroldan.com).}
\thanks{H.~Meyr is an emeritus of the Institute for Integrated Signal Processing Systems, RWTH Aachen University, 52056 Aachen, Germany and is now a grand professor of the Center for Advancing Electronics Dresden (cfaed) at Technische Universit\"at Dresden, 01062 Dresden, Germany (e-mail: meyr@iss.rwth-aachen.de).}
\thanks{Frank J\"ulicher is with the Max-Planck-Institute for the Physics of Complex Systems, Dresden and with the Center for Advancing Electronics Dresden (cfaed), Technische Universit\"at Dresden, 01062 Dresden, Germany, Germany (e-mail: julicher@pks.mpg.de).}}

\maketitle

\begin{abstract}
This paper provides a statistical method to test whether a system that performs a binary sequential hypothesis test is optimal in the sense of minimizing the average decision times while taking decisions with given reliabilities. The proposed method requires samples of the decision times, the decision outcomes, and the true hypotheses, but does not require knowledge on the statistics of the observations or the properties of the decision-making system. The method is based on fluctuation relations for decision time distributions which are proved for sequential probability ratio tests. These relations follow from the martingale property of probability ratios and hold under fairly general conditions. We illustrate these tests with numerical experiments and discuss potential applications.
\end{abstract}


\section{Introduction}
In sequential decision-making it is important to make fast and reliable decisions. In this regard, consider, e.g., an autonomous car which has to decide whether an obstacle is present or not on the road. Such decisions are executed by dedicated signal processing algorithms. These algorithms should use the available measurements in an optimal way such that the average time to take a decision is minimized. For practical use it is key to test if the implemented decision algorithms achieve the optimum performance, i.e., if decisions are made as fast as possible with a given reliability. 

Sequential decision-making has been first mathematically formulated in the seminal work by A.\ Wald who introduced a sequential probability ratio test \cite{Wald1945}. Wald's test takes binary decisions on two hypotheses based on sequential observations of a stochastic process. For independent and identically distributed (i.i.d.) observations this test yields the minimum mean decision time for decisions with a given probability of error and a given hypothesis \cite{wald1948optimum}. Wald's test accumulates the likelihood ratio given by the sequence of observations and decides as soon as this cumulative likelihood ratio exceeds or falls below two given thresholds which depend on the required reliability of the decision. A key characteristic of such a sequential test is that its termination time is a random quantity depending on the actual realization of the observation sequence. The Wald test has been applied to non i.i.d.\ observation processes,  nonhomogeneous and correlated continuous-time processes, and has been generalized for multiple hypotheses  \cite{tartakovsky2014sequential};  general optimality criteria for sequential probability ratio tests have been proved when probabilities of errors tend to zero, see e.g.~\cite{lai1981asymptotic, tartakovsky1998asymptotically, tartakovsky1998asymptotic, draglia1999multihypothesis, tartakovsky2014sequential}.

Now we consider the decision-making device as a {\em black box} which takes as input the observation process, corresponding to one of two hypotheses, and gives as output a binary decision variable at a random decision time. Can we determine whether this decision-making device is optimal based on the statistics of the output of the device --- the decisions and the decision times --- and the knowledge of the true hypothesis? Indeed, in the present paper we introduce a test for optimality of sequential decision-making based on necessary conditions for optimality. Notably, this test does \emph{not} require knowledge of the realizations or the statistics of the observation processes.

We first consider a device which takes as input the realization of a  {\it continuous} stochastic process corresponding to one of the two hypothesis $H_1$ or $H_2$, and gives as output a binary decision variable $\mathsf{D}\in\{1,2\}$ (corresponding to the hypotheses $H_1$ and $H_2$, respectively) at the random decision time $\mathsf{T}$ elapsed since the beginning of the observations. We will show that optimality of sequential probability ratio tests --- in the sense that the mean decision time is minimized while fulfilling given reliability constraints --- requires that the following conditions on the distribution of the decision time $\mathsf{T}$ hold
\begin{IEEEeqnarray}{rCL}
p_{\mathsf{T}}(t|\mathsf{H}=1,\mathsf{D}=1)&=&p_{\mathsf{T}}(t|\mathsf{H}=2,\mathsf{D}=1)\label{OptCond1}\\
p_{\mathsf{T}}(t|\mathsf{H}=1,\mathsf{D}=2)&=&p_{\mathsf{T}}(t|\mathsf{H}=2,\mathsf{D}=2)\label{OptCond2}.
\end{IEEEeqnarray}
Here, $p_{\mathsf{T}}$ is the probability density of the decision time and $\mathsf{H}\in\{1,2\}$ (corresponding to the hypotheses $H_1$ and $H_2$) denotes the random binary hypothesis. The necessary conditions (\ref{OptCond1}) and~(\ref{OptCond2}) for optimality imply that the distribution of the decision time $\mathsf{T}$ given a certain outcome are independent of the actual hypothesis. Moreover, this implies that the decision time $\mathsf{T}$ of the optimal sequential test does not contain any information on which hypothesis is true beyond the decision outcome $\mathsf{D}$. As a consequence, we can quantify the optimality of a given black-box test by measuring the mutual information between the hypothesis $\mathsf{H}$ and the decision time $\mathsf{T}$ conditioned on the output of the test $\mathsf{D}$, i.e., $I(\mathsf{H};\mathsf{T}|\mathsf{D})$. In case the test is optimal it must hold that 
\begin{IEEEeqnarray}{rCL}
I(\mathsf{H};\mathsf{T}|\mathsf{D})&=&0.\label{MutInfCondOpt_Intro}
\end{IEEEeqnarray}

Based on the following example it can be seen that (\ref{MutInfCondOpt_Intro}) is a necessary but not a sufficient condition for optimality in the sense of minimizing the mean decision time given a certain reliability. Consider we have an optimal decision device using the Wald test. Now we delay all decisions by a constant time $t_{\rm delay}$. Still $I(\mathsf{H};\mathsf{T}+t_{\rm delay}|\mathsf{D})=0$ with $\mathsf{T}$ being the decision time of the Wald test. Indeed, (\ref{MutInfCondOpt_Intro}) is not a sufficient condition for the minimal mean decision time, but rather a measure for the optimal usage of information by the decision device. If $I(\mathsf{H};\mathsf{T}|\mathsf{D})>0$ this means that the decision time $\mathsf{T}$ contains additional information on the hypothesis $\mathsf{H}$ beyond the actual decision $\mathsf{D}$ implying that the decision device does not exploit all the available information. Hence, $I(\mathsf{H};\mathsf{T}|\mathsf{D})$ measures the degree of divergence from optimality in the sense of optimal usage of information. For practical purposes it is easier to test whether the black-box decision device fulfills the optimality condition in (\ref{MutInfCondOpt_Intro}) rather than testing if the decision device decides with the minimum mean decision time, since the minimum mean decision time is in general not known. Furthermore, consider that in experimental setups we can measure a decision output $\mathsf{D}$ at a certain time $\mathsf{T}+\mathsf{T}_{\rm delay}$, with $\mathsf{T}_{\rm delay}$ a random delay time, but in general we do not know at which time $\mathsf{T}$ the decision has been taken, as the decision device is a black box device and we cannot clearly separate the actual decision making process from nondecision processes.    If the decision time $\mathsf{T}_{\rm delay}$ is statistical independent of $\mathsf{H}$ when conditioned on $\mathsf{D}$ and $\mathsf{T}$, then $I(\mathsf{H};\mathsf{T}|\mathsf{D})=0$ implies that $I(\mathsf{H};\mathsf{T}+\mathsf{T}_{\rm delay}|\mathsf{D})=0$, and $I(\mathsf{H};\mathsf{T}+\mathsf{T}_{\rm delay}|\mathsf{D})$ can be used as a necessary condition to test optimality of decision devices. If additionally $I(\mathsf{H};\mathsf{T}_{\rm delay}|\mathsf{D}, \mathsf{T}+\mathsf{T}_{\rm delay})=0$, then $I(\mathsf{H};\mathsf{T}|\mathsf{D}) = I(\mathsf{H};\mathsf{T}+\mathsf{T}_{\rm delay}|\mathsf{D})$.  In this paper we derive the optimality conditions (\ref{OptCond1}) - (\ref{MutInfCondOpt_Intro}) and generalize them to  {\it discrete-time} observation processes. We furthermore formulate tests for optimal sequential hypothesis testing (sequential decision-making) based on (\ref{OptCond1}) - (\ref{MutInfCondOpt_Intro}). Finally,  we illustrate our results in computer experiments.

The optimality conditions (\ref{OptCond1}) - (\ref{MutInfCondOpt_Intro}) are of interest in different contexts. For example these conditions allow to test optimality of sequential decision-making in engineered devices. An example are decision-making devices based on machine learning such as deep neural networks. Such algorithms have the power to solve very complex tasks and adapt to specific environments by learning. The advantage of machine learning is that not all environmental situations need to be learned at design time, which for many applications like self driving cars is not practical. However, while neural networks exhibit best performance in comparison to other approaches the principles which lead to their successful operation are yet unclear. It would be very useful to quantify if decisions made on such deep learning approaches are close to optimal. In addition to engineering, the tests proposed in this paper could also allow to understand if specific biological systems use all available information optimally to make reliable decisions on the fly. In this regard, consider the example of sequential decision-making by humans in  two-choice decision tasks based on perceptual stimuli or biological cells making decisions on their fate based on extracellular cues. We discuss the application of the tests for optimality to these examples in more detail below.

\paragraph*{Notation}
We denote random variables by upper case sans serif letters, e.g., $\mathsf{X}$. All random quantities are defined on the measurable space $(\Omega,\mathcal{F})$ and are governed by the probability measure $\mathbb{P}$. The probability density function of a random variable $\mathsf{X}$ given $\mathsf{Y}=y$ is written $p_{\mathsf{X}}(x|\mathsf{Y}=y)$.        Moreover, for discrete random variables $P(\mathsf{X}=x|\mathsf{Y}=y)$ denotes the probability of $\mathsf{X}=x$ given $\mathsf{Y}=y$.  The restriction of the measure $\mathbb{P}$ to a sub-$\sigma$ algebra $\mathcal{G} \subseteq\mathcal{F}$ is written as $\left.\mathbb{P}\right|_{\mathcal{G}}$.  
Finally, $\log$ denotes the natural logarithm and $\log_2$ is the logarithm w.r.t.\ base $2$.       The mutual information and the conditional mutual information are defined by $I(\mathsf{X};\mathsf{Y}) = \mathrm{E}\left[\log_2 \frac{p_{\mathsf{X}}\left(\mathsf{X}|\mathsf{Y}\right)}{p_{\mathsf{X}}\left(\mathsf{X}\right)}\right]$ and $I(\mathsf{X};\mathsf{Y}|\mathsf{Z}) = \mathrm{E}\left[ \log_2 \frac{p_{\mathsf{X}}\left(\mathsf{X}|\mathsf{Y}, \mathsf{Z}\right)}{p_{\mathsf{X}}\left(\mathsf{X}|\mathsf{Z}\right)}\right]$, respectively, where the mathematical expectation $\mathrm{E}[\cdot]$ is taken with respect to the measure $\mathbb{P}$.  

\paragraph*{Organization of the Paper}

The paper is organized as follows.  After the introduction, we describe the system setup in detail in Section \ref{Sect_System_Cont} where we also give a precise problem formulation including definitions of optimality for sequential decision-making.    Subsequently, in Section \ref{sec:3} we derive the main theorems and corollaries describing properties of optimal  sequential probability ratio tests for   the case of continuous observation processes.   In Section \ref{sec:4} for certain conditions we extend  these theorems and corollaries to the discrete-time scenario.  In Section \ref{SectionTests} we formulate statistical tests to decide whether a black-box decision device performs optimal sequential decision-making based on the theorems and corollaries derived in Section  \ref{sec:3} and   \ref{sec:4}, and we also discuss how to measure the distance of the black-box decision device to optimality.   We illustrate the application of these tests based on numerical experiments in Section \ref{SectNumericalExperminents}.   In Section  \ref{sec:7}  we discuss the applicability and the limitations of the provided tests for optimality.
Readers who are mainly  interested in the application of the statistical tests for optimality may skip Sections \ref{sec:3} and \ref{sec:4}.

\section{System Setup and Problem Formulation}\label{Sect_System_Cont}
\subsection{System Setup}
We consider a sequential binary decision problem based on an observation process $\mathsf{X}_t$ with the time index $t$ either discrete, $t\in\mathbb{Z}_+$, or continuous, $t\in\mathbb{R}_+$. The stochastic process $\mathsf{X}_t$ is generated by one of two possible models corresponding to two hypotheses $H_1$ and $H_2$. To describe the statistics of the process $\mathsf{X}_t$ we consider the filtered probability space $(\Omega,\mathcal{F},\{\mathcal{F}_t\}_{t\ge 0},\mathbb{P})$ with  $\{\mathcal{F}_t\}_{t\ge 0}$ the natural filtration generated by the observation process $\mathsf{X}_t$ and the hypothesis $\mathsf{H}$.  We consider $\mathsf{H}$ to be a time independent random variable.  The statistics of the observation process under the two hypothesis are described by the conditional probability measures given the hypothesis $\mathbb{P}_{l}\left[\Phi\right] = \mathrm{E}\left[1_{\Phi}|\mathsf{H}=l\right]$ with $l\in\left\{1,2\right\}$ corresponding to the hypothesis $H_1$ and $H_2$, respectively \cite{liptser2001statistics}; here $1_{\Phi}(\omega)$ is the indicator function on the set $\Phi$.   We also consider the filtered probability spaces  $(\Omega,\mathcal{F},\{\mathcal{F}_t\}_{t\ge 0},\mathbb{P}_l)$ with $l\in\left\{1,2\right\}$  associated with the two hypotheses. 
 We consider for continuous-time processes that the filtration $\{\mathcal{F}_t\}_{t\ge 0}$ is right-continuous \cite{tartakovsky2014sequential}, i.e., $\mathcal{F}_t=\cap_{s>t}\mathcal{F}_s$ for all times $t\in\mathbb{R}_+$.

A sequential test makes binary decisions based on sequential observations of the process $\mathsf{X}_t$ and tries to guess which of the hypotheses $H_1$ and $H_2$ is true.     A sequential test $\delta = (\mathsf{D}, \mathsf{T})$ returns 
 a binary output $\mathsf{D}$ at a random time $\mathsf{T}$.   The decision time $\mathsf{T}$ is a stopping time,  which is determined by the time when $\mathsf{X}_t$ satisfies for the first time a certain criterion.  This stopping rule is non-anticipating in the sense that it depends only on observations of the input sequence up to the current time, i.e., the decision causally depends on the observation process. The decision function is a map from the trajectory $\{\mathsf{X}_t\}_0^{\mathsf{T}}$ to $\{1,2\}$, which determines the  decision of the test.

 We now consider the following class of sequential tests with given reliabilities
\begin{IEEEeqnarray}{rCL}
\mathcal{C}(\alpha_1,\alpha_2)&=&\{\delta\!: P(\mathsf{D}=2|\mathsf{H}=1)\!\le\! \alpha_2, P(\mathsf{D}=1|\mathsf{H}=2)\!\le \!\alpha_1,\mathrm{E}[\mathsf{T}|\mathsf{H}=i]<\infty, i\in\{1,2\}\!\}\nonumber\\\label{TestReq}
\end{IEEEeqnarray}
where $\mathrm{E}[\mathsf{T}|\mathsf{H}=i]$ denotes the expected termination time in case hypothesis $i$ is true and where the expectation is taken over the observation sequences $\mathsf{X}_t$. Moreover, $\alpha_1$ and $\alpha_2$ are the maximum allowed error probabilities of the two error types. We assume that $\alpha_1,\alpha_2<0.5$. Notice that we restrict ourselves to tests which terminate almost surely. This assumption is fulfilled in many cases like the case of i.i.d.\ observation processes  \cite[Th.~6.2-1]{Melsa1978} and stationary observation processes.   Note that the class of sequential tests given by $\mathcal{C}(\alpha_1,\alpha_2)$ does not consider prior knowledge on the statistics of $\mathsf{H}$.

We define the following optimality criterion. 
\begin{definition}[Optimality in terms of mean decision times]\label{DefinitionMinimumMeanTime}
An optimal test $\delta^{*}=(\mathsf{D}^{*},\mathsf{T}^{*})$ minimizes the two mean decision times $\mathrm{E}[\mathsf{T}^{*}|\mathsf{H}=i]$ corresponding to the hypothesis $i=1$ and $i=2$ for a given reliability, i.e.
\begin{IEEEeqnarray}{rCL}
\mathrm{E}[\mathsf{T}^{*}|\mathsf{H}=i]&=&\inf_{\delta\in\mathcal{C}(\alpha_1,\alpha_2)}\mathrm{E}[\mathsf{T}|\mathsf{H}=i],\quad i=1,2.\label{OptDecisionTime}
\end{IEEEeqnarray}
\end{definition}

Note that in Definition~\ref{DefinitionMinimumMeanTime} we assume that there exists a test for which the infimum is attained. If such a test does not exist than we can still find a test for which the two mean decision times are arbitrarily close to their infimum values.

Sequential probability ratio tests or \emph{Wald}-tests are optimal in the sense of Definition~\ref{DefinitionMinimumMeanTime}  \cite{Wald1945}. It has been proved that for the case of time-discrete i.i.d.\ observation processes the Wald test is optimal in the sense of Definition~\ref{DefinitionMinimumMeanTime}  \cite{wald1948optimum}. Furthermore, under broad conditions for the observation process it has been proved that sequential probability ratio tests are optimal in the sense of Definition~\ref{DefinitionMinimumMeanTime} in the limit of small error probabilities \cite{lai1981asymptotic, tartakovsky1998asymptotically, tartakovsky1998asymptotic, draglia1999multihypothesis, tartakovsky2014sequential}. 

For discrete-time and i.i.d.\ processes the Wald test collects observations $\mathsf{X}_t$ (which can be understood as samples of a corresponding continuous-time process) until the cumulated log-likelihood ratio 
\begin{IEEEeqnarray}{rCL}
\mathsf{S}_k&=&\sum_{n=1}^{k}\mathsf{\Delta}_n=\sum_{n=1}^{k}\log\left(\frac{p_{\mathsf{X}}(\mathsf{X}_n|\mathsf{H}=1)}{p_{\mathsf{X}}(\mathsf{X}_n|\mathsf{H}=2)}\right)\quad\textrm{for } k\ge 1\label{Cumulated_LLR}
\end{IEEEeqnarray}
exceeds (falls below) a prescribed threshold $L_1$ ($L_2$) for the first time. In (\ref{Cumulated_LLR}) $\mathsf{\Delta}_n$ are the increments of the log-likelihood ratio at time instant $n$. The test decides $\mathsf{D}=1$ ($\mathsf{D}=2$), i.e., for $H_1$ ($H_2$), when $\mathsf{S}_k$ first crosses $L_1$ ($L_2$). In (\ref{Cumulated_LLR}), $p_{\mathsf{X}}(\cdot|\mathsf{H})$ denotes the probability density function of the observations $\mathsf{X}_{k}$ conditioned on the event $\mathsf{H}$. The thresholds $L_1$ and $L_2$ depend on the maximum allowed probabilities for making a wrong decision $\alpha_1$ and $\alpha_2$. A decision with the given reliability constraints $\alpha_1$ and $\alpha_2$ can be made when the cumulative log-likelihood ratio $\mathsf{S}_k$ for the first time crosses one of the thresholds before crossing the opposite one. The thresholds are functions of $\alpha_1$ and $\alpha_2$. In general, the thresholds $L_1$ and $L_2$ are difficult to obtain. However, the optimal thresholds yielding the minimum mean decision time can be approximated by \cite[p.~148]{Melsa1978}
\begin{IEEEeqnarray}{rCL}
L_1&=&\log\frac{1-\alpha_2}{\alpha_1}\label{Def_T1}\\
L_2&=&\log\frac{\alpha_2}{1-\alpha_1}\label{Def_T2}.
\end{IEEEeqnarray}
The choice in (\ref{Def_T1}) and (\ref{Def_T2}) still guarantees that the error constraints in (\ref{TestReq}) are fulfilled. In summary, the sequential probability ratio test decides at the time 
\begin{IEEEeqnarray}{rCL}
\mathsf{T}_{\textrm{Wald}}&=&\min\{k\in\mathbb{N} : \mathsf{S}_k\notin (L_2,L_1)\}\label{Def_T_Wald}  
\end{IEEEeqnarray}
for the decision 
\begin{IEEEeqnarray}{rCL}
\mathsf{D}_{\rm Wald}&=&\left\{\begin{array}{ll}
1 & \textrm{if } \mathsf{S}_{\mathsf{T}_{\mathrm{Wald}}}\ge L_1\\
2 & \textrm{if } \mathsf{S}_{\mathsf{T}_{\mathrm{Wald}}}\le L_2.
\end{array}\right. \label{Def_D_Wald}
\end{IEEEeqnarray}
Analogously, the Wald test for non-i.i.d.\ observation processes is given by (\ref{Def_T_Wald}) and (\ref{Def_D_Wald}) with the log-likelihood ratio
\begin{IEEEeqnarray}{rCL}
\mathsf{S}_k&=&\sum_{n=1}^{k}\mathsf{\Delta}_n=
\sum_{n=1}^{k}\log\frac{p_{\mathsf{X}_n}(\mathsf{X}_n|\mathsf{X}_1^{n-1},\mathsf{H}=1)}{p_{\mathsf{X}_n}(\mathsf{X}_n|\mathsf{X}_1^{n-1},\mathsf{H}=2)}\quad\textrm{for } k\ge 1\label{Cumulated_LLR_noniid}
\end{IEEEeqnarray}
where $\mathsf{X}_1^{n-1}=[\mathsf{X}_1,\hdots,\mathsf{X}_{n-1}]$.

The Wald test can also be formulated for continuous-time observation processes. In this case probability densities of the observation trajectories do not always exist. However, the likelihood ratio $e^{\mathsf{S}_t}$ can be defined in terms of the Radon-Nikod\'ym derivative of the probability space $(\Omega,\mathcal{F},\{\mathcal{F}_t\}_{t\ge 0},\mathbb{P}_1)$ with respect to the probability space $(\Omega,\mathcal{F},\{\mathcal{F}_t\}_{t\ge 0},\mathbb{P}_2)$:
\begin{IEEEeqnarray}{rCL}
e^{\mathsf{S}_t}&=&\frac{\mathrm{d}\mathbb{P}_1|_{\mathcal{F}_t}}{\mathrm{d}\mathbb{P}_2|_{\mathcal{F}_t}}\label{DefLikelihoodCont}
\end{IEEEeqnarray}
with $t\ge 0$. Here, the process $\mathsf{S}_t$ is the cumulative log-likelihood ratio and $\mathbb{P}_i|_{\mathcal{F}_t}$ ($i=1,2$) are restricted measures of $\mathbb{P}_i$ w.r.t.\ the $\sigma$-algebra $\mathcal{F}_t$. A decision with the given reliability constraints $\alpha_1$ and $\alpha_2$ can be made when the cumulative log-likelihood ratio $\mathsf{S}_t$ for the first time crosses one of the thresholds before crossing the opposite one. I.e., the test decides $\mathsf{D} = 1$ ($\mathsf{D} = 2$) in case it crosses $L_1$ ($L_2$) for the first time before crossing $L_2$ ($L_1$) where the thresholds are functions of $\alpha_1$ and $\alpha_2$. For continuous observation processes  $\mathsf{S}_t$ is continuous and the thresholds are exactly given by (\ref{Def_T1}) and (\ref{Def_T2}), see, e.g., \cite[p.~148]{Melsa1978}. Therefore, the Wald test for continuous-time observation processes is defined by
\begin{IEEEeqnarray}{rCL}
\mathsf{T}_{\rm dec}&=&\inf\{t\in\mathbb{R}_+ : \mathsf{S}_t\notin (L_2,L_1)\}  \label{Tau_dec}
\end{IEEEeqnarray}
with the decision output given by
\begin{IEEEeqnarray}{rCL}
\mathsf{D}_{\rm dec}&=&\left\{\begin{array}{ll}
1 & \textrm{if } \mathsf{S}_{\mathsf{T}_{\mathrm{dec}}}\ge L_1\\
2 & \textrm{if } \mathsf{S}_{\mathsf{T}_{\mathrm{dec}}}\le L_2.
\end{array}\right. \label{Def_D_Dec}
\end{IEEEeqnarray}

\subsection{Problem Statement} 
Consider now the black-box decision device as illustrated in Fig.~\ref{IzaakFig1} for which the stochastic observation process $\mathsf{X}_t$ and the algorithm of the decision device are both unknown.    Such a  black-box  decision device is a sequential test $\delta$ for which the function  $(\mathsf{D}, \mathsf{T})$ is unknown.   
We ask now the question: Is it possible to determine whether such a black-box decision device is optimal in the sense of Definition \ref{DefinitionMinimumMeanTime} based on many outcomes $\mathsf{D}$ and $\mathsf{T}$ of the device? 

Having access to the decision outcomes and decision times it is impossible to verify optimality in terms of Definition~\ref{DefinitionMinimumMeanTime}. In this regard consider that the value of the minimum mean decision time is typically unknown since the observed process $\mathsf{X}_t$ and its statistics are often not known.   We thus introduce the  following alternative definition of optimality, which is based on the idea that optimal sequential decision-making needs to exploit the available information optimally.   

\begin{definition}[Optimality in terms of information]\label{DefinitionOptInfUse}
An optimal test $\delta^{*}=(\mathsf{D}^{*},\mathsf{T}^{*})$ minimizes the mutual information $I(\mathsf{H};\mathsf{T}|\mathsf{D})$, i.e.
\begin{IEEEeqnarray}{rCL}
I(\mathsf{H};\mathsf{T}^{*}|\mathsf{D}^{*})&=&\inf_{\delta=(\mathsf{D},\mathsf{T})\in\mathcal{C}(\alpha_1,\alpha_2)}I(\mathsf{H};\mathsf{T}|\mathsf{D}).\label{CondOptimality InfUsageEq}
\end{IEEEeqnarray}
\end{definition}

 Later we will show that for continuous observation processes optimality in the sense of Definition~\ref{DefinitionMinimumMeanTime} implies optimality in the sense of Definition~\ref{DefinitionOptInfUse} but not vise versa.   In this regard, consider that (\ref{CondOptimality InfUsageEq}) is invariant w.r.t.\ time delays $\mathsf{T}_{\rm delay}$ in the decisions, i.e.,  $I(\mathsf{H};\mathsf{T}^{*}|\mathsf{D}^{*})=I(\mathsf{H};\mathsf{T}^{*}+\mathsf{T}_{\rm delay}|\mathsf{D}^{*})$, if $\mathsf{T}_{\rm delay}$ is statistically independent of $\mathsf{H}$ conditioned on $\mathsf{D}$ and $\mathsf{T}$ and if additionally $\mathsf{T}_{\rm delay}$ satisfies that $I(\mathsf{H};\mathsf{T}_{\rm delay}|\mathsf{D}^{*},\mathsf{T}^{*}+\mathsf{T}_{\rm delay})=0$.   Moreover, 
 we will show that for continuous observation processes optimal information usage implies that $I(\mathsf{H};\mathsf{T}^{*}|\mathsf{D}^{*})=0$, because a test achieving $I(\mathsf{H};\mathsf{T}|\mathsf{D})=0$ always exists and $I(\mathsf{H};\mathsf{T}|\mathsf{D})$ is nonnegative. For these reasons Definition~\ref{DefinitionOptInfUse} will allows us to formulate practical tests for optimality of sequential decision-making in black-box decision devices. In general, for the discrete-time setting $I(\mathsf{H};\mathsf{T}^{*}|\mathsf{D}^{*})> 0$, as the information on the hypothesis does not arrive continuously but in chunks, which makes it more difficult to test optimality in discrete-time settings.

\begin{figure*}[t!] 
	\centering  \includegraphics[width=1.05\columnwidth]{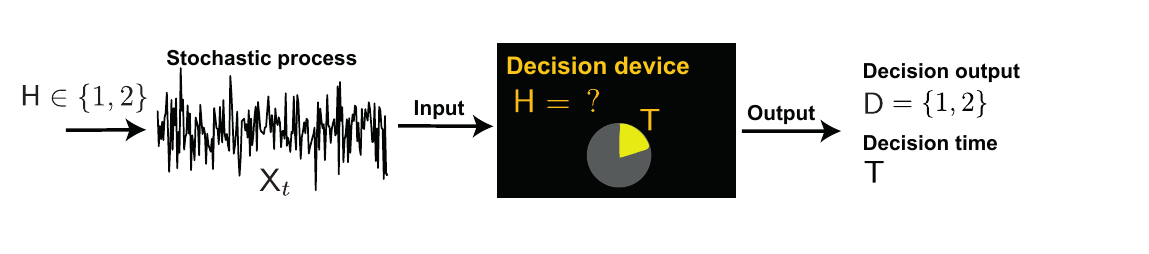}
	\caption{Black-box model of binary sequential decision-making: a decision device observes sequentially a stochastic process $\mathsf{X}_t$ until it takes a decision $\mathsf{D}=1$ ($\mathsf{D}=2$)  corresponding to  the hypothesis $\mathsf{H}=1$ ($\mathsf{H}=2$) at a random decision time~$\mathsf{T}$. }\label{IzaakFig1}
	
\end{figure*}

\section{Optimality Conditions for Continuous Observation Processes} \label{sec:3}
To understand the conditions on optimal sequential decision-making we will derive relations between decision time distributions of optimal binary sequential probability ratio tests.  In this section we consider optimal sequential probability ratio tests for continuous observation processes, which are given by $(\mathsf{T}_{\rm dec},\mathsf{D}_{\rm dec})$ in (\ref{Tau_dec}) - (\ref{Def_D_Dec}). We call these relations decision time fluctuation relations for their reminiscence to stopping time fluctuation relations in non-equilibrium statistical physics, in particular stochastic thermodynamics \cite{Roldan_etal15,Neri_etal17}. In order to derive these relations we use a key property of the exponential of the cumulative log-likelihood ratio $e^{\mathsf{S}_t}$ defined in (\ref{DefLikelihoodCont}), namely that it is a positive and uniformly integrable martingale process with respect to the probability measure $\mathbb{P}_2$ and the filtration generated by the observation process \cite{liptser2001statistics}. An $\mathcal{F}_t$-adapted and integrable  process is called a  martingale w.r.t.\ $\mathcal{F}_t$ and a measure $\mathbb{P}$ if its expected value at time $t$ equals to its value at a previous time $\tilde{t}$, when the expected value is conditioned on observations up to the time $\tilde{t}$. For $e^{\mathsf{S}_t}$, $\mathcal{F}_t$, and $\mathbb{P}_2$ this implies that
\begin{IEEEeqnarray}{rCL}
\mathrm{E}\left[e^{\mathsf{S}_t}\Big|\mathcal{F}_{\tilde{t}}, \mathsf{H}=2\right]&=&e^{\mathsf{S}_{\tilde{t}}}\label{MartingaleProp}
\end{IEEEeqnarray}
$\mathbb{P}_2$-almost surely and with $\tilde{t}<t$. Integrability of $e^{\mathsf{S}_t}$ implies that $\mathrm{E}[e^{\mathsf{S}_t}|\mathsf{H}=2]=1<\infty$.

\subsection{Decision Time Fluctuation Relation for Optimal Decision Devices}\label{Sect_StoppTwoHypo}

\begin{theorem}\label{Theorem1}
We consider a binary sequential hypothesis testing problem with the hypotheses $\mathsf{H}\in\{1,2\}$. Let $\mathbb{P}_1$ and $\mathbb{P}_2$ be two probability measures on the same filtered probability space $(\Omega,\mathcal{F},\{\mathcal{F}_t\}_{t\ge 0})$ corresponding to the hypothesis $\mathsf{H}=1$ and $\mathsf{H}=2$, respectively. We assume that $\{\mathcal{F}_t\}_{t\ge 0}$ is right continuous. We consider that on $\mathcal{F}_\infty =\cap_{t\ge 0}\mathcal{F}_t$, the probability measure $\mathbb{P}_2$ is absolutely continuous with respect to $\mathbb{P}_1$. Furthermore, we consider that the realization of the process $e^{\mathsf{S}_t}=\frac{\mathrm{d}\mathbb{P}_1|_{\mathcal{F}_t}}{\mathrm{d}\mathbb{P}_2|_{\mathcal{F}_t}}$ is $\mathbb{P}_2$ almost surely continuous. Let $\mathsf{T}_{\rm dec}$ and $\mathsf{D}_{\rm dec}$ be as in (\ref{Tau_dec}) and (\ref{Def_D_Dec}) with $\mathrm{E}[\mathsf{T}_{\rm{dec}}|\mathsf{H}=i]<\infty\; (i=1,2)$. We also assume that $\mathsf{T}_{\rm dec}$ has a density function. Under these assumptions the following holds 
\begin{IEEEeqnarray}{rCL}
p_{\mathsf{T}_{\rm dec}}(t|\mathsf{H}=1,\mathsf{D}_{\rm dec}=1)&=&p_{\mathsf{T}_{\rm dec}}(t|\mathsf{H}=2,\mathsf{D}_{\rm dec}=1)\label{CondInvolution_Dec}\\
p_{\mathsf{T}_{\rm dec}}(t|\mathsf{H}=1,\mathsf{D}_{\rm dec}=2)&=&p_{\mathsf{T}_{\rm dec}}(t|\mathsf{H}=2,\mathsf{D}_{\rm dec}=2)\label{CondInvolution_Dec2}
\end{IEEEeqnarray}
where $p_{\mathsf{T}_{\rm dec}}(t|\mathsf{H},\mathsf{D}_{\rm dec})$ is the decision time distribution conditioned on the hypothesis $\mathsf{H}$ and the decision output $\mathsf{D}_{\mathrm{dec}}$.  
\end{theorem}

\begin{proof}
Let 
\begin{IEEEeqnarray}{rCL}
\Phi_1(t)&=&\{\omega\in\Omega:\ \mathsf{T}_{\mathrm{dec}}(\omega)\le t \textrm{ and } \mathsf{D}_{\rm dec}(\omega)=1\}\label{SetPi_cont}
\end{IEEEeqnarray}
be the set of trajectories for which the decision time does not exceed $t$ and the test decides for $\mathsf{D}_{\rm dec}=1$. The probability of the event 
$\Phi_1(t)$ with respect to the measures $\mathbb{P}_1$ or $\mathbb{P}_2$ is equal to the cumulative distribution of the decision time $\mathsf{T}_{\rm dec}$ conditioned on the hypothesis $\mathsf{H}=1$ or $\mathsf{H}=2$, respectively, and conditioned on the decision outcome $\mathsf{D}_{\rm dec}=1$.   We find the following identity between $\mathbb{P}_1(\Phi_{1}(t))$ and  $\mathbb{P}_2(\Phi_{1}(t))$:
\begin{IEEEeqnarray}{rCL}
\mathbb{P}_1(\Phi_{1}(t))
&=&\int_{\omega\in\Phi_{1}(t)}\mathrm{d}\mathbb{P}_1|_{\mathcal{F}_t}\label{Eq19}\\
&=&\int_{\omega\in\Phi_{1}(t)}e^{\mathsf{S}_t}\mathrm{d}\mathbb{P}_2|_{\mathcal{F}_t} \label{ProofFluc1_1}\\
&=&\int_{\omega\in\Phi_{1}(t)}e^{\mathsf{S}_{\mathsf{T}_{\mathrm{dec}}}}\mathrm{d}\mathbb{P}_2|_{\mathcal{F}_t} \label{ProofFluc1_2}\\
&=&e^{L_1}\int_{\omega\in\Phi_{1}(t)}\mathrm{d}\mathbb{P}_2|_{\mathcal{F}_t}\label{Eq21}\\
&=&e^{L_1}\ \mathbb{P}_2(\Phi_{1}(t))\label{Eq23}
\end{IEEEeqnarray}
where for (\ref{ProofFluc1_1}) we have used the Radon-Nikod\'ym theorem and the definition  (\ref{DefLikelihoodCont}). For equality (\ref{ProofFluc1_2}) we have applied \emph{Doob's optional sampling theorem} \cite{liptser2001statistics,williams1991probability}   
to the uniformly integrable $\mathbb{P}_2$-martingale process $e^{\mathsf{S}_t}$. For (\ref{Eq21}) we have used that $e^{\mathsf{S}_t}$ is a continuous process and achieves the value $e^{L_1}$ at time $\mathsf{T}_{\mathrm{dec}}$.

The probability density functions of $\mathsf{T}_{\rm dec}$ can be expressed in terms of the derivatives of the cumulative distributions $\mathbb{P}_{k}(\Phi_1(t))$ ($k=1,2$)
\begin{IEEEeqnarray}{rCL}
p_{\mathsf{T}_{\rm dec}}(t|\mathsf{H}=1,\mathsf{D}_{\rm dec}=1)P(\mathsf{D}_{\rm dec}=1|\mathsf{H}=1)&=&\frac{\mathrm{d}}{\mathrm{d}t}{\mathbb{P}_1(\Phi_1(t))}\label{DistribDef1}\\
p_{\mathsf{T}_{\rm dec}}(t|\mathsf{H}=2,\mathsf{D}_{\rm dec}=1)P(\mathsf{D}_{\rm dec}=1|\mathsf{H}=2)&=&\frac{\mathrm{d}}{\mathrm{d}t}{\mathbb{P}_2(\Phi_1(t))}.\label{DistribDef2}
\end{IEEEeqnarray}
 The ratio of the decision probabilities is
\begin{IEEEeqnarray}{rCL}
\frac{P(\mathsf{D}_{\rm dec}=1|\mathsf{H}=1)}{P(\mathsf{D}_{\rm dec}=1|\mathsf{H}=2)}&=&\frac{1-\alpha_2}{\alpha_1}=e^{L_1} \label{eq:27}
\end{IEEEeqnarray}
which follows from $P(\mathsf{D}_{\rm dec}=1|\mathsf{H}=1) = \lim_{t\to\infty} \mathbb{P}_1(\Phi_{1}(t))$, $P(\mathsf{D}_{\rm dec}=1|\mathsf{H}=2) =$\linebreak $\lim_{t\to\infty} \mathbb{P}_2(\Phi_{1}(t))$, Eq.~(\ref{Eq23}), and from the assumption that the test terminates almost surely. Taking the derivative of the left hand side (LHS) of (\ref{Eq19}) and the right hand side (RHS) of (\ref{Eq23}), and using Eqs.~(\ref{DistribDef1}) to (\ref{eq:27}), we prove Eq.~(\ref{CondInvolution_Dec}). Analogously, Eq.~(\ref{CondInvolution_Dec2}) can be proved.
\end{proof}

\subsection{Decision Time Fluctuation Relation for Optimal Decision Devices with 
 Unknown Hypotheses}
\label{DTFTIH}
In the following, we derive a second fluctuation relation, which we will apply to test optimality of sequential decision-making with less
information than required for Theorem 1 (see Section~\ref{SectOptTestUnknowHyp}), but holds only if the
maximal allowed error probabilities are symmetric, i.e., $\alpha_1=\alpha_2$, and the measures
$\mathbb{P}_1$ and $\mathbb{P}_2$ on $(\Omega, \mathcal{F},
\left\{\mathcal{F}_t\right\}_{t\geq0})$
are related by a measurable involution  $\Theta$.   We consider that
\begin{IEEEeqnarray}{rCL}
\mathbb{P}_2&=&\mathbb{P}_1\circ \Theta\label{InvolutionCond}
\end{IEEEeqnarray}
with $\Theta:\Omega\rightarrow \Omega$  a measurable involution,
i.e., $\Theta$ is invertible with inverse $\Theta^{-1} = \Theta$ and  with $\Theta(\Phi)\in\mathcal{F}$ for all $\Phi\in  \mathcal{F}$.
\begin{theorem}\label{Theorem2}
Under the same conditions as in Theorem~\ref{Theorem1}, with the additional assumption that $\mathbb{P}_2=\mathbb{P}_1\circ \Theta$ with $\Theta$ a measurable involution, and with the additional assumption that the maximal allowed error probabilities fulfill $\alpha_1=\alpha_2$, the following holds 
\begin{IEEEeqnarray}{rCL}
p_{\mathsf{T}_{\rm dec}}(t|\mathsf{H}=1,\mathsf{D}_{\rm dec}=1)&=&p_{\mathsf{T}_{\rm dec}}(t|\mathsf{H}=1,\mathsf{D}_{\rm dec}=2)\label{eq:29}\\
p_{\mathsf{T}_{\rm dec}}(t|\mathsf{H}=2,\mathsf{D}_{\rm dec}=1)&=&p_{\mathsf{T}_{\rm dec}}(t|\mathsf{H}=2,\mathsf{D}_{\rm dec}=2)\label{eq:30}.
\end{IEEEeqnarray}
Furthermore, it holds that 
\begin{IEEEeqnarray}{rCL}
p_{\mathsf{T}_{\rm dec}}(t|\mathsf{D}_{\rm dec}=1) = p_{\mathsf{T}_{\rm dec}}(t|\mathsf{D}_{\rm dec}=2). \label{eq:31}
\end{IEEEeqnarray}
\end{theorem}
The proof of Theorem~\ref{Theorem2} is given in Appendix~\ref{App_ProofTheo2}.

A special case of the result in Theorem~\ref{Theorem2} has been found in the context of nonequilibrium statistical physics \cite{Roldan_etal15, Neri_etal17}: the two hypotheses correspond to a forward and a backward direction of the arrow of time, and $\Theta$  corresponds to the time-reversal operation.   The Radon-Nikod\'{y}m derivative $\mathsf{S}_t$ is then the stochastic entropy production, and the decision time $\mathsf{T}_{\rm dec}$ is its two-boundary first-passage time to cross one of two given symmetric values.       Moreover, in communication theory such a symmetry has been found to show that the probability of cycle slips to the
positive/negative boundary in phase-locked loops used for
synchronization is independent of time \cite[Eq. (74)]{LindseyMeyr77}.

\subsection{Information Theoretic Implications of Optimal Sequential Decision-Making}
Theorem~\ref{Theorem1} and Theorem~\ref{Theorem2} express statistical dependencies of different random quantities involved in optimal sequential decision-making.  Based on Theorem~\ref{Theorem1} we will now show the following.  
\begin{corollary}\label{Corollary3}
Under the same conditions as in Theorem~\ref{Theorem1}, the following equality for mutual information holds 
\begin{IEEEeqnarray}{rCL}
I(\mathsf{H};\mathsf{T}_{\rm dec}|\mathsf{D}_{\rm dec})&=&0
\label{MutInfIndepTermTime_Cont}
\end{IEEEeqnarray}
i.e., $I(\mathsf{H};\mathsf{T}_{\rm dec},\mathsf{D}_{\rm dec})=I(\mathsf{H};\mathsf{D}_{\rm dec})$.
\end{corollary}
\begin{proof}
By the chain rule for mutual information, $I(\mathsf{H};\mathsf{T}_{\rm dec}|\mathsf{D}_{\rm dec})$ can be expressed by
\begin{IEEEeqnarray}{rCL}
I(\mathsf{H};\mathsf{T}_{\rm dec},\mathsf{D}_{\rm dec})&=&I(\mathsf{H};\mathsf{D}_{\rm dec})+I(\mathsf{H};\mathsf{T}_{\rm dec}|\mathsf{D}_{\rm dec}).\label{ChainRuleMutSepCont}
\end{IEEEeqnarray}
The second term on the RHS of (\ref{ChainRuleMutSepCont}) is given by
\begin{IEEEeqnarray}{rCL}
&&I(\mathsf{H};\mathsf{T}_{\rm dec}|\mathsf{D}_{\rm dec})=\mathrm{E}\left[\log_2\left(\frac{p_{\mathsf{T}_{\rm dec}}(\mathsf{T}_{\rm dec}|\mathsf{D}_{\rm dec},\mathsf{H})}{p_{\mathsf{T}_{\rm dec}}(\mathsf{T}_{\rm dec}|\mathsf{D}_{\rm dec})}\right)\right]\\
&&=\mathrm{E}\left[\log_2\left(\frac{p_{\mathsf{T}_{\rm dec}}(\mathsf{T}_{\rm dec}|\mathsf{D}_{\rm dec},\mathsf{H})}{p_{\mathsf{T}_{\rm dec},\mathsf{H}}(\mathsf{T}_{\rm dec},\mathsf{H}=1|\mathsf{D}_{\rm dec})+p_{\mathsf{T}_{\rm dec},\mathsf{H}}(\mathsf{T}_{\rm dec},\mathsf{H}=2|\mathsf{D}_{\rm dec})}\right)\right]\\
&&=\mathrm{E}\left[\log_2\left(\frac{p_{\mathsf{T}_{\rm dec}}(\mathsf{T}_{\rm dec}|\mathsf{D}_{\rm dec},\mathsf{H})}{p_{\mathsf{T}_{\rm dec}}(\mathsf{T}_{\rm dec}|\mathsf{D}_{\rm dec},\mathsf{H}=1)P(\mathsf{H}=1|\mathsf{D}_{\rm dec})+p_{\mathsf{T}_{\rm dec}}(\mathsf{T}_{\rm dec}|\mathsf{D}_{\rm dec},\mathsf{H}=2)P(\mathsf{H}=2|\mathsf{D}_{\rm dec})}\right)\right]\nonumber\\
&&=\mathrm{E}\left[\log_2\left(\frac{p_{\mathsf{T}_{\rm dec}}(\mathsf{T}_{\rm dec}|\mathsf{D}_{\rm dec},\mathsf{H})}{p_{\mathsf{T}_{\rm dec}}(\mathsf{T}_{\rm dec}|\mathsf{D}_{\rm dec},\mathsf{H}=1)P(\mathsf{H}=1|\mathsf{D}_{\rm dec})+p_{\mathsf{T}_{\rm dec}}(\mathsf{T}_{\rm dec}|\mathsf{D}_{\rm dec},\mathsf{H}=1)P(\mathsf{H}=2|\mathsf{D}_{\rm dec})}\right)\right]\nonumber\\
\label{MutInfOptCont_1}\\
&&=\mathrm{E}\left[\log_2\left(\frac{p_{\mathsf{T}_{\rm dec}}(\mathsf{T}_{\rm dec}|\mathsf{D}_{\rm dec},\mathsf{H})}{p_{\mathsf{T}_{\rm dec}}(\mathsf{T}_{\rm dec}|\mathsf{D}_{\rm dec},\mathsf{H}=1)\left(P(\mathsf{H}=1|\mathsf{D}_{\rm dec})+P(\mathsf{H}=2|\mathsf{D}_{\rm dec})\right)}\right)\right]\\
&&=\mathrm{E}\left[\log_2\left(\frac{p_{\mathsf{T}_{\rm dec}}(\mathsf{T}_{\rm dec}|\mathsf{D}_{\rm dec},\mathsf{H})}{p_{\mathsf{T}_{\rm dec}}(\mathsf{T}_{\rm dec}|\mathsf{D}_{\rm dec},\mathsf{H}=1)}\right)\right]\\
&&=P(\mathsf{H}\!=\!1)\mathrm{E}\!\left[\log_2\!\left(\frac{p_{\mathsf{T}_{\rm dec}}(\mathsf{T}_{\rm dec}|\mathsf{D}_{\rm dec},\mathsf{H}=1)}{p_{\mathsf{T}_{\rm dec}}(\mathsf{T}_{\rm dec}|\mathsf{D}_{\rm dec},\mathsf{H}=1)}\right)\!\right]+P(\mathsf{H}\!=\!2)\mathrm{E}\!\left[\log_2\!\left(\frac{p_{\mathsf{T}_{\rm dec}}(\mathsf{T}_{\rm dec}|\mathsf{D}_{\rm dec},\mathsf{H}=2)}{p_{\mathsf{T}_{\rm dec}}(\mathsf{T}_{\rm dec}|\mathsf{D}_{\rm dec},\mathsf{H}=1)}\right)\!\right]\nonumber\\
&&=0\label{MutInfOptCont_2}
\end{IEEEeqnarray}
where for (\ref{MutInfOptCont_1}) and for (\ref{MutInfOptCont_2}) we have used Theorem~\ref{Theorem1}.
\end{proof}
Corollary~\ref{Corollary3} states that in case of optimal sequential decision-making  the decision time $\mathsf{T}_{\rm dec}$ does not give any  additional information on the hypothesis $\mathsf{H}$ beyond the decision outcome $\mathsf{D}_{\rm dec}$.
In this regard, consider that the first term on the RHS of (\ref{ChainRuleMutSepCont}) is  the mutual information the decision outcome of the test $\mathsf{D}_{\textrm{dec}}$ gives  about the actual hypothesis $\mathsf{H}$. The second term on the RHS of (\ref{ChainRuleMutSepCont}) $I(\mathsf{H};\mathsf{T}_{\rm dec}|\mathsf{D}_{\rm dec})$ is the additional information the termination time $\mathsf{T}_{\rm dec}$ gives on the hypothesis $\mathsf{H}$ beyond the information given by the decision $\mathsf{D}_{\rm dec}$. Thus, we have proved that for continuous observation processes optimal sequential decision-making w.r.t.\ Definition~\ref{DefinitionOptInfUse} is achievable and that $I(\mathsf{H};\mathsf{T}^{*}|\mathsf{D}^{*})=0$. 
Note that since sequential probability ratio tests $(\mathsf{D}_{\textrm{dec}},\mathsf{T}_{\textrm{dec}})$ have been shown to be optimal in the sense of Definition~\ref{DefinitionMinimumMeanTime}, Corollary~\ref{Corollary3} implies that optimality in the sense of Definition~\ref{DefinitionMinimumMeanTime} also implies optimality in the sense of Definition~\ref{DefinitionOptInfUse}. 

In case  the assumptions of Theorem~\ref{Theorem2} are satisfied additionally, the following two corollaries hold.

\begin{corollary}\label{Corollary5}
Under the same conditions as in Theorem~\ref{Theorem2}, 
the following equality holds
\begin{IEEEeqnarray}{rCL}
I(\mathsf{D}_{\rm dec};\mathsf{T}_{\rm dec}) &=& 0  .
\end{IEEEeqnarray}
\end{corollary}
\begin{proof}
It holds that 
\begin{IEEEeqnarray}{rCL}
I(\mathsf{D}_{\rm dec};\mathsf{T}_{\rm dec}) &=& \mathrm{E}\left[\log \frac{p_{\mathsf{T}_{\rm dec}}(\mathsf{T}_{\rm dec}|\mathsf{D}_{\rm dec})}{p_{\mathsf{T}_{\rm dec}}(\mathsf{T}_{\rm dec})} \right] 
\\ 
&=& \mathrm{E}\left[\log \frac{p_{\mathsf{T}_{\rm dec}}(\mathsf{T}_{\rm dec}|\mathsf{D}_{\rm dec})}{p_{\mathsf{T}_{\rm dec}}(\mathsf{T}_{\rm dec}|\mathsf{D}_{\rm dec}=1)P\left(\mathsf{D}_{\rm dec}=1\right) +p_{\mathsf{T}_{\rm dec}}(\mathsf{T}_{\rm dec}|\mathsf{D}_{\rm dec}=2)P\left(\mathsf{D}_{\rm dec}=2\right)  }\right] \nonumber \\ 
&=& \mathrm{E}\left[\log \frac{p_{\mathsf{T}_{\rm dec}}(\mathsf{T}_{\rm dec}|\mathsf{D}_{\rm dec})}{p_{\mathsf{T}_{\rm dec}}(\mathsf{T}_{\rm dec}|\mathsf{D}_{\rm dec}=1) }\right]  \label{eq:70x}
\\
&=& 0 \label{eq:71x}
\end{IEEEeqnarray} 
where we have used (\ref{eq:31}) in  (\ref{eq:70x}) and (\ref{eq:71x}).   
\end{proof}

\begin{corollary}\label{Corollary4}
Under the same conditions as in Theorem~\ref{Theorem2}, and with the additional assumption that  $P(\mathsf{H}= 1) = P(\mathsf{H}=2)$,
the following equality holds
\begin{IEEEeqnarray}{rCL}
I(\mathsf{H};\mathsf{T}_{\rm dec})&=&0.\label{MutInfIndepTermTime_Cont2}
\end{IEEEeqnarray}
\end{corollary} 
The proof of Corollary~\ref{Corollary4} is given in Appendix~\ref{App_Corollary4}.

\section{Optimality Conditions for Discrete-Time Observation Processes}\label{sec:4}
In the following, we extend the analysis on optimal information usage in sequential decision-making to the discrete-time setting. In discrete time the optimal test in the sense of Definition~\ref{DefinitionMinimumMeanTime} is given by $\mathsf{T}_{\rm Wald}$ and $\mathsf{D}_{\rm Wald}$ defined in (\ref{Def_T_Wald}) and (\ref{Def_D_Wald}). Extending our results to a discrete-time setting is relevant for  discrete-time systems.   Moreover,  in usual experimental setups  a continuous-time system is sampled yielding a discrete-time representation. The extension from continuous processes to discrete-time processes is not straightforward, as one key characteristic in the continuous-time setting is the fact that the test terminates with a cumulative log-likelihood ratio exactly hitting one of the thresholds. This property of continuous processes does not hold true in the discrete-time setting, where the mean value of the cumulative log-likelihood ratio at the decision time slightly overshoots the thresholds.

The thresholds $L_1$ and $L_2$ depend on the maximum allowed error probabilities $\alpha_1$ and $\alpha_2$, cf.~(\ref{TestReq}). Due to the fact that in the discrete-time setting the trajectory of the accumulated log-likelihood ratios $\mathsf{S}_k$ in (\ref{Cumulated_LLR}) does not necessarily hit one of the thresholds the determination of the optimal thresholds $L_1$ and $L_2$ in terms of $\alpha_1$ and $\alpha_2$ are rather involved, see \cite{Wald1945}. $L_1$ and $L_2$ are chosen such that the  allowed error probabilities given in (\ref{TestReq}) are obeyed with equality.

In the following, we study the statistical dependencies between the hypothesis $\mathsf{H}$, the decision $\mathsf{D}_{\rm Wald}$, and the number of observations $\mathsf{T}_{\rm Wald}$ the sequential probability ratio test given by (\ref{Def_T_Wald}) and (\ref{Def_D_Wald}) uses to make decisions. 

The necessary condition for optimal decision devices given in Theorem~\ref{Theorem1} for the continuous-time setting does not carry over to the discrete-time settings as we will discuss in the following. This can be understood from applying the steps in the proof of Theorem~\ref{Theorem1} in (\ref{Eq19}) to (\ref{Eq23}) to the discrete-time setting. In the discrete-time case with $t\in \mathbb{Z}_+$ the measure $\mathbb{P}_1(\Phi_{1}(t))$ of the discrete-time version of the set $\Phi_{1}(t)$ in (\ref{SetPi_cont}) can be expressed by
\begin{IEEEeqnarray}{rCL}
\mathbb{P}_1(\Phi_{1}(t))
&=&\int_{\omega\in\Phi_{1}(t)}\mathrm{d}\mathbb{P}_1|_{\mathcal{F}_t}\label{Eq19_Disc}\\
&=&\int_{\omega\in\Phi_{1}(t)}e^{\mathsf{S}_t}\mathrm{d}\mathbb{P}_2|_{\mathcal{F}_t} \label{ProofFluc1_1_Disc}\\
&=&\int_{\omega\in\Phi_{1}(t)}e^{\mathsf{S}_{\mathsf{T}_{\mathrm{Wald}}}}\mathrm{d}\mathbb{P}_2|_{\mathcal{F}_t} \label{ProofFluc1_2_Disc}\\
&=&\frac{\int_{\omega\in\Phi_{1}(t)}e^{\mathsf{S}_{\mathsf{T}_{\mathrm{Wald}}}}\mathrm{d}\mathbb{P}_2|_{\mathcal{F}_t}}{\int_{\omega\in\Phi_{1}(t)}\mathrm{d}\mathbb{P}_2|_{\mathcal{F}_t}}\int_{\omega\in\Phi_{1}(t)}\mathrm{d}\mathbb{P}_2|_{\mathcal{F}_t}\label{ProofFluc1_2_Disc}\\
&=&\mathrm{E}\left[e^{\mathsf{S}_{\mathsf{T}_{\mathrm{Wald}}}}|\mathsf{H}=2, \omega\in\Phi_1(t)\right]\ \mathbb{P}_2(\Phi_{1}(t))\\
&=&e^{L_1}\mathrm{E}\left[e^{\mathsf{M}_1}|\mathsf{H}=2, \omega\in\Phi_1(t)\right]\ \mathbb{P}_2(\Phi_{1}(t))\label{Eq23_Disc}
\end{IEEEeqnarray}
where $\mathsf{M}_1>0$ in (\ref{Eq23_Disc}) is the overshoot beyond the threshold $L_1$. Since in general the distribution of the overshoot $\mathsf{M}_1 = \mathsf{S}_{\mathsf{T}_{\rm Wald}}  - L_1$ depends on the time $t$ the fluctuation relations (\ref{CondInvolution_Dec}) and (\ref{CondInvolution_Dec2}) do not extend to the discrete-time case. 

Translating (\ref{DistribDef1}) and (\ref{DistribDef2}) to the discrete-time setting yields
\begin{IEEEeqnarray}{rCL}
p_{\mathsf{T}_{\rm dec}}(t|\mathsf{H}=1,\mathsf{D}_{\rm dec}=1)P(\mathsf{D}_{\rm dec}=1|\mathsf{H}=1)&=&\mathbb{P}_1(\Phi_1(t+1))-\mathbb{P}_1(\Phi_1(t))\label{DistribDef1_disc}\\
p_{\mathsf{T}_{\rm dec}}(t|\mathsf{H}=2,\mathsf{D}_{\rm dec}=1)P(\mathsf{D}_{\rm dec}=1|\mathsf{H}=2)&=&\mathbb{P}_2(\Phi_1(t+1))-\mathbb{P}_2(\Phi_1(t))\label{DistribDef2_disc}
\end{IEEEeqnarray}
for $t\ge 1$. 
 
Taking the difference between the values of $\mathbb{P}_1(\Phi_{1}(t))$ at two consecutive time instants we get
\begin{IEEEeqnarray}{rCL}
\mathbb{P}_1(\Phi_{1}(t+1))-\mathbb{P}_1(\Phi_{1}(t))
&=&e^{L_1}\mathrm{E}\left[e^{\mathsf{M}_1}|\mathsf{H}=2, \omega\in\Phi_1(t+1)\right]\ \mathbb{P}_2(\Phi_{1}(t+1))
\nonumber\\
&&
-e^{L_1}\mathrm{E}\left[e^{\mathsf{M}_1}|\mathsf{H}=2, \omega\in\Phi_1(t)\right]\ \mathbb{P}_2(\Phi_{1}(t)).
\end{IEEEeqnarray}
In case $\mathrm{E}\left[e^{\mathsf{M}_1}|\mathsf{H}=2,\omega\in \Phi_{1}(t)\right]=K$ is time independent, we get from (\ref{DistribDef1_disc}) and (\ref{DistribDef2_disc}) 
\begin{IEEEeqnarray}{rCL}
\frac{P(\mathsf{D}_{\rm Wald}=1|\mathsf{H}=1)}{P(\mathsf{D}_{\rm Wald}=1|\mathsf{H}=2)}&=&\frac{1-\alpha_2}{\alpha_1}=e^{L_1} K\label{eq:27_disc}
\end{IEEEeqnarray}
where we have used the assumption that the test terminates almost surely, and we get the fluctuation relations corresponding to Theorem~\ref{Theorem1} for decision times $\mathsf{T}_{\rm Wald}$ in the discrete-time case.

The constraint that $\mathrm{E}\left[e^{\mathsf{M}_1}|\mathsf{H}=2, \omega\in\Phi_1(t)\right]$ is time independent is approximately fulfilled in case the size of the thresholds $L_1$ and $L_2$ is large in comparison to the average increase of the log-likelihood ratio $\Delta_n$ per observation sample, see (\ref{Cumulated_LLR_noniid}). This can be seen as taking the continuum limit of the decision making process. In this regard, consider that the distribution of the overshoot $\mathsf{M}_1$ is time independent if the distribution of the distance $L_{\mathsf{D}_{\textrm{Wald}}}-\mathsf{S}_{T_{\textrm{Wald}}-1}$, at the time instant before a decision is taken, is time independent, and if the distribution of the increment $\Delta_{T_{\textrm{Wald}}}$ is independent of time. The distribution of $L_{\mathsf{D}_{\textrm{Wald}}}-\mathsf{S}_{T_{\textrm{Wald}}-1}$ is time independent if the initial value of the cumulative log-likelihood has no significant influence anymore on the distribution of $\mathsf{S}_{\mathsf{T}_{\textrm{Wald}}-1}$ when conditioning on termination at time instant $\mathsf{T}_{\textrm{Wald}}$. This is satisfied in case $\mathsf{T}_{\textrm{Wald}}$ is sufficiently large, which holds if the thresholds $L_1$ and $L_2$ are large in comparison to the average of the increments of the log-likelihood ratio $\Delta_n$. This is illustrated for an example based on numerical simulations in Section~\ref{SectOvershot}. 

In the following, we assume that the condition 
\begin{IEEEeqnarray}{rCL}
\mathrm{E}\left[e^{\mathsf{M}_1}|\mathsf{H}=2, \omega\in\Phi_1(t)\right]&=&K\label{ConditionDiscreteTime}
\end{IEEEeqnarray}
is fulfilled. For many practical applications this condition is approximately fulfilled, see the numerical experiments in Section~\ref{SectNumericalExperminents}.

The results on optimal information usage carry over from continuous time to discrete time given that (\ref{ConditionDiscreteTime}) holds.

\begin{theorem}\label{Theorem3Disc}
We consider a binary sequential hypothesis testing problem with the hypotheses $\mathsf{H}\in\{1,2\}$. Let $\{p_{\mathsf{X}_1^{k+1}}\left(\cdot|\mathsf{H}=1\right)\}$ and $\{p_{\mathsf{X}_1^{k+1}}\left(\cdot|\mathsf{H}=2\right)\}$ be two sequences of probability density functions of the sequence of real valued observations $\{\mathsf{X}_1,\mathsf{X}_2,\hdots\}$ in case hypothesis $\mathsf{H}=1$ and $\mathsf{H}=2$ are true, respectively, and with $k\in\mathbb{Z}_+$. Let $\mathsf{T}_{\rm Wald}$ and $\mathsf{D}_{\rm Wald}$ be as in (\ref{Def_T_Wald}) and (\ref{Def_D_Wald}) with $\mathrm{E}[\mathsf{T}_{\rm Wald}|\mathsf{H}=i]<\infty\; (i=1,2)$. Under these assumptions and the assumption that (\ref{ConditionDiscreteTime}) is fulfilled it holds that
\begin{IEEEeqnarray}{rCL}
P(\mathsf{T}_{\rm Wald}=k|\mathsf{H}=2,\mathsf{D}_{\rm Wald}=1)&=&P(\mathsf{T}_{\rm Wald}=k|\mathsf{H}=1,\mathsf{D}_{\rm Wald}=1)\label{CondInvolution_Dec_tilde_discrete1}\\
P(\mathsf{T}_{\rm Wald}=k|\mathsf{H}=2,\mathsf{D}_{\rm Wald}=2)&=&P(\mathsf{T}_{\rm Wald}=k|\mathsf{H}=1,\mathsf{D}_{\rm Wald}=2)\label{CondInvolution_Dec_tilde_discrete2}
\end{IEEEeqnarray}
for all $k\in\mathbb{Z}_+$.
\end{theorem}

Theorem~\ref{Theorem3Disc} implies optimal usage of information with respect to Definition~\ref{DefinitionOptInfUse} for the discrete-time setting yielding the following corollary.

\begin{corollary}\label{Corollary4_+}
Under the same conditions as in Theorem~\ref{Theorem3Disc}, the following equality for mutual information holds 
\begin{IEEEeqnarray}{rCL}
I(\mathsf{H};\mathsf{T}_{\rm Wald},\mathsf{D}_{\rm Wald})&=&I(\mathsf{H};\mathsf{D}_{\rm Wald})\label{MutInfIndepTermTime_Disc}
\end{IEEEeqnarray}
implying that 
\begin{IEEEeqnarray}{rCL}
I(\mathsf{H};\mathsf{T}_{\rm Wald}|\mathsf{D}_{\rm Wald})&=&0. \label{testwald}
\end{IEEEeqnarray}
\end{corollary}
The proof follows along the same line as the proof of Corollary~\ref{Corollary3}, but this time based on Theorem~\ref{Theorem3Disc}.

Analogously, Theorem~\ref{Theorem2} carries over to the discrete-time case.

\begin{theorem}\label{Theorem4Disc}
Under the same conditions as in Theorem~\ref{Theorem3Disc}, with the additional assumption that $p_{\mathsf{X}_1^{k+1}}\left(\cdot|\mathsf{H}=2\right)=p_{\mathsf{X}^{k+1}_{1}\circ\Theta}\left(\cdot|\mathsf{H}=1\right)$ for $k\in\mathbb{Z}_+$, where $\Theta$ is a measurable involution, and the  additional assumption that the maximal allowed error probabilities fulfill $\alpha_1=\alpha_2$, the following holds 
\begin{IEEEeqnarray}{rCL}
P(\mathsf{T}_{\rm Wald}=k|\mathsf{H}=1,\mathsf{D}_{\rm Wald}=1)&=&P(\mathsf{T}_{\rm Wald}=k|\mathsf{H}=1,\mathsf{D}_{\rm Wald}=2)\label{CondInvolution_Dec_tilde_discrete3}\\
P(\mathsf{T}_{\rm Wald}=k|\mathsf{H}=2,\mathsf{D}_{\rm Wald}=1)&=&P(\mathsf{T}_{\rm Wald}=k|\mathsf{H}=2,\mathsf{D}_{\rm Wald}=2)\label{CondInvolution_Dec_tilde_discrete4}
\end{IEEEeqnarray}
for all $k\in\mathbb{Z}_+$. Furthermore, it holds that 
\begin{IEEEeqnarray}{rCL}
P(\mathsf{T}_{\rm Wald}=k|\mathsf{D}_{\rm Wald}=1) = P(\mathsf{T}_{\rm Wald}=k|\mathsf{D}_{\rm Wald}=2) \quad \textrm{for all }k\in\mathbb{Z}_+. 
\end{IEEEeqnarray}
\end{theorem}

Theorem~\ref{Theorem4Disc} can be proved by carrying over the proof of Theorem~\ref{Theorem2} to the discrete-time case and additionally using a modification of the application of Doob's optional sampling theorem similar to (\ref{ProofFluc1_2_Disc}) to (\ref{Eq23_Disc}) leading to the additional assumption that (\ref{ConditionDiscreteTime}) is fulfilled.

A special case of Theorem~\ref{Theorem4Disc} was shown in \cite{dorpinghaus2017information} for the case of i.i.d.\ observation processes and low error probabilities $\alpha_1=\alpha_2$.

Using Theorem~\ref{Theorem4Disc} also Corollary~\ref{Corollary5} and Corollary~\ref{Corollary4} carry over to the discrete-time case.

\begin{corollary}\label{Corollary5_discrete}
Under the same conditions as in Theorem~\ref{Theorem4Disc}, 
the following equality holds
\begin{IEEEeqnarray}{rCL}
I(\mathsf{D}_{\rm Wald};\mathsf{T}_{\rm Wald}) &=& 0  .
\end{IEEEeqnarray}
\end{corollary}

\begin{corollary}\label{Corollary4_discrete}
Under the same conditions as in Theorem~\ref{Theorem4Disc}, and with the additional assumption that  $P(\mathsf{H}= 1) = P(\mathsf{H}=2)$,
the following equality holds
\begin{IEEEeqnarray}{rCL}
I(\mathsf{H};\mathsf{T}_{\rm Wald})&=&0.\label{MutInfIndepTermTime_Cont2_discrete}
\end{IEEEeqnarray}
\end{corollary} 

The proofs of Corollary~\ref{Corollary5_discrete} and Corollary~\ref{Corollary4_discrete} follow along the lines of the proves of Corollary~\ref{Corollary5} and Corollary~\ref{Corollary4}.

\section{Tests for Optimality of Sequential Decision-Making}\label{SectionTests}
For the case of continuous observation processes Theorem~\ref{Theorem1} and Corollary~\ref{Corollary3} hold for binary sequential probability ratio tests which are optimal in the sense of Definition~\ref{DefinitionMinimumMeanTime}. In case of additional symmetry conditions, Theorem~\ref{Theorem2}, Corollary~\ref{Corollary5}, and Corollary~\ref{Corollary4} hold as well.   Under the reasonable assumption that the  joint statistics of $(\mathsf{H}, \mathsf{T}, \mathsf{D})$ have a unique solution over all tests fulfilling~(\ref{OptDecisionTime}),  these theorems and corollaries are necessary conditions for optimal sequential decision-making in the sense of Definition~\ref{DefinitionMinimumMeanTime}. Likewise, for discrete-time observation processes which fulfill the condition given by (\ref{ConditionDiscreteTime}) Theorem~\ref{Theorem3Disc} and Corollary~\ref{Corollary4_+} give necessary conditions for optimal sequential decision-making. In case additional symmetry conditions are fulfilled also Theorem~\ref{Theorem4Disc}, Corollary~\ref{Corollary5_discrete}, and Corollary~\ref{Corollary4_discrete} hold. Based on these theorems and corollaries we formulate tests to test optimality of sequential decision-making in black-box decision devices and present algorithms to measure the distance to optimality of the decision process in the black-box decision device.

\subsection{Continuous-Observation Processes}

\subsubsection{Testing Optimality and Measuring the Distance to Optimality in Case of Known Hypotheses}\label{sec:test1}
We sample $m$  independent realizations  $\left\{h_i, t_i, d_i\right\}_{i=1,\ldots, m}$  of the joint random variables $(\mathsf{H}, \mathsf{T}, \mathsf{D})$ given by  subsequent decisions, where $\mathsf{H}$ corresponds to the random variable describing the actual hypothesis and $\mathsf{T}$  and $\mathsf{D}$ are the outputs, decision time and decision variable, of the black-box decision device. In case the observation window of the experimentalist is not sufficiently large such that for certain samples the black box has not decided yet, the experimentalist can discard those samples. 

We first state a statistical test which can reject, with a certain statistical significance, the null hypothesis that the given black-box device is optimal in the sense of Definition~\ref{DefinitionOptInfUse} and, thus, also in the sense of Definition~\ref{DefinitionMinimumMeanTime}. We create from  the whole set of realizations four subsets of decision times $\mathcal{A}_{r,s} = \left\{t_{i},  i\in \left\{1,2,...,m\right\}: (h_i, d_i) = (r,s)\right\}$ with  $(r,s) \in \left\{(1,1), (1,2), (2,1), (2,2)\right\}$.    Under the null hypothesis, Theorem \ref{Theorem1} implies that the subsets $\mathcal{A}_{1,1}$ and $\mathcal{A}_{2,1}$ contain  independent realizations of decision times from  the same distribution and, analogously,   $\mathcal{A}_{1,2}$ and $\mathcal{A}_{2,2}$ contain  independent realizations of decision times from  the same distribution.  Whether two sets of independent realizations are sampled from the same continuous distribution can be tested with a certain significance using the 
 two-sample Kolmogorov-Smirnov test \cite[pp.~663-665]{degroot2012probability}.     Note that we do not require knowledge on the statistics $\mathbb{P}_1$ and $\mathbb{P}_2$ of the observation process $\mathsf{X}_t$, which makes our test for optimality of decision devices very useful for practical applications where in many situations such statistics are unknown.   Notice that because the observation process  $\mathsf{X}_t$  may take values in  a  high-dimensional space  it can be difficult to get a good estimate of its  statistics.

A quantity for the distance of the sequential decision-making process of the black-box decision device with respect to the optimal sequential decision process in the sense of Definition~\ref{DefinitionOptInfUse} is given by the empirical estimate $\hat{I}(\mathsf{H};\mathsf{T}|\mathsf{D})$  of the mutual information~$I(\mathsf{H};\mathsf{T}|\mathsf{D})$. The estimate $\hat{I}(\mathsf{H};\mathsf{T}|\mathsf{D})$ can be gained from empirical estimates of entropy and differential entropy, see \cite{beirlant1997nonparametric} and  \cite{jiao2015minimax}.   
Note that with the chain rule for mutual information it holds that 
\begin{IEEEeqnarray}{rCL}
I(\mathsf{H};\mathsf{T}|\mathsf{D}) = I(\mathsf{H};\mathsf{D},\mathsf{T}) - I(\mathsf{H};\mathsf{D}) .  \label{eq:chain}
\end{IEEEeqnarray}
The first term on the RHS of (\ref{eq:chain}) is the complete mutual information that  the output of the black-box decision device, $(\mathsf{D}, \mathsf{T})$, gives on the hypothesis $\mathsf{H}$.   The second term on the RHS of (\ref{eq:chain}) is the mutual information between the decision  $\mathsf{D}$ and the hypothesis $\mathsf{H}$,  which   in case of optimal sequential decision-making   equals the complete mutual information $I(\mathsf{H};\mathsf{D},\mathsf{T}) $.   Hence, with (\ref{eq:chain})  $I(\mathsf{H};\mathsf{T}|\mathsf{D}) $ measures the  information the black-box device discards in case of non-optimal decision-making.    Therefore, $\hat{I}(\mathsf{H};\mathsf{T}|\mathsf{D})$ provides a measure for how much the decision statistics of a certain black-box device diverge from the optimal solution, or less formally stated, how close to optimality a decision device behaves.

\subsubsection{Testing Optimality in Case of Unknown Hypotheses}\label{SectOptTestUnknowHyp} 
If the statistics of the observation process $\mathsf{X}_t$ fulfill the  
involution condition (\ref{InvolutionCond}) and if the constraints on the error probabilities $\alpha_1$ and $\alpha_2$ are equal implying symmetric thresholds $L_1 = -L_2 = L$, then based on Theorem~\ref{Theorem2} we formulate a test which is able to reject optimality in the sense of Definition~\ref{DefinitionOptInfUse} and, thus, also in the sense of Definition~\ref{DefinitionMinimumMeanTime}. Different to the test formulated in Section~\ref{sec:test1} we do not require knowledge of the actual realizations of the hypothesis $\mathsf{H}$, which in certain situations is important for practical application.   

We sample $m$ independent realizations $\left\{ t_i, d_i\right\}_{i=1,\ldots, m}$  of the joint random variables $( \mathsf{T}, \mathsf{D})$,
where  $\mathsf{T}$  and $\mathsf{D}$ are the outputs, decision variable and decision time, of the black-box decision device. As in the case of known hypotheses samples for which the black-box decision device has not terminated yet can be discarded. Under the above assumptions, we  state a statistical test which can reject, with a certain statistical significance, the null hypothesis that the given black-box device is optimal.    We create from  the whole set of realizations   two subsets of decision times $\mathcal{A}_{r} = \left\{t_{i},  i\in \left\{1,2,...,m\right\}: d_i = r\right\}$ with  $r\in \left\{1,2\right\}$.    Under the current null hypothesis, Theorem~\ref{Theorem2} implies that the subsets $\mathcal{A}_{1}$ and $\mathcal{A}_{2}$ contain  independent realizations of decision times from  the same distribution.  We can again use a 
 two-sample Kolmogorov-Smirnov test  \cite[pp.~663-665]{degroot2012probability} to reject the null hypothesis with a certain significance.      Corollaries \ref{Corollary5} and \ref{Corollary4} provide alternative means to test optimality of sequential decision-making.   

Quantifying the degree of optimality using the mutual informations $I(\mathsf{D};\mathsf{T})$ or 
$I(\mathsf{H};\mathsf{T})$ provides in general no clear interpretation.   Hence, in order to quantify the divergence  of the black-box device from optimal sequential decision-making we can use $I(\mathsf{H};\mathsf{T}|\mathsf{D})$ based on Corollary~\ref{Corollary3}.

\subsection{Discrete-Time Observation Processes}\label{SectionTests_Disc}
Analogously to the case of continuous observation processes we formulate tests for optimality of sequential decision-making in discrete time and we also present algorithms to measure the distance to optimality of black-box decision devices. We use Theorems~\ref{Theorem3Disc}, Theorem~\ref{Theorem4Disc} and Corollary~\ref{Corollary4_+}. 

\subsubsection{Testing Optimality and Measuring the Distance to 
Optimality in Case of Known Hypotheses}\label{sec:test1_disc}
We sample $m$  independent realizations  $\left\{h_i, t_i, d_i\right\}_{i=1,\ldots, m}$  of the joint random variables $(\mathsf{H}, \mathsf{T}, \mathsf{D})$ given by  subsequent decisions, where $\mathsf{H}$ corresponds to the random variable describing the actual hypothesis and $\mathsf{T}$  and $\mathsf{D}$ are the outputs, decision time and the decision variable, of the black-box decision device. As before we discard samples for which the black-box decision device has not decided yet.

The algorithm to test optimality of a black-box decision device is analogous to the case of continuous observation processes. We construct from the whole set of realizations $\left\{h_i, t_i, d_i\right\}_{i=1,\ldots, m}$ four sets of decision times $\mathcal{A}_{r,s} = \left\{t_{i},  i\in \left\{1,2,...,m\right\}: (h_i, d_i) = (r,s)\right\}$ with\linebreak $(r,s) \in \left\{(1,1), (1,2), (2,1), (2,2)\right\}$. Theorem~\ref{Theorem3Disc} implies that the subsets $\mathcal{A}_{1,1}$ and $\mathcal{A}_{2,1}$ contain  independent realizations of decision times from  the same distribution and, analogously,   $\mathcal{A}_{1,2}$ and $\mathcal{A}_{2,2}$ contain  independent realizations of decision times from the same distribution. Whether two sets of independent realizations are sampled from the same discrete distribution can be tested with a certain significance using the two-sample $\chi^{2}$-test \cite[p.~253, Problem~3]{vanderVaart1998}.
 
A quantity for the distance to optimality of a black-box decision device with respect to the optimal sequential decision process in the sense of Definition~\ref{DefinitionOptInfUse} is in discrete time given by the empirical estimate $\hat{I}(\mathsf{H};\mathsf{T}|\mathsf{D})$  of the mutual information~$I(\mathsf{H};\mathsf{T}|\mathsf{D})$, see , cf.\ Corollary~\ref{Corollary4_+}.

\subsubsection{Testing Optimality in Case of Unknown Hypotheses}\label{SectOptTestUnknowHyp} 
As in the continuous case testing optimality of a black-box decision device can be done even in case of unknown hypothesis in case certain additional conditions are fulfilled. Namely, the statistics of the observation process $\mathsf{X}_t$ have to fulfill the involution condition (\ref{InvolutionCond}) and the constraints on the error probabilities $\alpha_1$ and $\alpha_2$ have to be equal, implying symmetric thresholds $L_1 = -L_2 = L$. Then based on Theorem~\ref{Theorem4Disc} we can formulate the following test. We sample $m$ independent realizations $\left\{ t_i, d_i\right\}_{i=1,\ldots, m}$  of the joint random variables $( \mathsf{T}, \mathsf{D})$, where  $\mathsf{T}$  and $\mathsf{D}$ are the decision time and decision variable of the black-box decision device. We create from  the whole set of realizations   two subsets of decision times $\mathcal{A}_{k} = \left\{t_{i},  i\in \left\{1,2,...,m\right\}: d_i = k\right\}$ with  $k\in \left\{1,2\right\}$. Theorem \ref{Theorem4Disc} implies that the subsets $\mathcal{A}_{1}$ and $\mathcal{A}_{2}$ contain  independent realizations of decision times from  the same distribution in case the black-box decision device is optimal.  We can again use a two-sample $\chi^2$-test \cite[p.~253, Problem~3]{vanderVaart1998} to test whether the two subsets $\mathcal{A}_1$ and $\mathcal{A}_2$ are sampled from the same distribution.  Corollary~\ref{Corollary5_discrete} and Corollary~\ref {Corollary4_discrete} provide alternatives to test optimality of sequential decision-making.

\section{Testing Optimality in Numerical Experiments}\label{SectNumericalExperminents} 
In this section we apply our algorithms to test optimality of binary sequential decision-making of black-box devices and to measure the degree of divergence from optimality.    We consider a class of  decision devices which for certain parameter values are optimal, and we verify whether our algorithms are able to detect the parameters for which the decision devices are optimal.    In this section we distinguish again continuous and discrete  observation processes.   However, here we will start with discrete-time processes which allow for simpler numerical study.  

To distinguish theoretical quantities from empirical estimates we denote by $\hat{P}$ and $\hat{I}$ the empirical estimates of the probabilities $P$ and mutual informations $I$.    Furthermore, we write $\hat{\mathrm{E} }\left[\cdot  \right]$ for the  empirical estimate of the expectation $\mathrm{E} \left[\cdot\right]$.   

\subsection{Discrete-time observation processes}
\subsubsection{Testing optimality in case of known hypotheses} \label{sec:5a}
We consider an observation sequence  $\mathsf{X} = (\mathsf{X}_1, \mathsf{X}_2, \ldots, \mathsf{X}_k)$   where the $\mathsf{X}_n$  (where $n\in[1,k]$) are i.i.d.~random variables drawn from one of two possible probability distributions corresponding to the two hypotheses $\mathsf{H} \in\left\{ 1,2\right\}$, i.e., 
\begin{IEEEeqnarray}{rCL}
p_{\mathsf{X}^k_1}\left(x^k_1|\mathsf{H}\right) = \prod^k_{n=1}p_{\mathsf{X}}\left(x_n|\mathsf{H}\right). \label{eq:obsModel} 
\end{IEEEeqnarray}
   In our example the densities $p_{\mathsf{X}}(\cdot|\mathsf{H})$ are Gaussian with mean $\mu_i$ and variance $\sigma^2_i$ with $i\in\left\{1,2\right\}$ corresponding to the two hypotheses $\mathsf{H}=i$.  In the special case where $\mu_1 = -\mu_2$ and $\sigma_1 = \sigma_2$ the involution property  (\ref{InvolutionCond}) holds.  

We consider a class of decision models representing the black-box decision devices.   
These decision models use the Wald sequential probability ratio test based on a model of the external world which  may be incorrect.   Each decision model computes the cumulative log-likelihood ratio of two Gaussian distributions with mean $\tilde{\mu}_i$ and variance  $\tilde{\sigma}^2_i$ (with $i\in\left\{1,2\right\}$), i.e.,
\begin{IEEEeqnarray}{rCL}
\tilde{\mathsf{S}}_k &=& \sum^k_{n=1} \log \left(\frac{p_{\mathsf{X}}\left(\mathsf{X}_n|\mathsf{H}=1\right)}{p_{\mathsf{X}}\left(\mathsf{X}_n|\mathsf{H}=2\right)}\right)  \nonumber
\\ 
&=& k \log \frac{\tilde{\sigma}_2}{\tilde{\sigma}_1} +  \sum^k_{n=1}\left(\frac{(\mathsf{X}_n-\tilde{\mu}_2)^2}{2\tilde{\sigma}^2_2} -\frac{(\mathsf{X}_n-\tilde{\mu}_1)^2}{2\tilde{\sigma}^2_1}  \right). \label{eq:modelS}
\end{IEEEeqnarray}
The decision time  of the  model is 
\begin{IEEEeqnarray}{rCL}
\mathsf{T}&=&\min\{k\in\mathbb{N} : \tilde{\mathsf{S}}_k\notin (L_2,L_1)\}\label{Def_T_Waldxx}  
\end{IEEEeqnarray}
and the decision variable is given by
\begin{IEEEeqnarray}{rCL}
\mathsf{D}&=&\left\{\begin{array}{ll}
1 & \textrm{if } \tilde{\mathsf{S}}_{\mathsf{T}}\ge L_1\\
2 & \textrm{if } \tilde{\mathsf{S}}_{\mathsf{T}}\le L_2
\end{array}\right. \label{Def_D_Waldxx}
\end{IEEEeqnarray}
with  the two thresholds $L_1>0$ and $L_2<0$.  
\begin{figure}
\subfigure[Optimal decision-making: black-box decision device uses the correct model of the external world]
{\includegraphics[width=0.48\columnwidth]{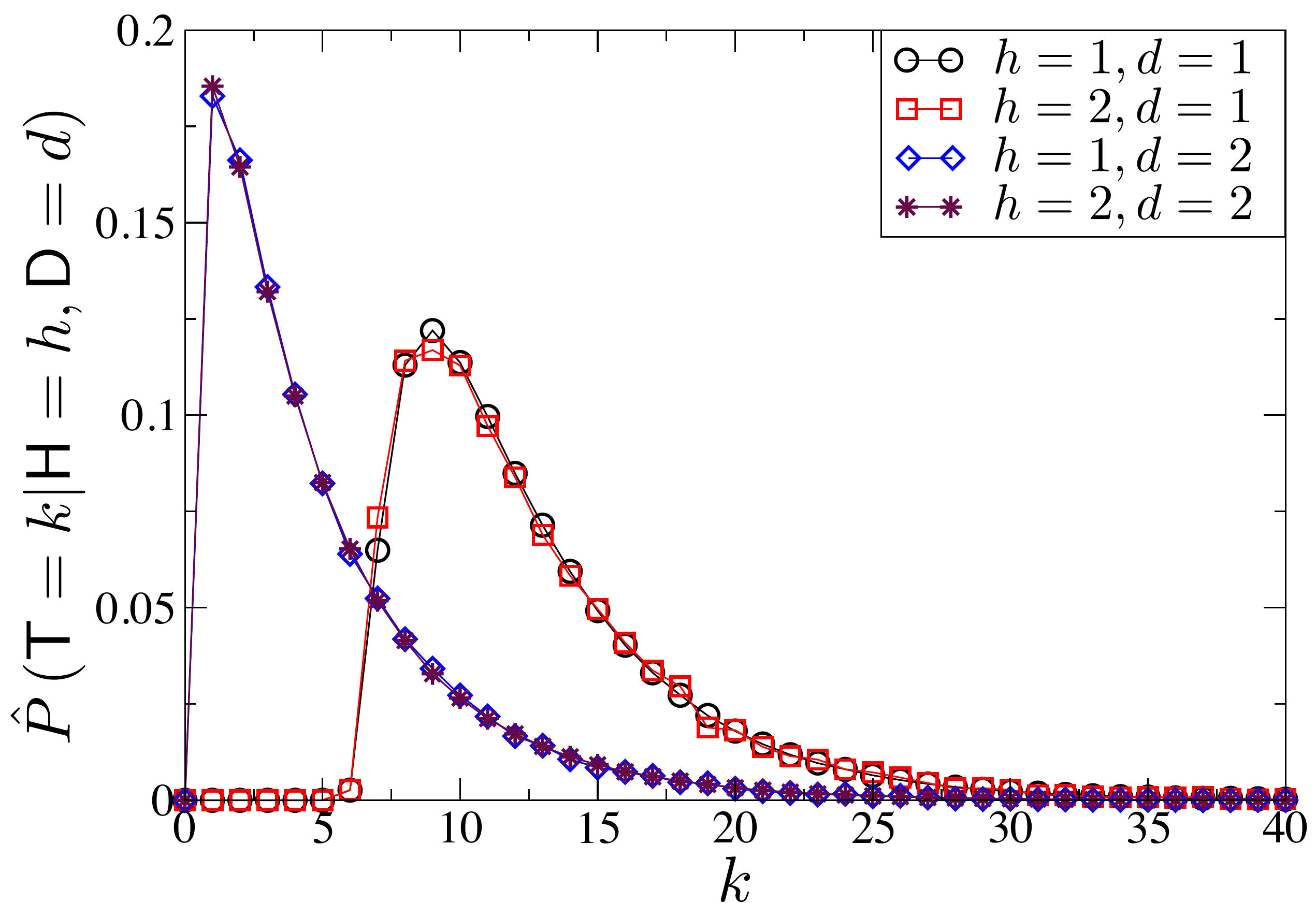}\label{fig3a}} {\hspace{+0.2cm}}
\subfigure[Sub-optimal decision-making: black-box decision device uses the wrong model of the external world]
{\includegraphics[width=0.48\columnwidth]{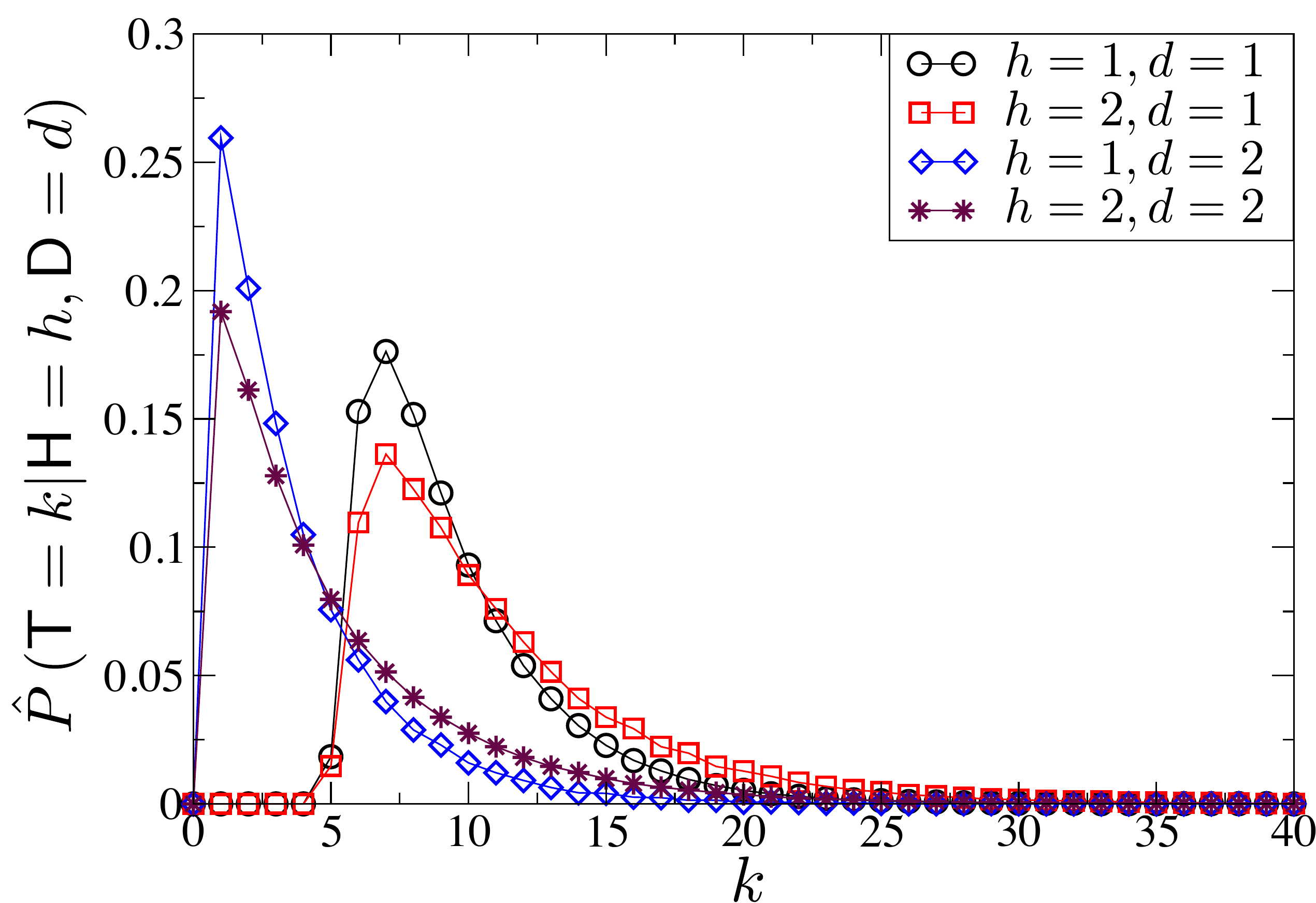}\label{fig3b}}
\caption{Illustration of Theorem~\ref{Theorem3Disc} using the observation model given by (\ref{eq:obsModel}) and the decision model given by (\ref{eq:modelS}) to (\ref{Def_D_Waldxx}). The distributions of the observation process corresponding to the two hypotheses  are Gaussian with parameters  $(\mu_1, \sigma_1) = (0,5)$ and $(\mu_2, \sigma_2) = (1,10)$, respectively.    The decision model has threshold values  $L_1 = 4$ and $L_2 = -2$ and parameters:  (a)  $(\tilde{\mu}_1, \tilde{\sigma}_1) = (\mu_1, \sigma_1) $ and $(\tilde{\mu}_2, \tilde{\sigma}_2) = (\mu_2, \sigma_2)$, (b)   $(\tilde{\mu}_1, \tilde{\sigma}_1) = (\mu_1, \sigma_1) $ and $(\tilde{\mu}_2, \tilde{\sigma}_2) = (5, \sigma_2)$.   The empirical error probabilities are in (a) $\alpha_1 = 0.041$ and $\alpha_2 = 0.0133$ and (b) $\alpha_1 = 0.0335$ and $\alpha_2 = 0.0506$.   Distributions are estimated using $1e+6$ simulation runs.  }\label{fig3}
\end{figure}   
If $\tilde{\mu}_i = \mu_i$ and $\tilde{\sigma}_i = \sigma_i$ then the black-box device 
uses the correct model of the external world and makes optimal sequential decisions in the sense of minimizing the decision time   (see Definition~\ref{DefinitionMinimumMeanTime}) since $\mathsf{T} = \mathsf{T}_{\rm Wald}$ and $\mathsf{D} = \mathsf{D}_{\rm Wald}$.      Corollary~\ref{Corollary4_+} implies that for these parameter values the black-box device makes also optimal sequential decisions in the sense of information usage  (see Definition~\ref{DefinitionOptInfUse}).     If additionally $L_1=-L_2$, then $\alpha_1 = \alpha_2 = \alpha$. 

\begin{figure}
\subfigure[Two sample $\chi^2$-test of  (\ref{CondInvolution_Dec_tilde_discrete1}) for $\mathsf{D}=1$]
{\includegraphics[width=0.48\columnwidth]{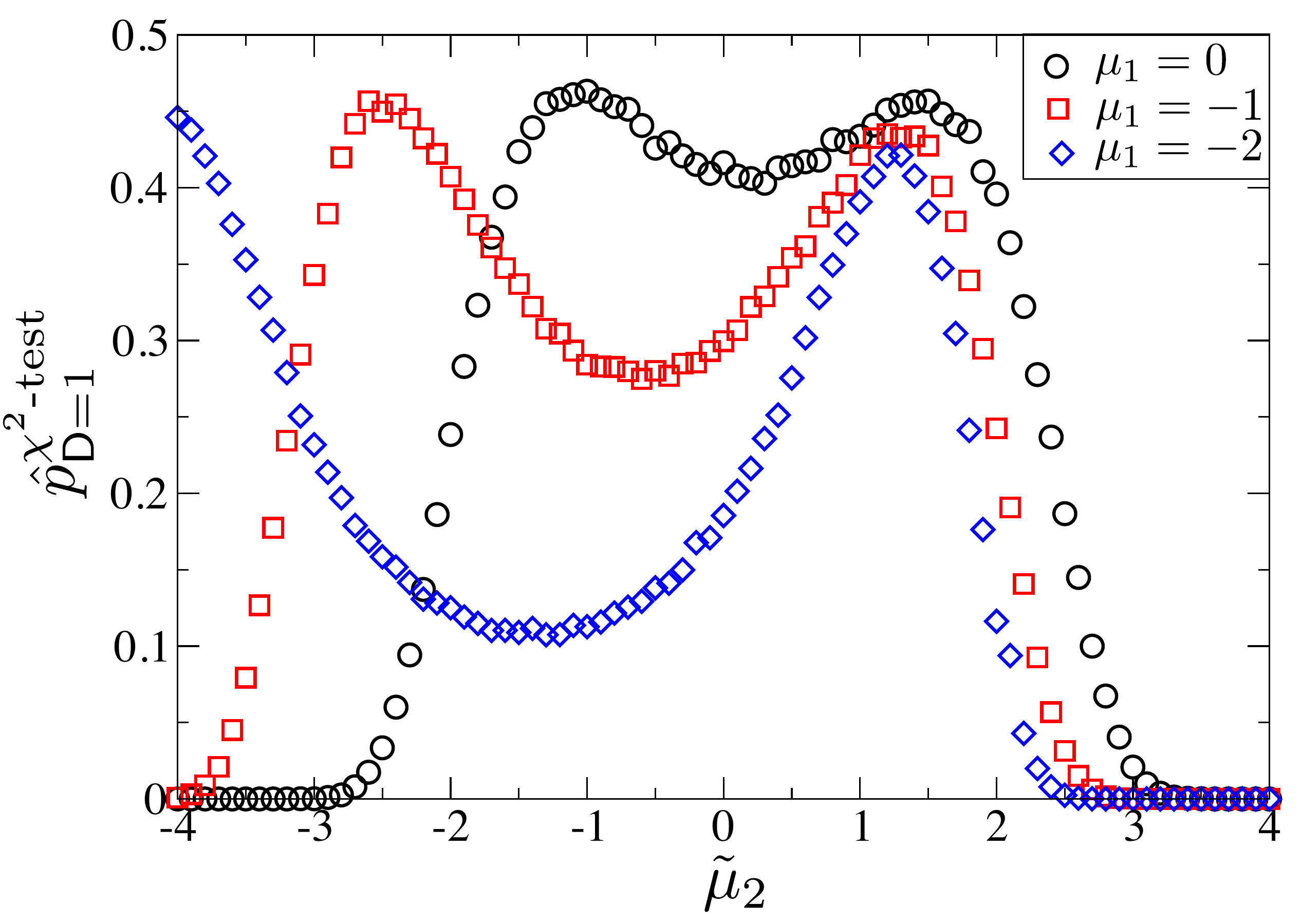}\label{fig4a}} {\hspace{+0.2cm}}
\subfigure[Two sample $\chi^2$-test of  (\ref{CondInvolution_Dec_tilde_discrete2}) for $\mathsf{D}=2$]
{\includegraphics[width=0.48\columnwidth]{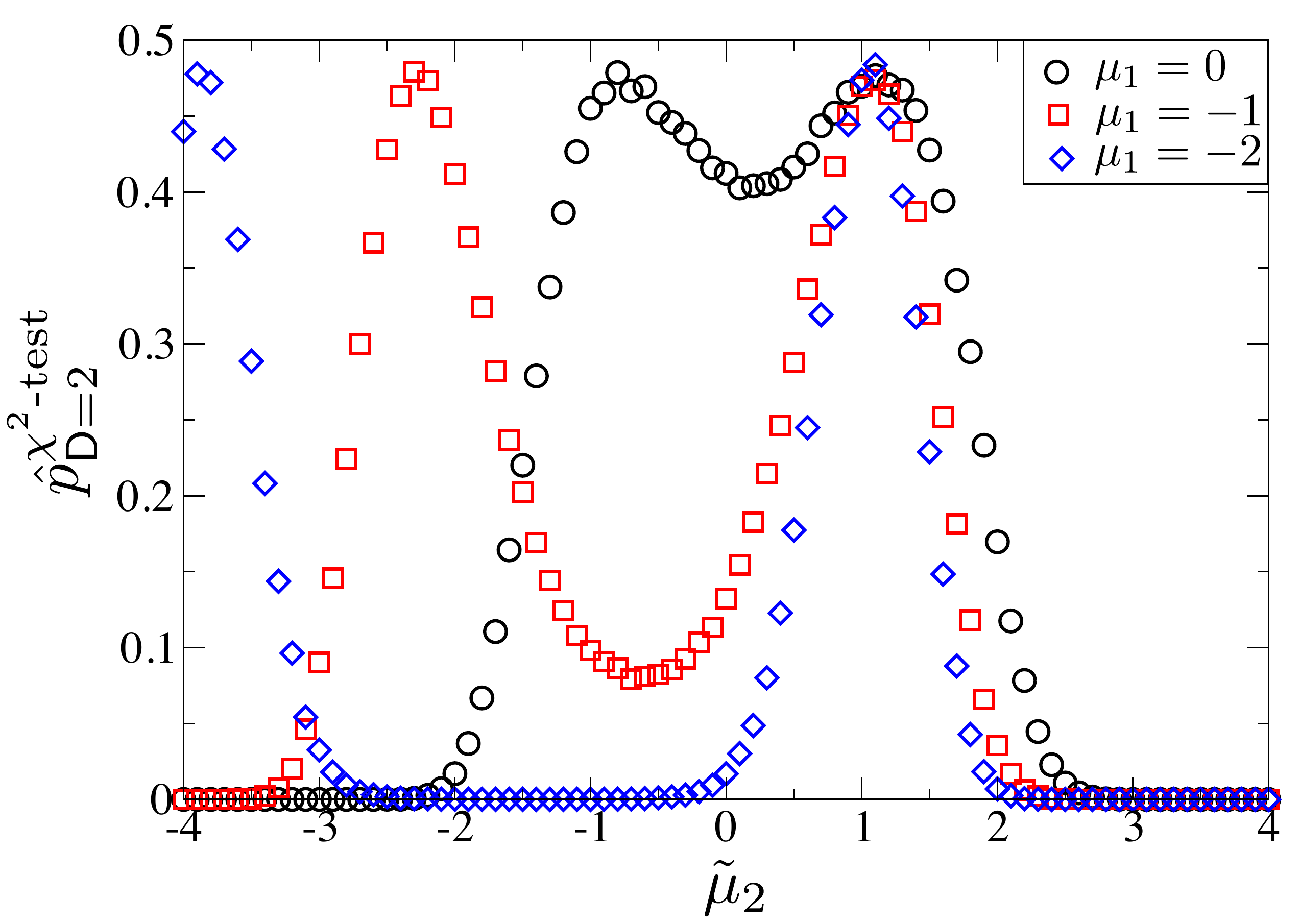}\label{fig4b}}
\caption{ Illustration of  testing optimality for discrete-time observations and for known hypotheses as described in Section~\ref{sec:test1_disc}.     Numerical results presented are for    the observation model given by (\ref{eq:obsModel}) and the decision model of the black box given by (\ref{eq:modelS}) to (\ref{Def_D_Waldxx}).  The observation process has parameters $\sigma_1  =5$, $(\mu_2, \sigma_2) = (1,10)$ and  the values of $\mu_1$ are given in the legend.   The decision model of the black-box decision device has threshold values  
 $L_1 = 4$ and $L_2 = -2$   and parameters $(\tilde{\mu}_1, \tilde{\sigma}_1) = (\mu_1, \sigma_1) $, $\tilde{\sigma}_2 = \sigma_2$ and with $\tilde{\mu}_2$ as  given by the abscissa.         We plot the $p$-values for the null hypothesis that (a) $p_{\mathsf{T}}(\mathsf{T}|\mathsf{D}=1,\mathsf{H}=1) =  p_{\mathsf{T}}(\mathsf{T}|\mathsf{D}=1,\mathsf{H}=2)$   (b)  $p_{\mathsf{T}}(\mathsf{T}|\mathsf{D}=2,\mathsf{H}=1) =  p_{\mathsf{T}}(\mathsf{T}|\mathsf{D}=2,\mathsf{H}=2)$.     The estimates of the p-values are average values over $1e+4$  two-sample $\chi^2$-tests; each two-sample $\chi^2$-tests   evaluates a p-value over a  population of $1e+5$ outcomes  of the black-box decision device.    We took  $P(\mathsf{ H}=1) = P\left(\mathsf{H}=2\right)=1/2$, and a maximum observation window of $10$ observations, i.e., all outcomes with more than $10$ observations are discarded.   
  }\label{fig:test}
\end{figure}

\begin{figure}
\subfigure[Divergence to optimality in information usage]
{\includegraphics[width=0.48\columnwidth]{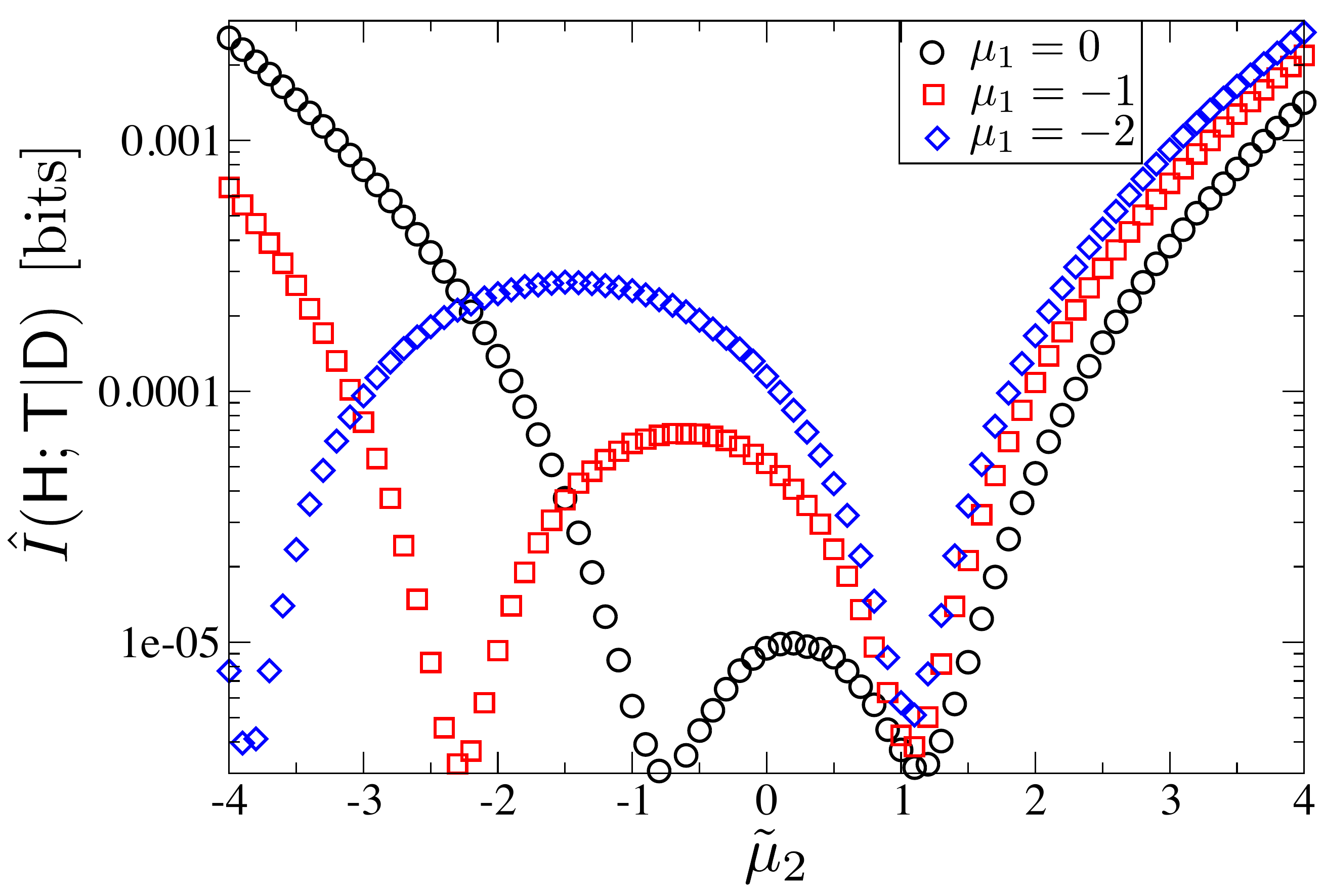}\label{figmutualaT}} {\hspace{+0.2cm}}
\subfigure[Divergence to optimality in average decision times]
{\includegraphics[width=0.48\columnwidth]{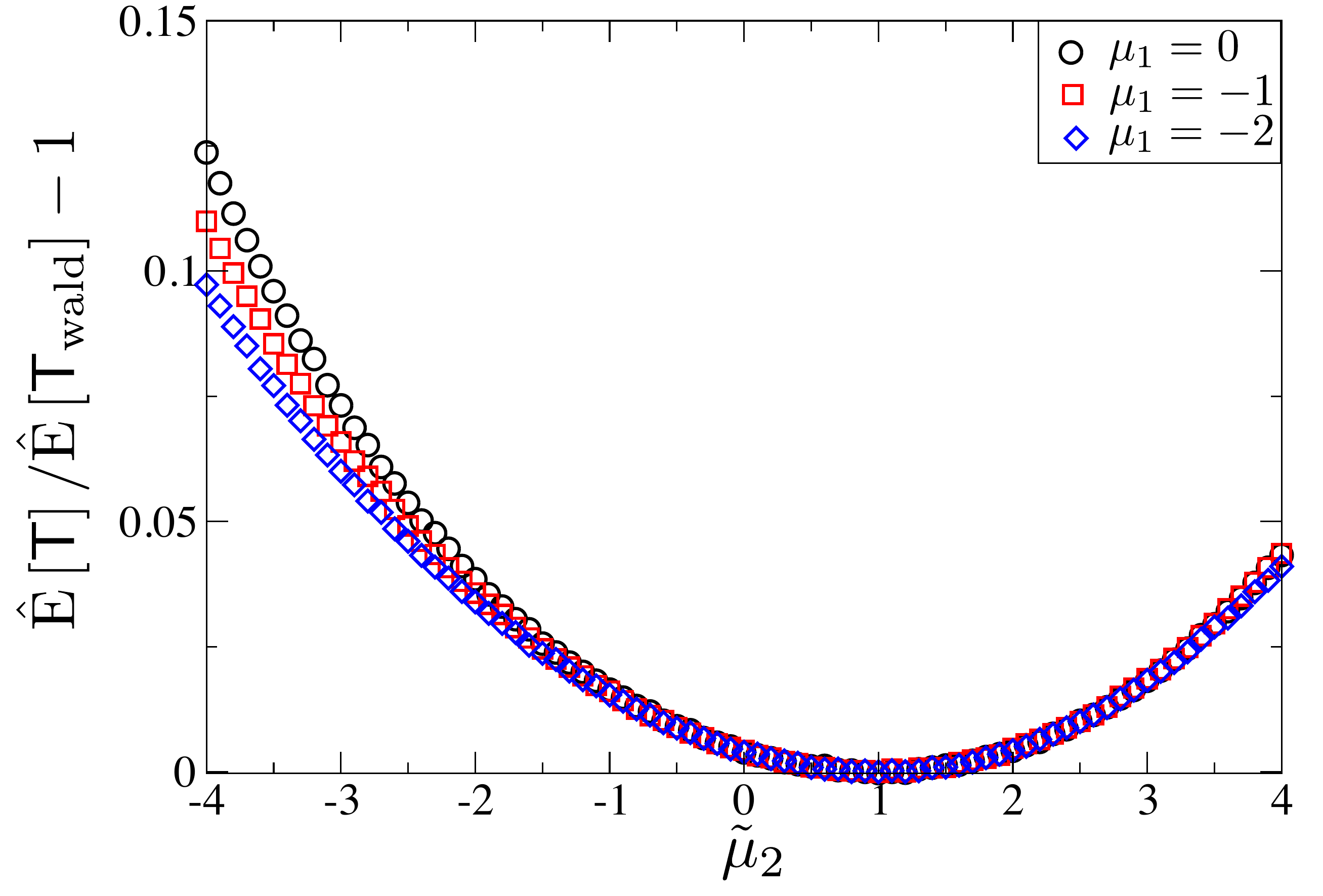}\label{figmutualbT}}
\caption{Estimating the divergence to optimality in the sense of Definitions \ref{DefinitionMinimumMeanTime} and \ref{DefinitionOptInfUse},    using the observation model given by (\ref{eq:obsModel}) and the decision model of the black box given by (\ref{eq:modelS}) to (\ref{Def_D_Waldxx}).  The parameters of the observation process and of the decision model are the same as in Fig.~\ref{fig:test}.  Fig.~\ref{figmutualbT} shows the estimate of the average decision time of the black-box decision device divided by the estimate of the average decision time of a Wald test ($\mu_2 = \tilde{\mu}_2$),   with the same error probabilities as achieved by the black-box decision device, minus one;  therefore each Wald test has different  values of the thresholds $L_1$ and $L_2$  depending on $\tilde{\mu}_2$.     In Fig.~\ref{figmutualaT}  each sample point is calculated using $1e+9$ simulation runs, and in Fig.~\ref{figmutualbT}  using $1e+8$ simulation runs.    We took  $P(\mathsf{H}=1) = P\left(\mathsf{H}=2\right)=1/2$.   }\label{figmutual}
\end{figure}

We now study the decision time distributions to illustrate Theorem~\ref{Theorem3Disc} using numerical simulations.    In Fig.~\ref{fig3} we present the estimated decision time distributions for optimal and suboptimal  sequential decision-making.  Consistent with Theorem~\ref{Theorem3Disc}  the estimates of the distributions $\hat{P}(\mathsf{T}=k|\mathsf{H}=2,\mathsf{D}=a)$ and $\hat{P}(\mathsf{T}=k|\mathsf{H}=1,\mathsf{D}=a)$ ($a\in\left\{1,2\right\}$) overlap if the black-box decision device performs the Wald test and if condition (\ref{ConditionDiscreteTime}) is approximately fulfilled as shown in Fig.~\ref{fig3a}.    If the black-box decision device is suboptimal, as is the case in  Fig.~\ref{fig3b}, then these two distributions are different.      Moreover, since (\ref{ConditionDiscreteTime}) is only approximately fulfilled, the theoretical distributions $P(\mathsf{T}=k|\mathsf{H}=2,\mathsf{D}=a)$ and $P(\mathsf{T}=k|\mathsf{H}=1,\mathsf{D}=a)$ ($a\in\left\{1,2\right\}$) corresponding to the estimates shown in Fig.~\ref{fig3a} are also different.   This example illustrates the value of Theorem~\ref{Theorem3Disc} to quantify optimality for practical purposes in discrete time.

In Fig.~\ref{fig:test} we use the statistical test for optimality described in Section \ref{sec:test1_disc}.
 We plot the estimates of the p-values $p^{\chi^2\textrm{test}}_{\mathsf{D}=1}$ and $p^{\chi^2\textrm{test}}_{\mathsf{D}=2}$ corresponding to, respectively,  a  two-sample $\chi^2$-test of the subsets $\mathcal{A}_{1,1}$ and $\mathcal{A}_{2,1}$, see Fig.~\ref{fig4a}, and a   two-sample $\chi^2$-test of the subsets $\mathcal{A}_{1,2}$ and $\mathcal{A}_{2,2}$, see Fig.~\ref{fig4b}.   These p-values denote the probability to falsely reject the null hypothesis that the samples in the two data sets are drawn from the same decision time distribution. 
Therefore, for example in the case of $\mu_1=-2$ and $\mathsf{D}=1$ we can safely reject the null hypothesis for values of $\tilde{\mu_2}>3$ since the $p$-value is small.  For values of $\tilde{\mu_2}\in[-4,3]$ we need more data to safely reject the  hypothesis that the test is optimal.

\subsubsection{Measuring divergence to optimality of black-box decision devices}
With the same example  as in Fig.~\ref{fig:test} we illustrate  how to use Corollary~\ref{Corollary4_+}  to estimate the divergence of a black-box decision device to the optimal case given by Definition~\ref{DefinitionOptInfUse}.     Note that in the example of Fig.~\ref{fig:test} condition (\ref{ConditionDiscreteTime}) is only approximately fulfilled, and therefore we only expect (\ref{testwald}) to be  approximately fulfilled.  
In Fig.~\ref{figmutualaT} we present the numerical estimates of  $I\left(\mathsf{H};\mathsf{T}|,\mathsf{D}\right)$.   In accordance with  Corollary \ref{Corollary4_+}, if the test is optimal, i.e., the black-box decision device uses a Wald test, than the  estimate of the mutual information  $I\left(\mathsf{H};\mathsf{T}|,\mathsf{D}\right)$ is minimal and approaches zero.   In the example of Fig.~\ref{figmutualaT}  this happens at $\tilde{\mu}_2 = 1$.   For this case also the mean decision time is minimized as shown in Fig.~\ref{figmutualbT}.   Note that Corollary \ref{Corollary4_+} is not a sufficient condition to test optimality with respect to Definition~\ref{DefinitionMinimumMeanTime}, i.e., to test whether the black-box decision device achieves the minimum mean decision time, as is illustrated by Fig.~\ref{figmutualaT}  where we observe a second minimum for the estimate of     $I\left(\mathsf{H};\mathsf{T}|,\mathsf{D}\right)$.    However, for this second minimum the black-box decision device is optimal with respect to Definition~\ref{DefinitionOptInfUse}, which is not related to a minimum mean decision time as illustrated in Fig.~\ref{figmutualbT}.  Note that the estimation of $I\left(\mathsf{H};\mathsf{T}|,\mathsf{D}\right)$ in Fig.~\ref{figmutualaT} requires only knowledge of the output of the decision device whereas the estimation of the minimum mean decision time $\mathrm{E}\left[\mathsf{T}_{\rm Wald}\right]$  requires knowledge on the statistics of the observation process, which in practical applications is  often  unavailable.

 \subsubsection{Testing optimality in case of unknown hypotheses} 
 In this section now we consider testing optimality in case of unknown hypotheses based on Theorem~\ref{Theorem4Disc}.  However, the example given by (\ref{eq:obsModel}) to  (\ref{Def_D_Waldxx}) is not suitable to discuss Theorem~\ref{Theorem4Disc}. The reason is that the cumulative log-likelihood ratio process $\tilde{\mathsf{S}}_k$  becomes a drift-diffusion process in the continuous limit independent of the choice of $\tilde{\mu}_i$ and $\tilde{\sigma}_i$, for which it is known that the two-boundary first-passage time distribution with symmetric thresholds satisfies the fluctuation relation \cite{Roldan_etal15}.   Therefore, the estimates of the distributions $\hat{P}(\mathsf{T}=k|\mathsf{H}=a,\mathsf{D}=1)$ and $\hat{P}(\mathsf{T}=k|\mathsf{H}=a,\mathsf{D}=2)$ ($a\in\left\{1,2\right\}$) always overlap (data not shown).   
 
Therefore, we choose a different example to illustrate the value of Theorem~\ref{Theorem4Disc}. We consider Markovian observation processes  $\mathsf{X} = (\mathsf{X}_1, \mathsf{X}_2, \ldots, \mathsf{X}_k)$   drawn from one of two probability distributions 
\begin{IEEEeqnarray}{rCL}
p_{\mathsf{X}^k_1}\left(x^k_1|\mathsf{H}=h\right) =   \prod^k_{n=1}p_{\mathsf{X}_2}\left(x_n|\mathsf{H}=h, \mathsf{X}_1 = x_{n-1}\right) \label{eq:obsModelMark} 
\end{IEEEeqnarray}
with $\mathsf{X}_0 = 0$ and $h\in\left\{1,2\right\}$.
   In our example the densities $p_{\mathsf{X}_2}\left(\cdot|\mathsf{H}=h, \mathsf{X}_1 = x_{n-1}\right)$ are Gaussian with mean $\mu_h(\mathsf{X}_1)=v_h + (w_h+1)\: \mathsf{X}_1$ and variance $\sigma^2_h$ with $h\in\left\{1,2\right\}$, corresponding to the two hypotheses $\mathsf{H}=h$.      If $v_1 = -v_2$, $w_1 = w_2$ and $\sigma_1 = \sigma_2$, then the involution property  (\ref{InvolutionCond}) holds, such that   Theorem~\ref{Theorem4Disc}  can be applied.

\begin{figure}
\subfigure[Optimal decision-making: black-box decision device uses the correct model of the external world]
{\includegraphics[width=0.48\columnwidth]{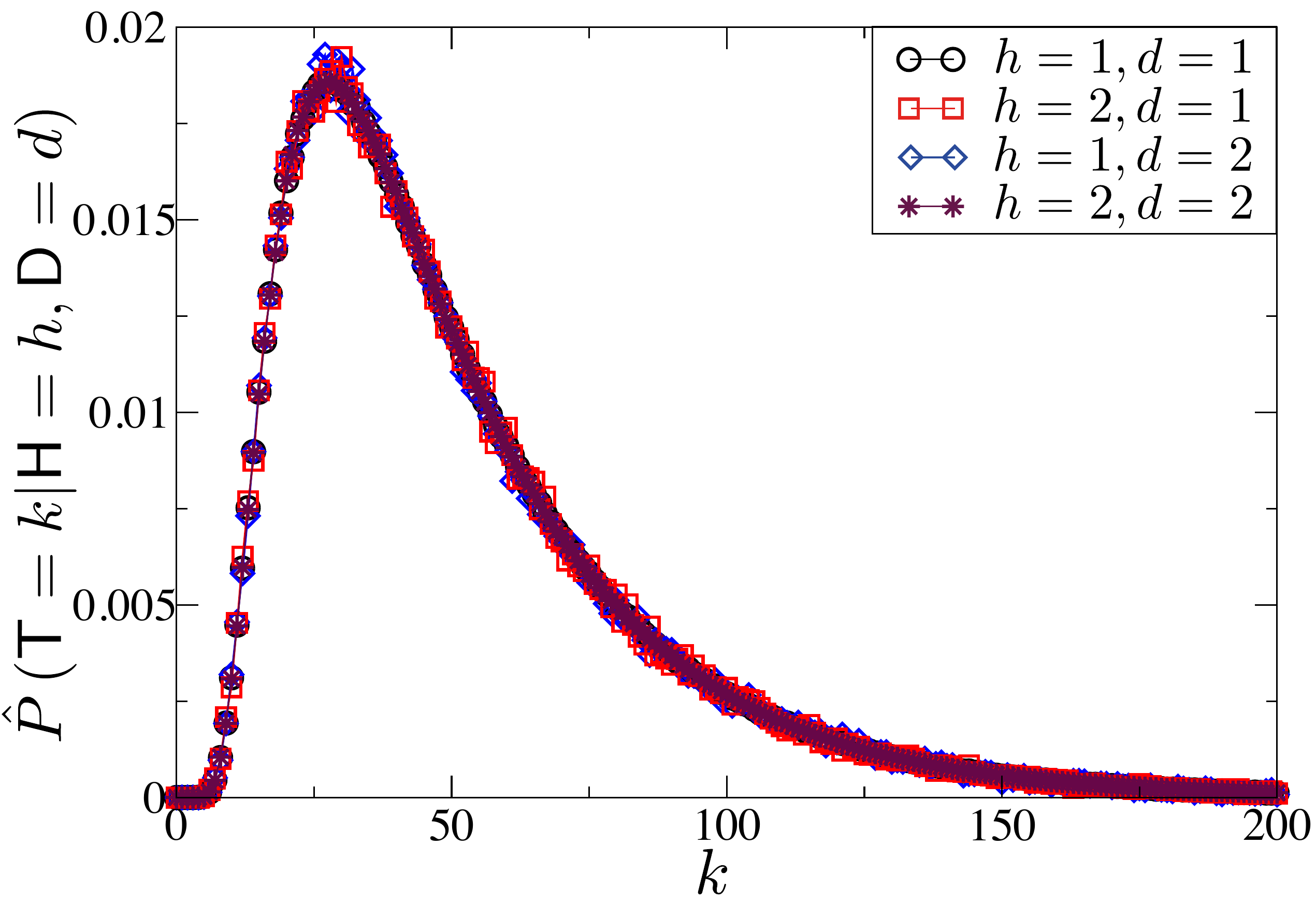}\label{fig3aT}} {\hspace{+0.2cm}}
\subfigure[Sub-optimal decision-making: black-box decision device uses the wrong model of the external world]
{\includegraphics[width=0.48\columnwidth]{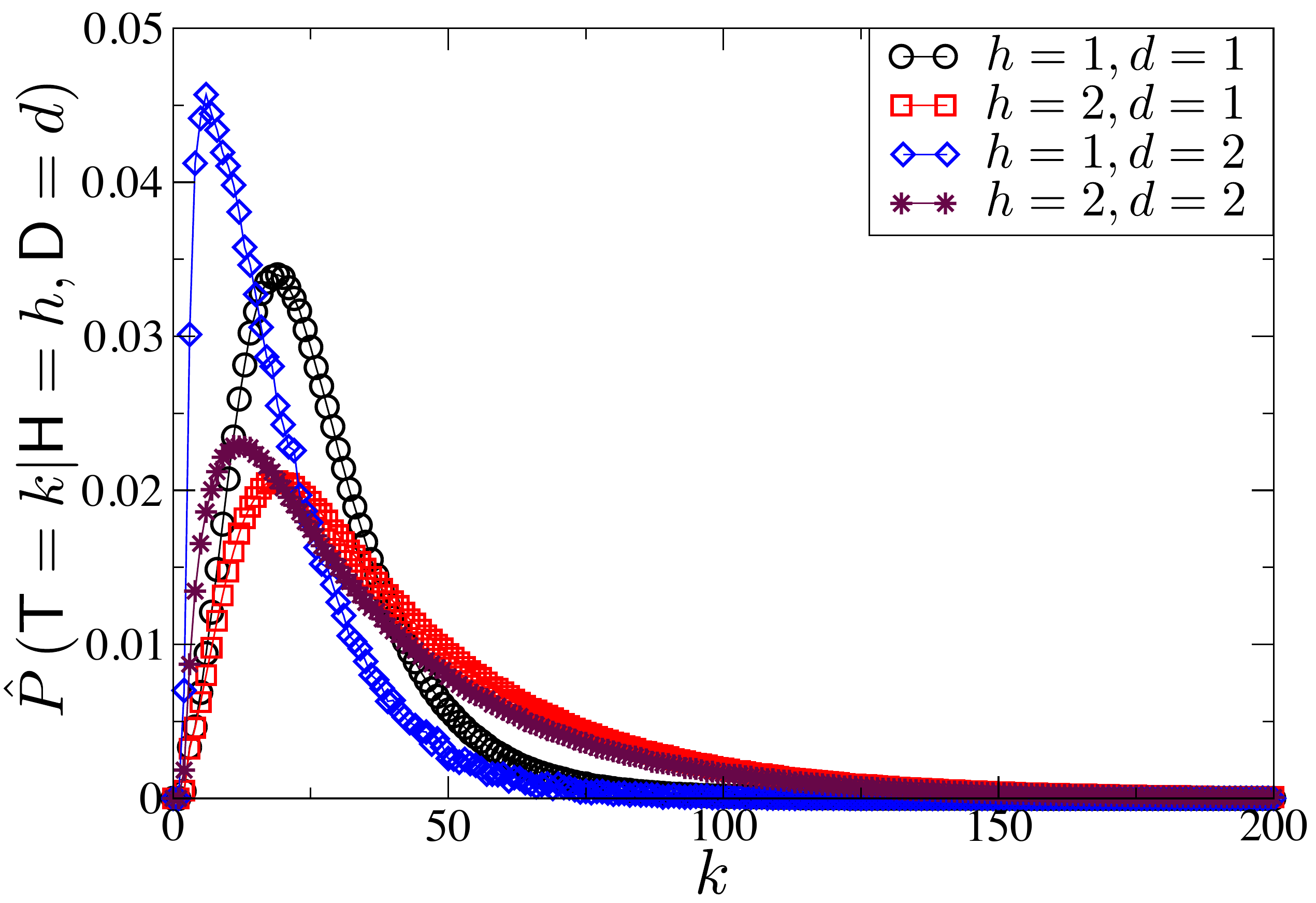}\label{fig3bT}}
\caption{Illustration of Theorem~\ref{Theorem4Disc} using the observation model given by (\ref{eq:obsModelMark}) and the decision model given by (\ref{eq:modelSMark}), (\ref{Def_T_Waldxx}), and (\ref{Def_D_Waldxx}). The  distributions of the observation process corresponding to the two hypotheses  are  with parameters  $v_1 = -v_2 = 1$,  
$w_1 = w_2 = -1$, $\sigma_1 = \sigma_2 = 5$, respectively.    The decision model has threshold values  $L_1 = 4$ and $L_2 = -4$ and uses the parameters:  (a)  $(\tilde{v}_1, \tilde{w}_1,\tilde{\sigma}_1) = (v_1, w_1, \sigma_1)$ and $(\tilde{v}_2, \tilde{w}_2, \tilde{\sigma}_2) = (v_2, w_2, \sigma_2)$, (b)   $(\tilde{v}_1, \tilde{w}_1,\tilde{\sigma}_1) = (v_1, w_1, \sigma_1)$ and $(\tilde{v}_2, \tilde{w}_2,\tilde{\sigma}_2) = (v_2, -0.5,  \sigma_2)$.   The empirical error probabilities are in (a) $\alpha_1 = 0.014$ and $\alpha_2 = 0.014$ and (b) $\alpha_1 = 0.0092$ and $\alpha_2 = 0.6932$.  The distributions are estimated using $1e+7$ simulation runs.   }\label{fig3T}
\end{figure}  

   We consider again a class of  black-box decision devices which use the Wald sequential probability ratio test based on its model of the external world.   The black-box decision devices compute the cumulative log-likelihood ratio   based on parameters $\tilde{v}_h$, $\tilde{w}_h$ and    $\tilde{\sigma}^2_h$ (with $h\in\left\{1,2\right\}$), i.e.,
\begin{IEEEeqnarray}{rCL}
\tilde{\mathsf{S}}_k &=& \sum^k_{n=1} \log \left(\frac{p_{\mathsf{X}_2|\mathsf{H},\mathsf{X}_1}\left(\mathsf{X}_{n}|\mathsf{H}=1,\mathsf{X}_{n-1}\right)}{p_{\mathsf{X}_2|\mathsf{H},\mathsf{X}_1}\left(\mathsf{X}_{n}|\mathsf{H}=2, \mathsf{X}_{n-1}\right)}\right)  \nonumber
\\ 
&=& k \log \frac{\tilde{\sigma}_2}{\tilde{\sigma}_1}  
  + \sum^k_{n=1}\left(\frac{(\mathsf{X}_{n}-\mathsf{X}_{n-1}-\tilde{v}_2 - \tilde{w}_2\mathsf{X}_{n-1})^2}{2\tilde{\sigma}^2_2}-\frac{(\mathsf{X}_{n}-\mathsf{X}_{n-1}-\tilde{v}_1 - \tilde{w}_1\mathsf{X}_{n-1})^2}{2\tilde{\sigma}^2_1}\right).\nonumber\\\label{eq:modelSMark}
\end{IEEEeqnarray}
The decision time $\mathsf{T}$  and the decision variable $\mathsf{D}$  are still 
given by (\ref{Def_T_Waldxx}) and (\ref{Def_D_Waldxx}) with  the two thresholds $L_1>0$ and $L_2<0$.  

We now illustrate Theorem~\ref{Theorem4Disc}   using numerical simulations.         In Fig.~\ref{fig3T} we illustrate Theorem~\ref{Theorem4Disc} for optimal and suboptimal  sequential decision-making with symmetric thresholds $L_1 = -L_2=4$ and for $v_1 = -v_2 = 1$,  
$w_1 = w_2 = -1$, $\sigma_1 = \sigma_2 = 5$ such that the  involution property  (\ref{InvolutionCond}) holds.  
Consistent with Theorem~\ref{Theorem4Disc}  the distributions  $\hat{P}(\mathsf{T}=k|\mathsf{H}=h,\mathsf{D}=1)$ and $\hat{P}(\mathsf{T}=k|\mathsf{H}=h,\mathsf{D}=2)$ ($h\in\left\{1,2\right\}$) overlap if the black-box decision device performs the Wald test, and is thus optimal, and if (\ref{ConditionDiscreteTime}) approximately applies.    If the black-box decision device is suboptimal, as is the case in  Fig.~\ref{fig3bT}, then these two distributions may be different.   
Note that since Theorem~\ref{Theorem3Disc} also applies, all distributions in  Fig.~\ref{fig3aT} overlap.

 \subsubsection{Overshoot problem}\label{SectOvershot}
 Due to the overshoot problem for discrete-time observation processes in general the  condition given  by (\ref{ConditionDiscreteTime}) is violated.   Therefore, even in the case of the Wald test  $I(\mathsf{H};\mathsf{T}_{\rm Wald}|\mathsf{D}_{\rm Wald})$ is in general larger than zero.  In the present section, we discuss how far $I(\mathsf{H};\mathsf{T}_{\rm Wald}|\mathsf{D}_{\rm Wald})$ deviates  from zero  in practical situations.   We also discuss 
 how far the condition imposed by (\ref{ConditionDiscreteTime}) is fulfilled in our numerical examples.   
 
 For this purpose, we first estimate $I(\mathsf{H};\mathsf{T}_{\rm Wald}|\mathsf{D}_{\rm Wald})$ as a function of  the threshold values and the number of test runs.   In Fig.~\ref{figrunsb} and Fig.~\ref{figruns} it can be seen that for the Wald test  the estimate of $I(\mathsf{H};\mathsf{T}_{\rm Wald}|\mathsf{D}_{\rm Wald})$  saturates with an increasing number of test runs at a non-zero value, and therefore  $I(\mathsf{H};\mathsf{T}_{\rm Wald}|\mathsf{D}_{\rm Wald})>0$.   This is an intrinsic aspect of  sequential decision-making with discrete observation processes and cannot be avoided.  For small values of $\lambda$, which parameterizes the threshold values, we see ripples in the mutual information. The minima occur approximately at integer multiples of the  most likely value of the increase of the cumulative log-likelihood ratio $\mathsf{S}_k$.  For example, in Fig.~\ref{figruns} we illustrate how the estimate of the mutual information converges to its asymptotic value for $\lambda=0.16$ and $\lambda = 0.36$, corresponding to the first maximum and the third minimum in Fig.~\ref{figruns}.    For large values of $\lambda$, i.e., when the   distance of the thresholds to the origin is large with respect to the typical  increase of the cumulative log-likelihood ratio, the mutual information  $I(\mathsf{H};\mathsf{T}_{\rm Wald}|\mathsf{D}_{\rm Wald})$ decreases as a function of $\lambda$.    
 Even for large values of $\lambda$, the estimate of the mutual information   $I(\mathsf{H};\mathsf{T}_{\rm Wald}|\mathsf{D}_{\rm Wald})$  does not converge to zero as a function of the number of test runs but saturates, as the condition in    (\ref{ConditionDiscreteTime}) is not  fulfilled.     This is illustrated in Fig.~\ref{figruns} for the values $\lambda=1$ and $\lambda=3$.

 The fact that $I(\mathsf{H};\mathsf{T}_{\rm Wald}|\mathsf{D}_{\rm Wald})$ is larger than zero  indicates that here the condition given by    (\ref{ConditionDiscreteTime}) is not fulfilled. To show this, in Fig.~\ref{figrunsc} we plot $\mathrm{E}\left[e^{\mathsf{M}_1}|\Phi_{1}(k),\mathsf{H}=2\right]$ as a function of time $k$.    

In conclusion,  Corollary~\ref{Corollary4_+} is applicable to test optimality of the black-box decision device if condition  (\ref{ConditionDiscreteTime})  is approximately fulfilled, which is the case when  the threshold values  of the Wald test are far enough from the origin in comparison to the average increase of the cumulative log-likelihood ratio  per observation.

\begin{figure*}[h!]  
\subfigure[Estimate  $\hat{I}(\mathsf{H};\mathsf{T}_{\rm Wald}|\mathsf{D}_{\rm Wald})$ as a function of  the distance parameter  $\lambda$  of the thresholds and for given values of the number of test runs $N$.]
{\includegraphics[width=0.48\columnwidth]{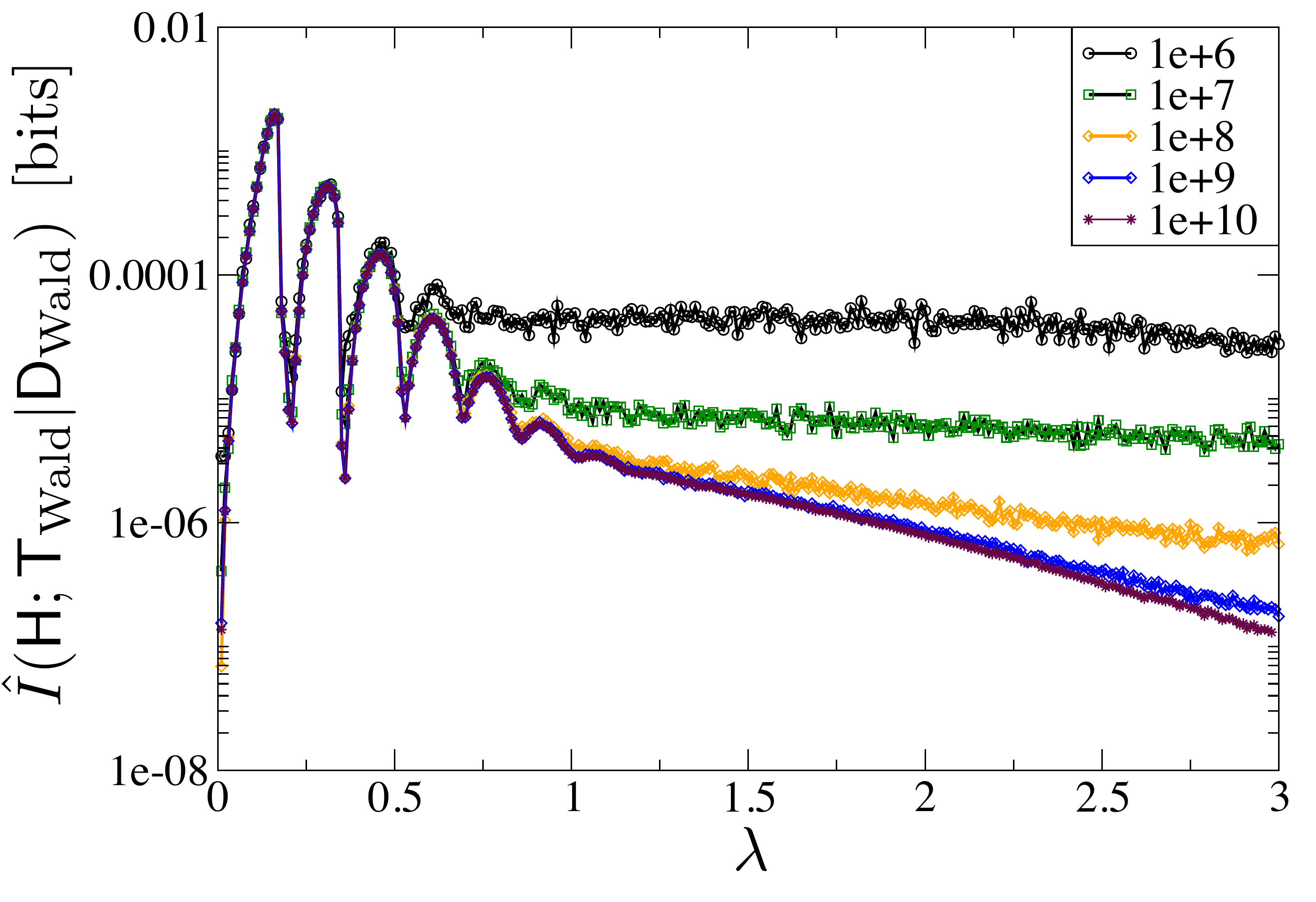}\label{figrunsb}}
\subfigure[Convergence of the estimate  $\hat{I}(\mathsf{H};\mathsf{T}_{\rm Wald}|\mathsf{D}_{\rm Wald})$ over the number of test runs $N$ and for given values of the distance parameter $\lambda$  of the thresholds. ]
{\includegraphics[width=0.48\columnwidth]{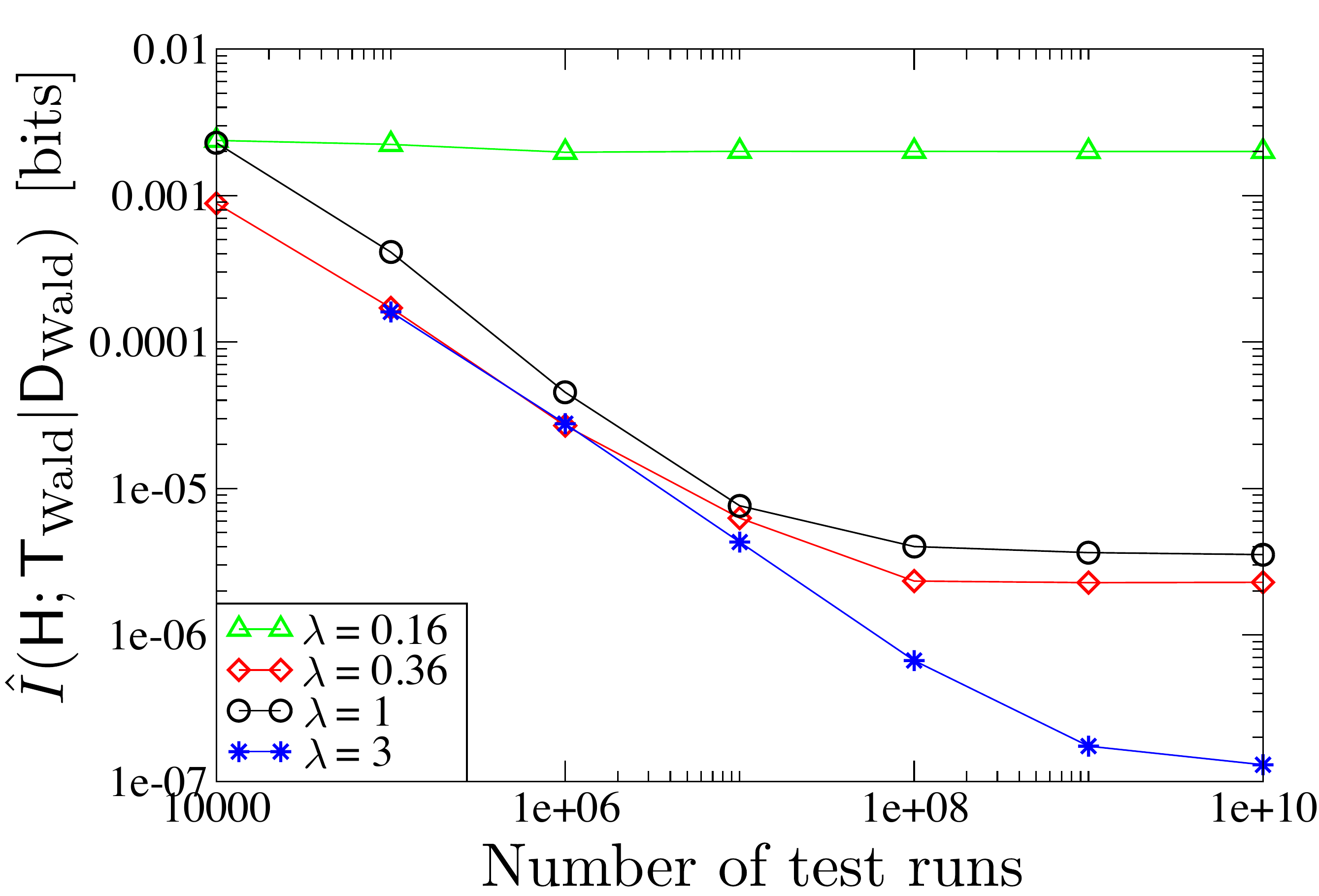}\label{figruns}} {\hspace{+0.2cm}}
\subfigure[Evaluation of the condition (\ref{ConditionDiscreteTime}) for the Wald test based on plotting $\mathrm{E}\left(e^{\mathsf{M}_1}|\Phi_{1}(k),\mathsf{H}=2\right)$ over $k$ (solid lines).  Also the corresponding $\hat{P}\left(\mathsf{T}_{\rm Wald} = k|\mathsf{H}=2, \mathsf{D}_{\rm Wald}=1\right)$ is shown (dashed lines).   It can be seen that $\mathrm{E}\left(e^{\mathsf{M}_1}|\Phi_{1}(k),\mathsf{H}=2\right)$ is not independent of $k$ such that  (\ref{ConditionDiscreteTime}) does not hold.   Especially for $\lambda=0.16$, the value corresponding to the first maximum in Fig.~\ref{figrunsb},    $\mathrm{E}\left(e^{\mathsf{M}_1}|\Phi_{1}(k),\mathsf{H}=2\right)$ varies in the area with the majority of the probability mass of the termination time yielding a larger $I(\mathsf{H};\mathsf{T}_{\rm Wald}|\mathsf{D}_{\rm Wald})$  than in the case of $\lambda=0.36$, the value corresponding to the third minimum in    Fig.~\ref{figrunsb}, where $\mathrm{E}\left(e^{\mathsf{M}_1}|\Phi_{1}(k),\mathsf{H}=2\right)$ varies less over $k$.
] 
{\includegraphics[width=0.48\columnwidth]{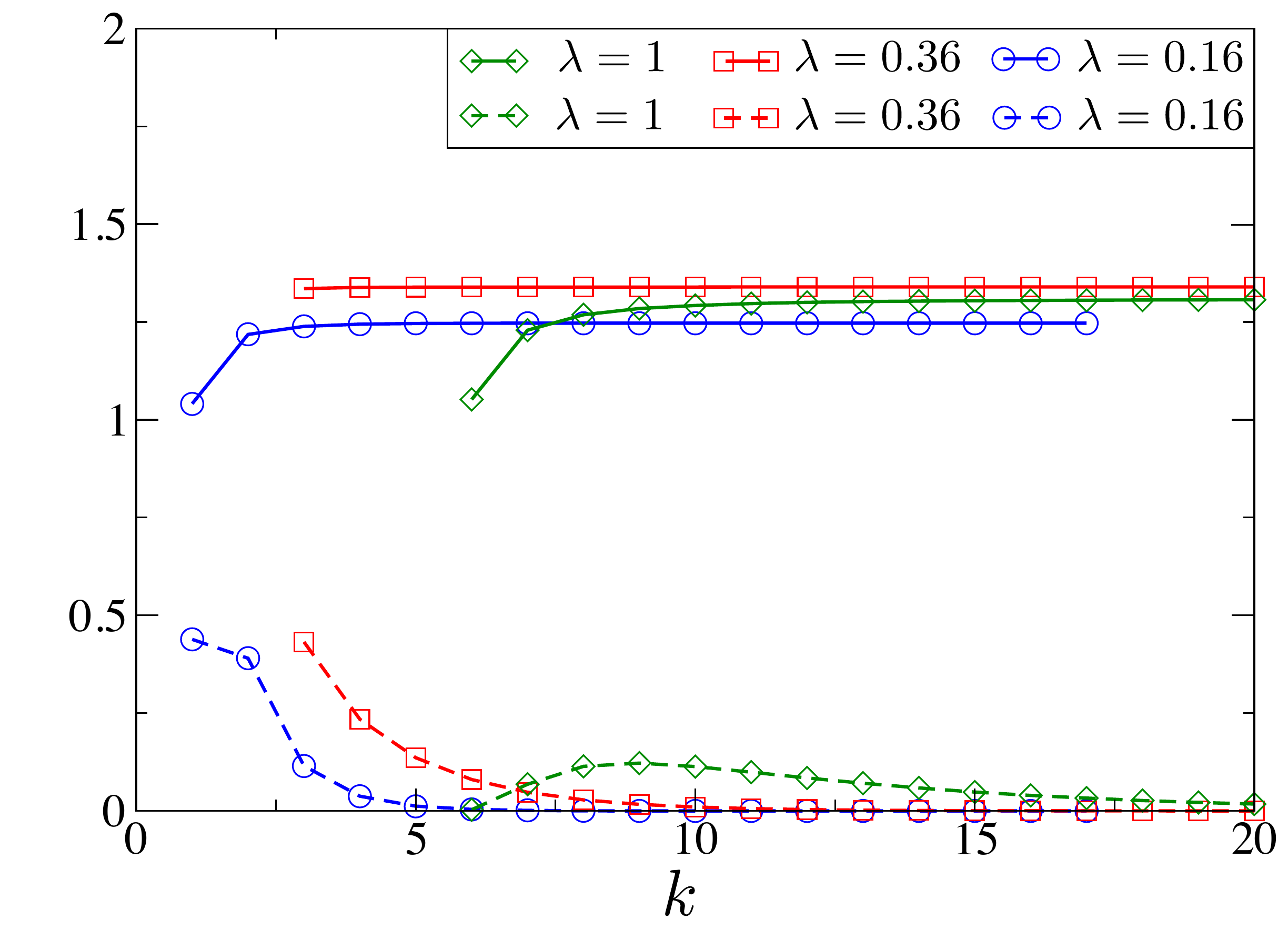}\label{figrunsc}}
	\caption{Illustration of the impact of discreteness of the observation process.  As in Fig.~\ref{fig3}, \ref{fig:test} and \ref{figmutual}, we use the observation model given by (\ref{eq:obsModel}) with parameters  $(\mu_1, \sigma_1) = (0,5)$ and $(\mu_2, \sigma_2) = (1,10)$.   The decision model is the Wald test given by (\ref{eq:modelS}) to (\ref{Def_D_Waldxx}) with parameters $(\tilde{\mu}_1, \tilde{\sigma}_1) = (\mu_1, \sigma_1)$,  $(\tilde{\mu}_2, \tilde{\sigma}_2) = (\mu_2, \sigma_2)$,  $L_1 = 4\lambda$, $L_2 = -2\lambda$. }
	
\end{figure*}
 
\subsection{Continuous observation processes}
 In this section we illustrate Theorem~\ref{Theorem1} and  Corollary~\ref{Corollary3}  for continuous observation processes.   
 
 The decision model we study here is a drift diffusion process and has been used to describe reaction-time distributions of two-choice decision tasks of human subjects \cite{ratcliff2004comparison, ratcliff2008diffusion}.   
 
\subsubsection{Observation model and decision model} We consider an observation process $\mathsf{X}_t$   which is an It\^{o}-process  
 solving the stochastic differential equation
\begin{IEEEeqnarray}{rCL} 
\mathrm{d}\mathsf{X}_t = \mu_i\: \mathrm{d}t + \sigma \:\mathrm{d}\mathsf{W}_t  \label{eq:Ito}
\end{IEEEeqnarray}
 where $\mu_i$   is a constant drift,  with $i\in\left\{1,2\right\}$  corresponding to the two hypotheses $\mathsf{H}=i$, where  $\sigma$ is a constant noise amplitude,  
 and where $\mathsf{X}_0= 0$.  Here $\mathsf{W}_t$ is a standard Wiener process. If $\mu_1 = -\mu_2$ then the involution property  (\ref{InvolutionCond}) holds.

We consider black-box decision devices which compute the continuous version 
of the cumulative  log-likelihood ratio in the Wald sequential probability ratio test, cf.~(\ref{eq:modelS}), which is given by
\begin{IEEEeqnarray}{rCL}
\tilde{\mathsf{S}}_t  &=&  t \:\frac{\tilde{\mu}^2_2  - \tilde{\mu}_1^2}{2\tilde{\sigma}^2} + \mathsf{X}_t  \:\frac{\tilde{\mu}_1 -  \tilde{\mu}_2}{\tilde{\sigma}^2}.  \label{eq:modelSCont}
\end{IEEEeqnarray}
The decision time  of the  model is 
\begin{IEEEeqnarray}{rCL}
\mathsf{T}&=&{\rm inf}\{t\in\mathbb{R} : \tilde{\mathsf{S}}_t\notin (L_2,L_1)\}\label{Def_T_Waldxxcont}  
\end{IEEEeqnarray}
and the decision variable is given by
\begin{IEEEeqnarray}{rCL}
\mathsf{D}&=&\left\{\begin{array}{ll}
1 & \textrm{if } \tilde{\mathsf{S}}_{\mathsf{T}}\ge L_1\\
2 & \textrm{if } \tilde{\mathsf{S}}_{\mathsf{T}}\le L_2
\end{array}\right. \label{Def_D_WaldxxCont}
\end{IEEEeqnarray}
with  the two thresholds $L_1>0$ and $L_2<0$.      

Note that the cumulative log-likelihood ratio, in the case the hypothesis $\mathsf{H}  = i$ is true,  is the following It\^{o} process
\begin{IEEEeqnarray}{rCL}
{\rm d}\tilde{\mathsf{S}}_t  = a_i\:{\rm d}t  +\sqrt{2b}\:{\rm d}\mathsf{W}_t   \label{eq:ST}
\end{IEEEeqnarray}   
with 
\begin{IEEEeqnarray}{rCL}
a_i &=& \frac{\tilde{\mu}_1-\tilde{\mu}_2}{\tilde{\sigma}^2} \left(-\frac{\tilde{\mu}_1+\tilde{\mu}_2}{2}+ \mu_i\right), \quad i\in \left\{1,2\right\} \\  
b &=& \frac{1}{2}\left(\sigma \frac{\tilde{\mu}_1-\tilde{\mu}_2}{\tilde{\sigma}^2}\right)^2.
\end{IEEEeqnarray}   
The  sequential decision-making device $(\mathsf{T}, \mathsf{D})$ has error probabilities 
\begin{eqnarray}
\alpha_1 &=& P\left(\mathsf{D}=1|\mathsf{H}=2\right)  =   \frac{1-e^{\frac{a_2L_2}{b}}}{1-e^{\frac{a_2(L_2-L_1)}{b}}} \label{eq:diffP1}\\
\alpha_2 &=&  P\left(\mathsf{D}=2|\mathsf{H}=1\right)  =  \frac{e^{\frac{a_1L_2}{b}}-e^{\frac{a_1(L_2-L_1)}{b}}}{1-e^{\frac{a_1(L_2-L_1)}{b}}}. \label{eq:diffP2}
\end{eqnarray}
The values of $a_1$, $a_2$, $b$, $L_1$ and $L_2$ are  chosen such that $\alpha_1,\alpha_2\in[0,1/2]$.
If $\tilde{\mu}_1 =  \mu_1$,  $\tilde{\mu}_2 = \mu_2$ and $\tilde{\sigma}=\sigma$, then   $\tilde{\mathsf{S}}_t  = \mathsf{S}_t $  
 and $(\mathsf{T}, \mathsf{D}) = (\mathsf{T}_{\rm dec}, \mathsf{D}_{\rm dec})$ with error probabilities as given by (\ref{Def_T1}) and (\ref{Def_T2}).       Notice that the stochastic differential equation of $\mathsf{S}_t$  is of the form~\cite{PhysRevLett.119.140604}
\begin{IEEEeqnarray}{rCL}
{\rm d}\mathsf{S}_t  = \frac{(-1)^{i+1}}{2}\left(\frac{\mu_1-\mu_2}{\sigma}\right)^2 {\rm d}t  +\frac{\mu_1-\mu_2}{\sigma} {\rm d}\mathsf{W}_t, \quad i\in\left\{1,2\right\} \label{eq:STx}
\end{IEEEeqnarray} 
and $e^{-\mathsf{S}_t  } = -\sqrt{2b} \int^{t}_0  e^{-\mathsf{S}_{t'}  }{\rm d}\mathsf{W}_{t'}$ is a $\mathbb{P}_i$-martingale process.  
For the special case of 
\begin{IEEEeqnarray}{rCL}
  \tilde{\mu}_1+\tilde{\mu}_2  &=&  \mu_1 +\mu_2\label{EqSymOptCond}
\end{IEEEeqnarray} 
we have   $\tilde{\mathsf{S}}_{t} = c\, \mathsf{S}_t$   with $c = \left(\frac{\sigma}{\tilde{\sigma}}\right)^2 \left(\frac{\tilde{\mu}_1-\tilde{\mu}_2}{\mu_1-\mu_2}\right)$ and hence  $(\mathsf{T}, \mathsf{D}) = (\mathsf{T}_{\rm dec}, \mathsf{D}_{\rm dec})$ with error probabilities, $ \alpha_1 =(1-e^{L_2/c})/(1-e^{(L_2-L_1)/c})$ and $\alpha_2 =(e^{L_2/c}-e^{(L_2-L_1)/c})/(1-e^{(L_2-L_1)/c})$. Thus, in case (\ref{EqSymOptCond}) holds the black box decision device is optimal. Note that (\ref{EqSymOptCond}) implies that $a_1=-a_2$.

\subsubsection{Illustration of Theorem~\ref{Theorem1}}
We consider now the special case of   
\begin{IEEEeqnarray}{rCL}
|L_2|\gg b/|a_1|, \quad |L_2|\gg b/|a_2|  \label{eq:conddiv}
\end{IEEEeqnarray} 
for which the expression of the distribution of decision times simplifies and allows analytical evaluation. 

In the following we illustrate Theorem~\ref{Theorem1}. The Laplace transform of the distributions of decision times are known for arbitrary values of $L_1$ and $L_2$  \cite{redner2001guide}.     If the conditions  in (\ref{eq:conddiv}) are fulfilled, we get
\begin{IEEEeqnarray}{rCL}
p_{\mathsf{T}}(t|\mathsf{D}=1, \mathsf{H}=1) &=& \frac{L_1}{2\sqrt{\pi b}\:t^{3/2}} e^{-\frac{(|a_1|t-L_1)^2}{4bt}} + o(1)\label{eq:distriTime1} \\ 
p_{\mathsf{T}}(t|\mathsf{D}=1, \mathsf{H}=2) &=&\frac{L_1}{2\sqrt{\pi b}\:t^{3/2}} e^{-\frac{(|a_2|t-L_1)^2}{4bt}}+ o(1)\label{eq:distriTime2} \\
p_{\mathsf{T}}(t|\mathsf{D}=2, \mathsf{H}=1) &=&\frac{1}{1-e^{-\frac{|a_1|}{b}L_1}} \frac{1}{\sqrt{\pi b}\:t^{3/2}} e^{-\frac{(|a_1|t+L_2)^2}{4bt}} \nonumber\\
&&\times
 \left\{\frac{1}{2} |L_2| -  (L_1+\frac{1}{2} |L_2|)e^{-L^2_1/(bt) - |L_2|L_1/(bt)}\right\} + o(1) \label{eq:distriTime3} \\
p_{\mathsf{T}}(t|\mathsf{D}=2, \mathsf{H}=2) &=&\frac{1}{1-e^{-\frac{|a_2|}{b}L_1}}\frac{1}{\sqrt{\pi b}\:t^{3/2}} e^{-\frac{(|a_2|t+L_2)^2}{4bt}}  \nonumber\\
&&\times \left\{\frac{1}{2}  |L_2| -  (L_1+\frac{1}{2} |L_2|)e^{-L^2_1/(bt) - |L_2|L_1/(bt)}\right\}+ o(1)   \label{eq:distriTime4}
\end{IEEEeqnarray}  
 where $o$ denotes the little-$o$ notation taken  with respect to $|L_2|$ going to infinity. 
The fluctuation relations (\ref{CondInvolution_Dec}) and  (\ref{CondInvolution_Dec2})  hold for $a_1 = -a_2$, and thus for $\tilde{\mu}_1+\tilde{\mu}_2  =  \mu_1 +\mu_2$.    This is consistent with Theorem \ref{Theorem1} which states that the fluctuation relation must hold whenever  $(\mathsf{T}, \mathsf{D}) = (\mathsf{T}_{\rm dec}, \mathsf{D}_{\rm dec})$.   

\subsubsection{Optimality in mean decision times}
With this example we can also verify optimality of sequential hypothesis testing in the sense of Definition~\ref{DefinitionMinimumMeanTime}. The mean decision times are given by
\begin{IEEEeqnarray}{rCL}
 \mathrm{E}\left[\mathsf{T}|\mathsf{D}=1, \mathsf{H}=1\right]  &=& \frac{L_1}{|a_1|} + O\left(|L_2|e^{|a_1|L_2/b}\right) \label{eq:TTheory1}\\ 
    \mathrm{E}\left[ \mathsf{T} | \mathsf{D}=1, \mathsf{H}=2 \right]&=& \frac{L_1}{|a_2|} + O\left(|L_2|e^{|a_2|L_2/b} \right)  \\ 
      \mathrm{E}\left[ \mathsf{T} |\mathsf{D}=2, \mathsf{H}=1 \right]&=& \frac{1}{|a_1|}\left(|L_2|-2L_1 \frac{e^{-(|a_1|/b)L_1}}{1-e^{-(|a_1|/b)L_1}}\right)  + O\left(|L_2|e^{|a_1|L_2/b}\right)\\ 
         \mathrm{E}\left[ \mathsf{T} |\mathsf{D}=2, \mathsf{H}=2 \right]&=& \frac{1}{|a_2|}\left(|L_2|-2L_1 \frac{e^{-(|a_2|/b)L_1}}{1-e^{-(|a_2|/b)L_1}}\right)  +O\left(|L_2|e^{-a_2L_2/b}\right) 
\end{IEEEeqnarray}  
 where $O$ denotes the big-$O$ notation.  
The corresponding values of the average decision times of the Wald test yielding the same error probabilities $\alpha_1$ and $\alpha_2$ as $(\mathsf{T}, \mathsf{D})$ are 
\begin{IEEEeqnarray}{rCL}
 \mathrm{E}\left[\mathsf{T}_{\rm dec}|\mathsf{D}=1, \mathsf{H}=1\right] &=&  \mathrm{E}\left[\mathsf{T}_{\rm dec}|\mathsf{D}=1, \mathsf{H}=2\right]   = 2\left(\frac{\sigma}{\mu_1-\mu_2}\right)^{2}\log \frac{1-\alpha_2}{\alpha_1}   + O\left(|L_2|e^{L_2}\right). \nonumber\\
\label{eq:TTheory2}
\end{IEEEeqnarray}  
It can be shown that for $L_2\rightarrow -\infty$ we have  $ \mathrm{E}\left[\mathsf{T}|\mathsf{D}=1, \mathsf{H}=1\right] -
 \mathrm{E}\left[\mathsf{T}_{\rm dec}|\mathsf{D}=1, \mathsf{H}=1\right]  \geq 0$, which is consistent with the optimality of the  sequential probability ratio test in the sense of minimal decision times.

\begin{figure}
\subfigure[Divergence to optimality in information usage; note that the black solid line and the green dotted line overlap.]
{\includegraphics[width=0.48\columnwidth]{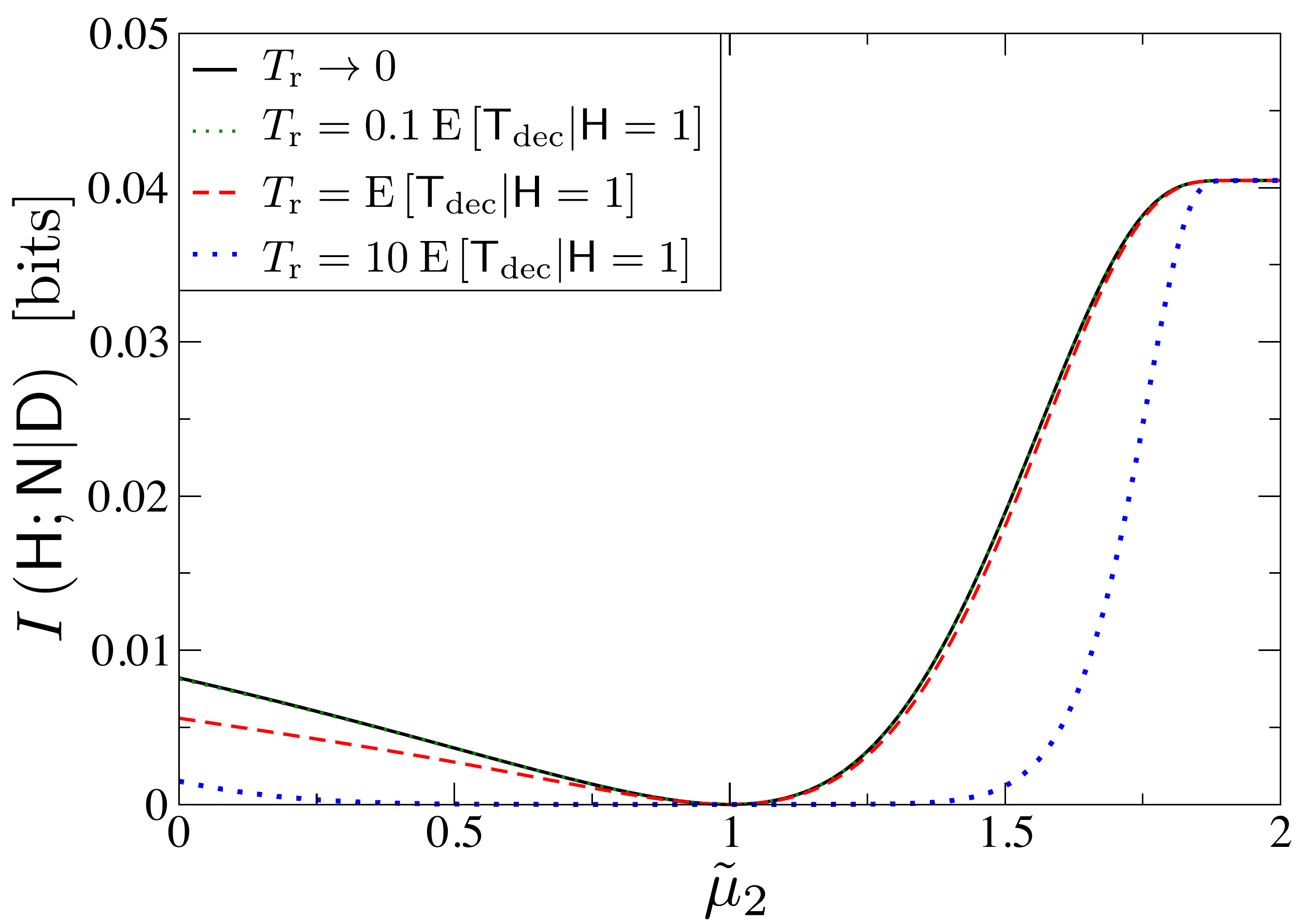}\label{fig:ICont}} {\hspace{+0.2cm}}
\subfigure[Divergence to optimality in average decision times]
{\includegraphics[width=0.48\columnwidth]{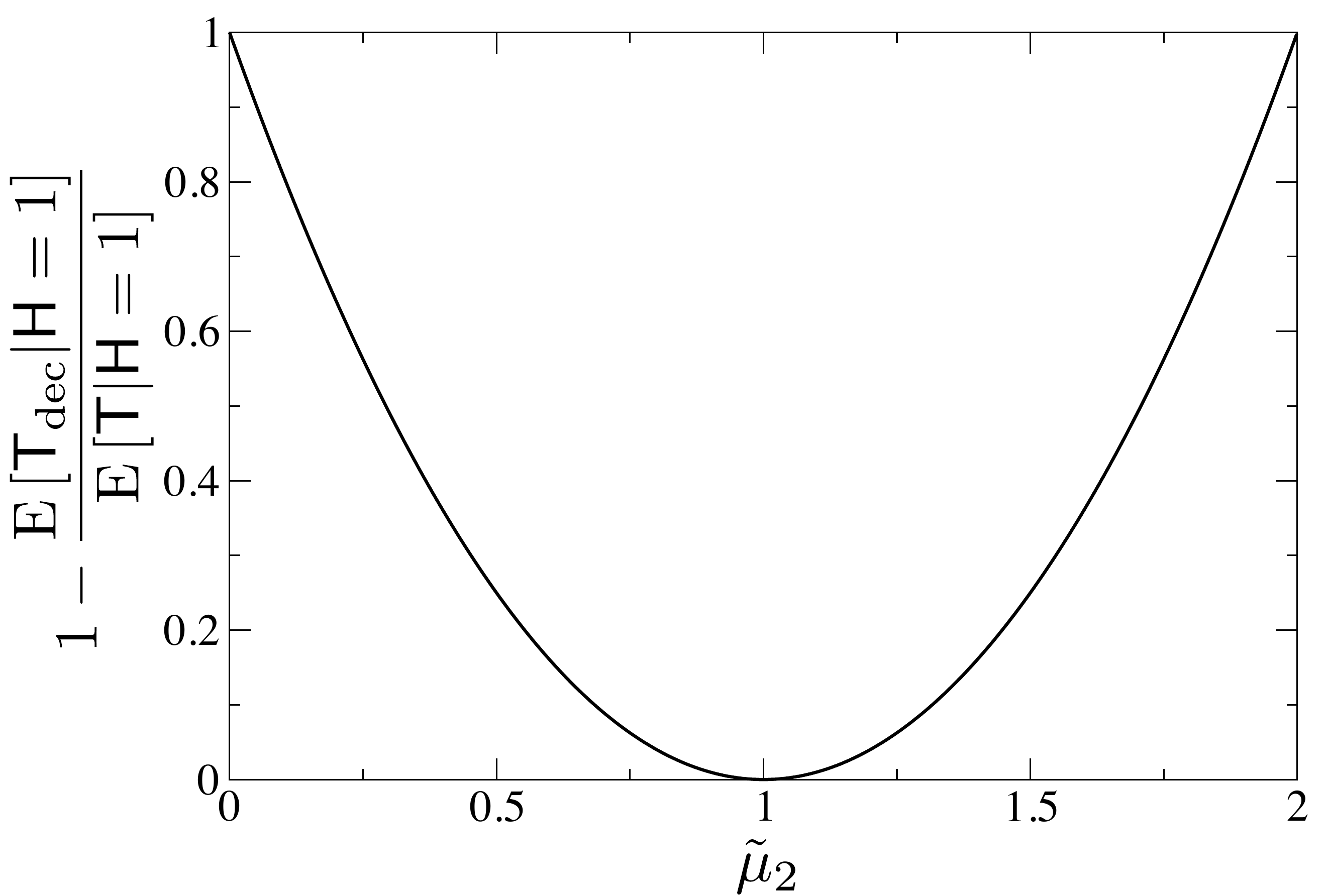}\label{fig:ICont2}}
\subfigure[Empirical estimate of the mutual information as a function of the number of test runs]
{\includegraphics[width=0.48\columnwidth]{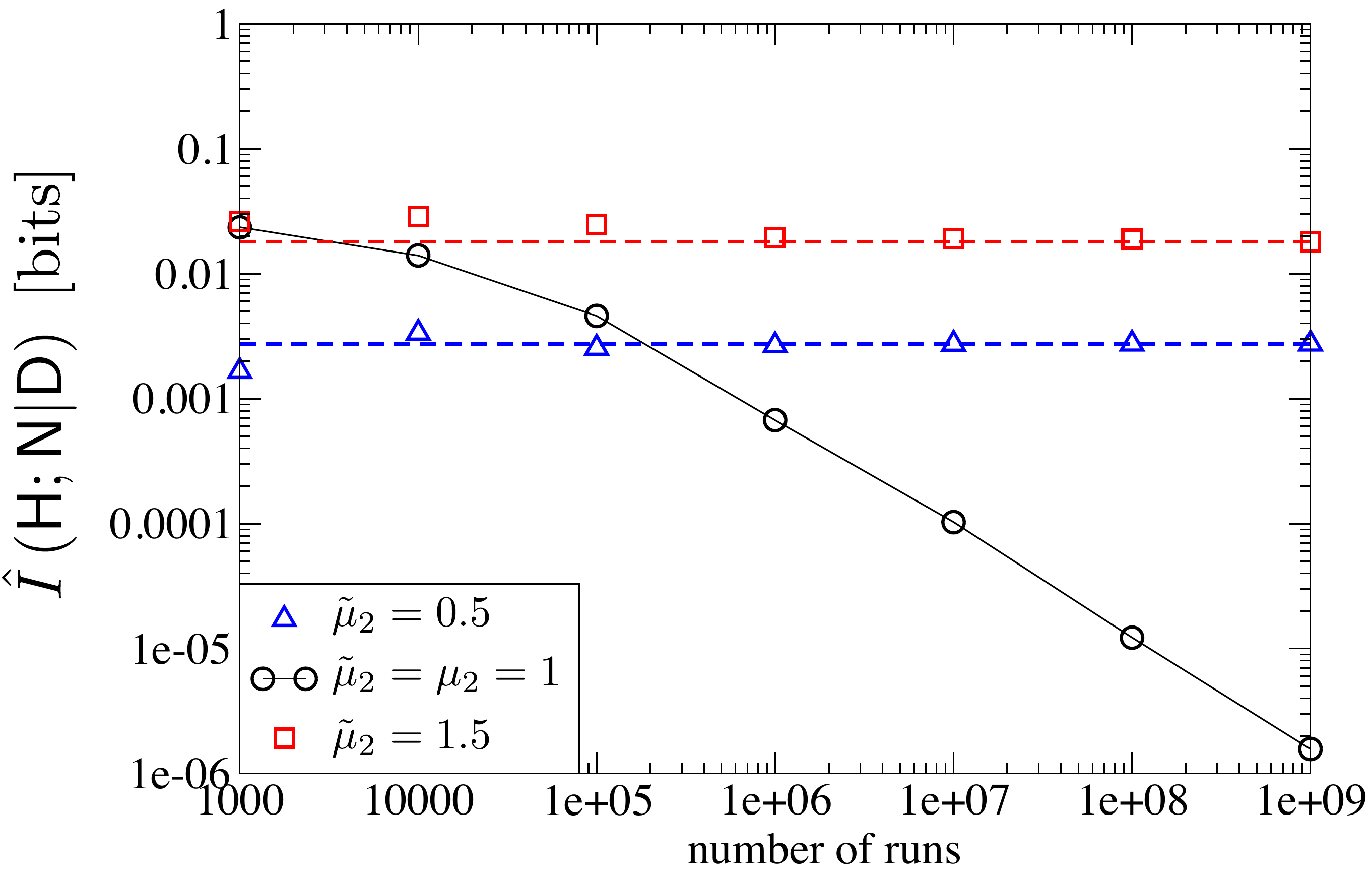}\label{fig:ICont3}}
\caption{Measuring the divergence to optimality in the sense of Definitions \ref{DefinitionMinimumMeanTime} and \ref{DefinitionOptInfUse} using the  {\it continuous} observation process (\ref{eq:Ito}) and the decision-making model (\ref{eq:modelSCont}) - (\ref{Def_D_WaldxxCont}) .   The parameters defining the observation process are   $\mu_1 = 0$, $\mu_2 = 1$ and $\sigma=5$ and the parameters defining the decision-making device are $\tilde{\mu}_1 = 0$, $L_1 = 4$, $\alpha_1 = 0.01$,  $\tilde{\sigma} =\sqrt{\sigma^2\frac{\tilde{\mu}_1-\tilde{\mu}_2}{2 \mu_2-\tilde{\mu}_1-\tilde{\mu}_2}\frac{\log \alpha_1}{L_1}}$ and  $\tilde{\mu}_2$ as given on the abscissa.   
The range of $\tilde{\mu}_2$ plotted  corresponds with  $\tilde{\sigma}\geq 0$.   (a):  Theoretical value of the mutual information $I(\mathsf{H};\mathsf{N}|\mathsf{D})$ as given by (\ref{eq:mutualContInfxx}) for given values of the time resolution $T_{\rm r}$ in the legend.  (b):     Theoretical values of the mean decision time as given by (\ref{eq:TTheory1}) - (\ref{eq:TTheory2}).   (c): Empirical estimate $\hat{I}(\mathsf{H};\mathsf{N}|\mathsf{D})$ as a function of the number of test runs, for given values of $\tilde{\mu}_2$ and for $T_r = \mathrm{E}\left[\mathsf{T}_{\rm dec}|\mathrm{H}=1\right]$.  
  The corresponding theoretical value of the mutual information is indicated by the horizontal dashed lines.  
 Note that we have   generated variates of the random variables $(\mathsf{T}, \mathsf{H}, \mathsf{D})$ according to the distributions (\ref{eq:diffP1}) - (\ref{eq:diffP2}), (\ref{eq:distriTime1}) - (\ref{eq:distriTime4}), and $P\left(\mathsf{H}=1\right) = P\left(\mathsf{H}=2\right) = 1/2$; we have generated variates from the inverse Gaussian distribution  with the algorithm in \cite{michael1976generating}. }\label{figmutualCont}
\end{figure}  

\subsubsection{Illustration of Corollary~\ref{Corollary3}}
We can also compute the mutual information in the limit $L_2\rightarrow -\infty$, which is for $P(\mathsf{H}=1) = P(\mathsf{H}=2) = 1/2$   given by
\begin{IEEEeqnarray}{rCL}
I\left(\mathsf{H};\mathsf{T} |\mathsf{D} \right) &=&  \frac{1+\alpha_1}{2}\log_2\left(1+\alpha_1\right) 
\nonumber\\
 &&- \frac{L_1}{4\sqrt{\pi b}} \int^\infty_0 {\rm d}t\:t^{-3/2} e^{-\frac{\left(|a_1|t-L_1\right)^2}{4bt}}  \log_2\left(1 + \alpha_1 e^{-\frac{\left(|a_2|t-L_1\right)^2}{4bt}+\frac{\left(|a_1|t-L_1\right)^2}{4bt}} \right) \nonumber\\
  &&- \frac{\alpha_1 L_1}{4\sqrt{\pi b}} \int^\infty_0 {\rm d}t\:t^{-3/2}  e^{-\frac{\left(|a_2|t-L_1\right)^2}{4bt}}  \log_2\left(\alpha_1 +  e^{\frac{\left(|a_2|t-L_1\right)^2}{4bt}-\frac{\left(|a_1|t-L_1\right)^2}{4bt}} \right)   \label{eq:mutualContI}
\end{IEEEeqnarray}    
with $\alpha_1  = e^{a_2L_1/b}$  following from (\ref{eq:diffP1}) and $\alpha_1\in[0,1/2]$.
 If $|a_2|= |a_1|$, then (\ref{eq:mutualContI}) yields
  $I\left(\mathsf{H};\mathsf{T} |\mathsf{D} \right)  = 0$, and otherwise   $I\left(\mathsf{H};\mathsf{T} |\mathsf{D} \right)  > 0$.     In Fig.~\ref{fig:ICont}  with the black solid line we illustrate   $  I\left(\mathsf{H};\mathsf{T} |\mathsf{D} \right)$ as a function $\tilde{\mu}_2$ for  values of $\tilde{\sigma} =\sqrt{\sigma^2\frac{\tilde{\mu}_1-\tilde{\mu}_2}{2 \mu_2-\tilde{\mu}_1-\tilde{\mu}_2}\frac{\log \alpha_1}{L_1}}$ such that the  error probability is fixed and we choose $\alpha_1 = 0.01$.    The mutual information is zero for $\tilde{\mu}_2 = \mu_2 =1$, corresponding to $a_2 = -a_1$ and $(\mathsf{T},\mathsf{D}) = (\mathsf{T}_{\rm dec},\mathsf{D}_{\rm dec})$.    In Fig.~\ref{fig:ICont2}   we can easily observe the optimality of the test for  $\tilde{\mu}_2 = \mu_2 =1$    where the decision time takes its minimal value given by  $\mathrm{E}\left[\mathsf{T}_{\rm dec}|\mathsf{D}_{\rm dec}=1, \mathsf{H}=1\right]$. Fig.~\ref{fig:ICont} and Fig.~\ref{fig:ICont2} illustrate the advantage of using mutual information as a measure for 
testing optimality with respect to  the average decision time.   The mutual information is a useful quantity since $I(\mathsf{H};\mathsf{T}|\mathsf{D})=0$ at the optimal point, whereas the  average decision time $\mathrm{E}[\mathsf{T}|\mathsf{H}] = \mathrm{E}[\mathsf{T}_{\rm{dec}}|\mathsf{H}]$, and hence we require knowledge of $\mathrm{E}[\mathsf{T}_{\rm{dec}}|\mathsf{H}]$ to test optimality using decision times.    
    
\subsubsection{Practical implementation of tests for optimality of continuous observation processes}     
Implementation of our tests for optimality in a computer does not allow to directly treat the absolutely continuous random variable $\mathsf{T}$.   Moreover as any practical time measurement device has a finite time resolution, we are only able to retrieve $\mathsf{T}$ up to a finite quantization resolution.   Thus, we discuss here how far finite resolution of $\mathsf{T}$ influences our tests for optimality.  
  Note that we still consider that the decision device operates in continuous time and also that  the observation process is continuous.    Measuring the decision times $\mathsf{T}$ under a finite resolution  $T_{\rm r}$ is equivalent to discretizing the distributions (\ref{eq:distriTime1}) - (\ref{eq:distriTime4}) 
such that $P(\mathsf{N}=n|\mathsf{D},\mathsf{H})=\int_{(n-1)T_r}^{nT_r}dtp_{\mathsf{T}}(t|\mathsf{D},\mathsf{H})$ with $\mathsf{N}\in \mathbb{N}$ being a discrete random variable.   Corresponding to (\ref{eq:mutualContI}) the mutual information   $I\left(\mathsf{H};\mathsf{N} |\mathsf{D} \right)$ can be expressed by 
\begin{IEEEeqnarray}{rCL}
\lefteqn{I\left(\mathsf{H};\mathsf{N} |\mathsf{D} \right) =  \frac{1+\alpha_1}{2}\log_2\left(1+\alpha_1\right)  }&&
\nonumber\\
&& - \frac{L_1}{4\sqrt{\pi b }} \sum^{\infty}_{n=1}\: \int^{nT_{\rm r}}_{(n-1)T_{\rm r}} \frac{dt}{t^{3/2}}\exp\left(-\frac{\left(|a_1| t-L_1\right)^2}{4bt}\right)  \log_2\left(1 + \alpha_1  \frac{\int^{nT_{\rm r}}_{(n-1)T_{\rm r}} \frac{dt}{t^{3/2}} \exp\left(-\frac{\left(|a_2| t-L_1\right)^2}{4bt}\right)}{\int^{nT_{\rm r}}_{(n-1)T_{\rm r}} \frac{dt}{  t^{3/2}} \exp\left(-\frac{\left(|a_1| t-L_1\right)^2}{4bt}\right)}  \right)  \nonumber\\
 && - \frac{\alpha_1 L_1}{4\sqrt{\pi b }} \sum^{\infty}_{n=1} \: \int^{nT_{\rm r}}_{(n-1)T_{\rm r}} \frac{dt}{t^{3/2}} \exp\left(-\frac{\left(|a_2| t-L_1\right)^2}{4bt}\right)     \log_2\left(\alpha_1 +  \frac{\int^{nT_{\rm r}}_{(n-1)T_{\rm r}} \frac{dt}{ t^{3/2}} \exp\left(-\frac{\left(|a_1| t-L_1\right)^2}{4bt}\right)}{\int^{nT_{\rm r}}_{(n-1)T_{\rm r}} \frac{dt}{  t^{3/2}} \exp\left(-\frac{\left(|a_2| t-L_1\right)^2}{4bt}\right)}  \right).  \nonumber\\
     \label{eq:mutualContInfxx}
  \end{IEEEeqnarray}      
   Fig.~\ref{fig:ICont}  illustrates  the impact of the discretization time $T_r$ on  $I\left(\mathsf{H};\mathsf{N} |\mathsf{D} \right)$.   Note that for $a_1 = -a_2$, corresponding to $\tilde{\mu}_2=1$, the mutual information $I\left(\mathsf{H};\mathsf{N} |\mathsf{D} \right) = 0$ for any value of $T_r$, since  by the data processing inequality time discretization of $\mathsf{T}$ can just discard information \cite[Theorem~2.8.1]{CoverBook2}.    However, for  $a_1 \neq -a_2$, corresponding to $\tilde{\mu}_2\ne 1$, the mutual information might significantly decrease because of discarding information by time discretization.  Fig.~\ref{fig:ICont} shows that for $T_{\rm r} \sim 0.1 \:{\rm E}\left[\mathsf{T}_{\rm dec}|\mathsf{H}=1\right]$  the mutual information $I\left(\mathsf{H};\mathsf{N} |\mathsf{D} \right) \approx I\left(\mathsf{H};\mathsf{T} |\mathsf{D} \right)$ and the effect of finite resolution is negligible.

 Direct implementation of our tests for optimality also requires to deal with a finite number of runs of the test.   In    Fig.~\ref{fig:ICont3} we evaluate the dependency of the estimate  $\hat{I}\left(\mathsf{H};\mathsf{N} |\mathsf{D} \right)$ on the number of runs of the test for suboptimal tests ($\tilde{\mu}_2 = 0.5$, $\tilde{\mu}_2 = 1.5$) and an optimal test ($\tilde{\mu}_2 = 1$).    The estimate of the mutual information decreases with the number of test runs,  and for suboptimal tests converges to a theoretical value which is larger than zero.    For an optimal sequential decision-making test, the estimate of the  mutual information  converges to zero.   Note that this is contrary to the case of discrete-time observation processes, as illustrated in    Fig.~\ref{figruns},  where the estimate of the mutual information, even in the optimal case,  saturates as a function of the number of test runs and  converges to a positive value.
 
\section{Discussion}\label{sec:7} 
In the present paper we have shown that optimality of black box decision devices can be tested by studying decision time distributions given the knowledge of the actual hypothesis and the decision variable. To obtain these results we have shown that decisions times of binary sequential probability ratio tests of continuous observation processes satisfy fluctuation relations given by Theorem~\ref{Theorem1} and Theorem~\ref{Theorem2}.   Based on these fluctuation relations we have shown that the conditional mutual information $I(\mathsf{H}, \mathsf{T}_{\rm dec}|\mathsf{D}_{\rm dec})$ between the hypothesis $\mathsf{H}$, the decision time $\mathsf{T}_{\rm dec}$ conditioned on the decision variable $\mathsf{D}_{\rm dec}$ is equal to zero, see Corollary~\ref{Corollary3}. Using several numerical experiments we have illustrated our statistical tests.   We have also discussed the limitations of our tests for sequential decision-making based on  discrete-time observations.

Applying our tests for optimality has several advantageous properties.   Testing the necessary conditions given by Theorem~\ref{Theorem1}  and Corollary~\ref{Corollary3}  requires knowledge about three random variables, namely, the hypothesis $\mathsf{H}$,  the decision variable $\mathsf{D}$, and the output time  of the decision device $ \mathsf{T} + \mathsf{T}_{\rm delay}$.   Note that we do not require direct measurements of the decision time $\mathsf{T}$, but allow for random or deterministic delay $\mathsf{T}_{\rm delay}$ in the output time, which needs to be statistically independent of  $\mathsf{H}$ conditioned on $\mathsf{T}$ and $\mathsf{D}$.   Remarkably,  the statistics of the actual observation process and  the properties of the decision-making device, such as the allowed error probabilities $\alpha_1$ and $\alpha_2$, are not required.      For these reasons our tests are well applicable under practical experimental conditions.
   We now discuss a few practical examples.

Studies in cognitive psychology  have measured the reaction time distributions in experiments of   two-choice decision tasks performed by human subjects about simple perceptual and cognitive stimuli, see e.g. \cite{ratcliff2004comparison, ratcliff2008diffusion}.    For fast decisions  -- of the order of one second -- distributions of reaction times and error probabilities can  be well described with a  simple model for  sequential decision-making in  continuous time   \cite{ratcliff2004comparison, ratcliff2008diffusion}.    Neural activity associated with the actual decision-making process has been identified in experiments with  rhesus monkeys  trained to perform rapid two-choice decisions in simple visual tasks \cite{ shadlen2001neural, roitman2002response}.   Interestingly, it was found that the firing rates of neurons in the lateral intra-parietal area correlate with  the cumulative evidence associated with the hypothesis,  and that a decision model based on a threshold crossing process describes the decision-time data well \cite{kira2015neural}.
Furthermore, it has been conjectured that  the cortex and basal ganglia, two brain regions in vertebrates, perform a multihypothesis sequential probability ratio test \cite{bogacz2007basal, bogacz2007optimal}, which is  optimal for small error probabilities \cite{draglia1999multihypothesis, tartakovsky2014sequential}.    Theorem~\ref{Theorem1} and Corollary~\ref{Corollary3} may be used as  tools to quantify the closeness to  optimality of sequential decision making by human subjects or monkeys in two-choice decision tasks.   In this regard, note that experiments of  two-choice decision tasks   performed by human subjects or monkeys allow to measure reaction times,  decision variables, and the actual realizations of the  hypothesis, which are known by the construction of the experiment.

Cell fate decisions are important changes of cell behavior in response to external signals.
Examples are cell division controlled by growth factors, programmed cell death due to signals or the
differentiation of pluripotent progenitor cell to a specific cell type as a result of biochemical signals.
Cellular signaling events that control cell fate can involve signaling molecules, such as, hormones, growth factors, and cytokines \cite{losick2008stochasticity, raj2008nature}.  Because of intrinsic and extrinsic noise, cellular signaling processes have a stochastic component.
Cell-fate decisions can be considered as an example of sequential decision-making based on a sequence of
noisy input signals. An example of  how cells could implement sequential probability ratio tests with simple examples of protein reaction networks has been given in  \cite{Siggia2013}. Theorem~\ref{Theorem1} and Corollary~\ref{Corollary3} could
be used to investigate the degree of optimality of cell-fate decisions. The timing of cell-fate decisions could be measured in experiments by monitoring the expression levels of fluorescently labelled molecular markers associated with the cell-fate transition within clonal populations \cite{losick2008stochasticity, raj2008nature}. Following at the same time the input signals could in principle permit to calculate the differences of decision time distributions of correct and incorrect
decisions.

As already stated with our introductory example on obstacle detection for autonomous cars, sequential binary decision problems arise in many engineered systems. However, different to the assumption made for the Wald test the statistics of the observation processes $\mathbb{P}_l$ ($l=1,2$) are often unknown, corresponding to a nonparametric decision problem. One approach to tackle such sequential decision problems is to apply neural networks in combination with reinforcement learning \cite{GuoKuh97}. The approach presented in \cite{GuoKuh97} closely approximates the behavior of the optimal sequential probability ratio test and achieves a similar performance. Alternatively, in \cite{teng2015learning} an approach for nonparametric binary sequential hypothesis testing is presented, where the binary sequential detector is learned form training samples based on a so-called \emph{Wald-Kernel}. The aim of these algorithms is to use the available measurements in an optimal way such that the average time to take a decision is minimized. However, the behavior of algorithms like neural networks \cite{bishop2006pattern}, \cite{engel2001statistical} can hardly be analyzed making them similar to a black-box decision device causing the problem to verify their optimality which nevertheless is crucial for application in safety critical systems like autonomous cars. This gap can be filled by out test for optimality based on Theorem~\ref{Theorem1} and Corollary~\ref{Corollary3} allowing to determine the degree of optimality of these decision-making devices just requiring the actual hypothesis $\mathsf{H}$, the decision variable $\mathsf{D}$ and the decision time $\mathsf{T}$ of several test runs. This is especially important to determine, whether the learning process already converged sufficiently.

So far our approach is limited to binary sequential probability ratio tests without prior knowledge on the hypotheses.     Sequential probability ratio tests have been extended to a Bayesian setting where prior knowledge on the hypothesis $\mathsf{H}$ is available \cite[Ch.~6.2]{Melsa1978}, and have also been extended to the multihypothesis scenario. The extension of
our results to these settings is for further study.

\appendices
\section{Proof of Theorem~\ref{Theorem2}}\label{App_ProofTheo2}
\begin{proof}
We first show that the log-likelihood ratio $\mathsf{S}_t$ is odd under the transformation given by the involution $\Theta$. This can be shown as follows
\begin{IEEEeqnarray}{rCL}
e^{\mathsf{S}_t(\Theta(\omega))}&=&\frac{\mathrm{d}\mathbb{P}_1|_{\mathcal{F}_t}}{\mathrm{d}\mathbb{P}_2|_{\mathcal{F}_t}}(\Theta(\omega))\label{eq:31+}\\
&=&\frac{\mathrm{d}\mathbb{P}_1|_{\mathcal{F}_t}}{\mathrm{d}(\mathbb{P}_1\circ\Theta)|_{\mathcal{F}_t}}(\Theta(\omega))\\
&=&\frac{\mathrm{d}(\mathbb{P}_1\circ\Theta)|_{\mathcal{F}_t}}{\mathrm{d}(\mathbb{P}_1\circ\Theta\circ\Theta)|_{\mathcal{F}_t}}(\omega)\\
&=&\frac{\mathrm{d}(\mathbb{P}_1\circ\Theta)|_{\mathcal{F}_t}}{\mathrm{d}\mathbb{P}_1|_{\mathcal{F}_t}}(\omega)\\
&=&\frac{\mathrm{d}\mathbb{P}_2|_{\mathcal{F}_t}}{\mathrm{d}\mathbb{P}_1|_{\mathcal{F}_t}}(\omega)\\
&=&e^{-\mathsf{S}_t(\omega)}. \label{eq:36}
\end{IEEEeqnarray}
Let 
\begin{IEEEeqnarray}{rCL}
\Phi_1(t)&=&\{\omega\in\Omega:\ \mathsf{T}_{\mathrm{dec}}(\omega)\le t \textrm{ and } \mathsf{D}_{\rm dec}(\omega)=1\}\\
\Phi_2(t)&=&\{\omega\in\Omega:\ \mathsf{T}_{\mathrm{dec}}(\omega)\le t \textrm{ and } \mathsf{D}_{\rm dec}(\omega)=2\}
\end{IEEEeqnarray}
be the set of trajectories for which the decision time does not exceed $t$ and the test decides for $\mathsf{D}_{\rm dec}=1$ and $\mathsf{D}_{\rm dec}=2$, respectively.
Since $\alpha_1=\alpha_2$ we have also $L_1=-L_2$, and because of the property $\mathsf{S}_t(\Theta(\omega))=-\mathsf{S}_t(\omega)$, it follows that 
\begin{IEEEeqnarray}{rCL}
\Phi_1(t)=\Theta\left(\Phi_2(t)\right). 
\end{IEEEeqnarray}
Therefore, also
\begin{IEEEeqnarray}{rCL}
\mathbb{P}_2(\Phi_1(t))&=&(\mathbb{P}_1\circ \Theta)(\Phi_{1}(t))\\
&=&\mathbb{P}_1(\Theta(\Phi_{1}(t)))\\
&=&\mathbb{P}_1(\Phi_{2}(t)).
\end{IEEEeqnarray}
Now the following holds
\begin{IEEEeqnarray}{rCL}
\mathbb{P}_1(\Phi_{1}(t))
&=&\int_{\omega\in\Phi_{1}(t)}\mathrm{d}\mathbb{P}_1|_{\mathcal{F}_t}\label{Eq42}\\
&=&\int_{\omega\in\Phi_{1}(t)}e^{\mathsf{S}_t}\mathrm{d}\mathbb{P}_2|_{\mathcal{F}_t} \label{ProofFluc2_1}\\
&=&\int_{\omega\in\Phi_{1}(t)}e^{\mathsf{S}_t}\mathrm{d}(\mathbb{P}_1\circ\Theta)|_{\mathcal{F}_t} \label{ProofFluc2_2}\\
&=&\int_{\omega\in\Theta(\Phi_{1}(t))}e^{\mathsf{S}_t(\Theta (\omega))}\mathrm{d}\mathbb{P}_1|_{\mathcal{F}_t} \label{ProofFluc2_3}\\
&=&\int_{\omega\in\Phi_{2}(t)}e^{-\mathsf{S}_t}\mathrm{d}\mathbb{P}_1|_{\mathcal{F}_t} \label{ProofFluc2_4}\\
&=&\int_{\omega\in\Phi_{2}(t)}e^{-\mathsf{S}_{\mathsf{T}_{\mathrm{dec}}}}\mathrm{d}\mathbb{P}_1|_{\mathcal{F}_t} \label{ProofFluc2_5}\\
&=&e^{-L_2}\int_{\omega\in\Phi_{2}(t)}\mathrm{d}\mathbb{P}_1|_{\mathcal{F}_t}\label{ProofFluc2_6}\\
&=&e^{-L_2}\ \mathbb{P}_1(\Phi_{2}(t))\label{ProofFluc2_7}
\end{IEEEeqnarray}
where for (\ref{ProofFluc2_1}) we have used the Radon-Nikod\'ym theorem and the definition in (\ref{DefLikelihoodCont}). For equality~(\ref{ProofFluc2_2}) we have used the involution relation~(\ref{InvolutionCond}) between the measures.  In~(\ref{ProofFluc2_3}) we have applied a variable transformation in the integral. In~(\ref{ProofFluc2_4}) we have used the sign reversal of $\mathsf{S}_t$ given by Eqs.~(\ref{eq:31+}) and~(\ref{eq:36}) and the involution relation between the sets $\Phi_2(t)=\Theta\left(\Phi_1(t)\right)$. In~(\ref{ProofFluc2_5}) we have applied Doob's optional sampling theorem to the $\mathbb{P}_1$-martingale $e^{-\mathsf{S}_t}$. For~(\ref{ProofFluc2_6}) we have used that $e^{-\mathsf{S}_t}$ is a continuous process and reaches the value $e^{-L_2}$  at the time $\mathsf{T}_{\rm dec}$. 

The probability density functions of $\mathsf{T}_{\rm dec}$ can be expressed in terms of the derivatives of the cumulative distributions $\mathbb{P}(\Phi_k(t))$ ($k=1,2$)
\begin{IEEEeqnarray}{rCL}
p_{\mathsf{T}_{\rm dec}}(t|\mathsf{H}=1,\mathsf{D}_{\rm dec}=1)P(\mathsf{D}_{\rm dec}=1|\mathsf{H}=1)&=&\frac{\mathrm{d}}{\mathrm{d}t}{\mathbb{P}_1(\Phi_1(t))}\label{xDistribDef1}\\
p_{\mathsf{T}_{\rm dec}}(t|\mathsf{H}=1,\mathsf{D}_{\rm dec}=2)P(\mathsf{D}_{\rm dec}=2|\mathsf{H}=1)&=&\frac{\mathrm{d}}{\mathrm{d}t}{\mathbb{P}_1(\Phi_2(t))}.\label{xDistribDef2}
\end{IEEEeqnarray}
For the ratio of the decision probabilities we find
\begin{IEEEeqnarray}{rCL}
\frac{P(\mathsf{D}_{\rm dec}=1|\mathsf{H}=1)}{P(\mathsf{D}_{\rm dec}=2|\mathsf{H}=1)}=\frac{1-\alpha_2}{\alpha_2}=\frac{1-\alpha_1}{\alpha_2}=e^{-L_2} \label{eq:56} 
\end{IEEEeqnarray}
which follows from 
\begin{IEEEeqnarray}{rCL}
P(\mathsf{D}_{\rm dec}=1|\mathsf{H}=1) &=& \lim_{t\to\infty} \mathbb{P}_1(\Phi_{1}(t)) \\
P(\mathsf{D}_{\rm dec}=2|\mathsf{H}=1) &=& \lim_{t\to\infty} \mathbb{P}_1(\Phi_{2}(t)). \end{IEEEeqnarray}
Eq.~(\ref{ProofFluc2_7}), and the assumption that the test terminates almost surely. Notice that we have used symmetric error probabilities for which $\alpha_1=\alpha_2$. 
Taking the derivative of the LHS of (\ref{Eq42}) and the RHS of (\ref{ProofFluc2_7}) and using Eq.~(\ref{eq:56}) we prove Eq.~(\ref{eq:29}). Analogously, Eq.~(\ref{eq:30}) can be proved.  

Equation~(\ref{eq:31}) follows from the identities 
\begin{eqnarray}
\lefteqn{p_{\mathsf{T}_{\rm dec}}(t|\mathsf{D}_{\rm dec}=d) }&& \nonumber\\ 
 &=&\! \! p_{\mathsf{T}_{\rm dec}}(t|\mathsf{D}_{\rm dec}=d,\mathsf{H}=1 )P\left(\mathsf{H}=1|\mathsf{D}_{\rm dec}=d\right) + p_{\mathsf{T}_{\rm dec}}(t|\mathsf{D}_{\rm dec}=d,\mathsf{H}=2 )P\left(\mathsf{H}=2|\mathsf{D}_{\rm dec}=d\right)  \nonumber
\\ 
&=&\!\!  p_{\mathsf{T}_{\rm dec}}(t|\mathsf{D}_{\rm dec}=d,\mathsf{H}=1 )P\left(\mathsf{H}=1|\mathsf{D}_{\rm dec}=d\right) + p_{\mathsf{T}_{\rm dec}}(t|\mathsf{D}_{\rm dec}=d,\mathsf{H}=1 )P\left(\mathsf{H}=2|\mathsf{D}_{\rm dec}=d\right) \nonumber
\\ 
&=&\!\!  p_{\mathsf{T}_{\rm dec}}(t|\mathsf{D}_{\rm dec}=d,\mathsf{H}=1 ) \label{eq:57}
\end{eqnarray}  
with $d\in\left\{1,2\right\}$ and 
where we have used Theorem~\ref{Theorem1}.    Using (\ref{eq:29}) and (\ref{eq:57}) we find (\ref{eq:31}), which completes the proof.
\end{proof}

\section{Proof of Corollary~\ref{Corollary4}}\label{App_Corollary4}
\begin{proof}
The mutual information in (\ref{MutInfIndepTermTime_Cont2}) is given by
\begin{IEEEeqnarray}{rCL}
&&I(\mathsf{H};\mathsf{T}_{\rm dec}) = \mathrm{E}\left[\log \left(\frac{p_{\mathsf{T}_{\rm dec}}(\mathsf{T}_{\rm dec}|\mathsf{H})}{p_{\mathsf{T}_{\rm dec}}(\mathsf{T}_{\rm dec})}\right)\right] 
\nonumber\\ 
&&\quad= \mathrm{E}\left[\log \left(\frac{P(\mathsf{D}_{\rm dec}=1) p_{\mathsf{T}_{\rm dec}}(\mathsf{T}_{\rm dec}|\mathsf{H}, \mathsf{D}_{\rm dec}=1) +P(\mathsf{D}_{\rm dec}=2) p_{\mathsf{T}_{\rm dec}}(\mathsf{T}_{\rm dec}|\mathsf{H}, \mathsf{D}_{\rm dec}=2)}{p_{\mathsf{T}_{\rm dec}}(\mathsf{T}_{\rm dec})}\right)\right]. \nonumber\\\label{eq:corollary67}
\end{IEEEeqnarray}
We find 
\begin{IEEEeqnarray}{rCL}
\frac{P(\mathsf{D}_{\rm dec}=1)}{P(\mathsf{D}_{\rm dec}=2)} &=&\frac{P(\mathsf{D}_{\rm dec}=1|\mathsf{H}=1)P(\mathsf{H}=1) +P(\mathsf{D}_{\rm dec}=1|\mathsf{H}=2)P(\mathsf{H}=2) }{P(\mathsf{D}_{\rm dec}=2|\mathsf{H}=1)P(\mathsf{H}=1) +P(\mathsf{D}_{\rm dec}=2|\mathsf{H}=2)P(\mathsf{H}=2)} \\  
&=& \frac{P(\mathsf{D}_{\rm dec}=1|\mathsf{H}=1)}{P(\mathsf{D}_{\rm dec}=2|\mathsf{H}=1)}\:\cdot\: \frac{P(\mathsf{H}=1) +\frac{P(\mathsf{D}_{\rm dec}=1|\mathsf{H}=2)}{P(\mathsf{D}_{\rm dec}=1|\mathsf{H}=1)}P(\mathsf{H}=2) }{P(\mathsf{H}=1) +\frac{P(\mathsf{D}_{\rm dec}=2|\mathsf{H}=2)}{P(\mathsf{D}_{\rm dec}=2|\mathsf{H}=1)}P(\mathsf{H}=2)} 
\\ 
&=&  \frac{1-\alpha}{\alpha}\:\cdot\: \frac{P(\mathsf{H}=1) +\frac{\alpha}{1-\alpha}P(\mathsf{H}=2) }{P(\mathsf{H}=1) +\frac{1-\alpha}{\alpha}P(\mathsf{H}=2)} \label{eq:70}  \\ 
&=&  \frac{(1-\alpha)P(\mathsf{H}=1) + \alpha P(\mathsf{H}=2) }{\alpha P(\mathsf{H}=1) +(1-\alpha)P(\mathsf{H}=2)}   \\
&=& 1 \label{eq:corollary72}
\end{IEEEeqnarray}
where we have used that  the priors on $\mathsf{H}$ are identical, and that $\alpha_1=\alpha_2=\alpha$.      For (\ref{eq:70}) we have used (\ref{eq:27}) and (\ref{eq:56}).    As $\mathsf{D}_{\rm dec}$ is a binary random variable, it follows that $P(\mathsf{D}_{\rm dec}=1) = P(\mathsf{D}_{\rm dec}=2) = \frac{1}{2}$.

It also holds that 
\begin{IEEEeqnarray}{rCL}
p_{\mathsf{T}_{\rm dec}}(\mathsf{T}_{\rm dec}) &=& \sum_{h\in \left\{1,2\right\}}\sum_{d\in\left\{1,2\right\}} P(\mathsf{H}=h)P(\mathsf{D}_{\rm dec}=d|\mathsf{H}=h)\:p_{\mathsf{T}_{\rm dec}}(\mathsf{T}_{\rm dec}|\mathsf{H}=h, \mathsf{D}_{\rm dec}=d) \\
&=& \sum_{h\in \left\{1,2\right\}}\sum_{d\in\left\{1,2\right\}} P(\mathsf{H}=h)P(\mathsf{D}_{\rm dec}=d|\mathsf{H}=h)\:p_{\mathsf{T}_{\rm dec}}(\mathsf{T}_{\rm dec}|\mathsf{H}, \mathsf{D}_{\rm dec}=d) \label{eq:step1}\\ 
&=&p_{\mathsf{T}_{\rm dec}}(\mathsf{T}_{\rm dec}|\mathsf{H}, \mathsf{D}_{\rm dec}) \sum_{h\in \left\{1,2\right\}}\sum_{d\in\left\{1,2\right\}} P(\mathsf{H}=h)P(\mathsf{D}_{\rm dec}=d|\mathsf{H}=h) \label{eq:step2}\\ 
&=& p_{\mathsf{T}_{\rm dec}}(\mathsf{T}_{\rm dec}|\mathsf{H}, \mathsf{D}_{\rm dec}) \label{eq:corollary76}
\end{IEEEeqnarray}
where in (\ref{eq:step1}) we have used Theorem~\ref{Theorem1} and in (\ref{eq:step2}) we have used Theorem~\ref{Theorem2}.  

Using (\ref{eq:corollary72}) and (\ref{eq:corollary76}) we get for the argument of the $\log$ in (\ref{eq:corollary67})
\begin{IEEEeqnarray}{rCL}
\lefteqn{\frac{P(\mathsf{D}_{\rm dec}=1) p_{\mathsf{T}_{\rm dec}}(\mathsf{T}_{\rm dec}|\mathsf{H}, \mathsf{D}_{\rm dec}=1) +P(\mathsf{D}_{\rm dec}=2) p_{\mathsf{T}_{\rm dec}}(\mathsf{T}_{\rm dec}|\mathsf{H}, \mathsf{D}_{\rm dec}=2)}{p_{\mathsf{T}_{\rm dec}}(\mathsf{T}_{\rm dec})}}&\ \ \ \ \ \ \   &
\nonumber\\
&&=\frac{p_{\mathsf{T}_{\rm dec}}(\mathsf{T}_{\rm dec}|\mathsf{H}, \mathsf{D}_{\rm dec}=1) +p_{\mathsf{T}_{\rm dec}}(\mathsf{T}_{\rm dec}|\mathsf{H}, \mathsf{D}_{\rm dec}=2)}{2\,p_{\mathsf{T}_{\rm dec}}(\mathsf{T}_{\rm dec}|\mathsf{H}, \mathsf{D}_{\rm dec})} 
\\ 
&&=\frac{p_{\mathsf{T}_{\rm dec}}(\mathsf{T}_{\rm dec}|\mathsf{H}, \mathsf{D}_{\rm dec})}{p_{\mathsf{T}_{\rm dec}}(\mathsf{T}_{\rm dec}|\mathsf{H}, \mathsf{D}_{\rm dec})}  
\\ 
&&= 1
\end{IEEEeqnarray}
where we have applied again Theorem~\ref{Theorem2}.   This completes the proof.
\end{proof}

\section*{Acknowledgement} 
We acknowledge Yannis Kalaidzidis, Mostafa Khalili-Marandi, and Marino Zerial  for fruitful discussions.

\bibliographystyle{IEEEtran}
\bibliography{IEEEabrv,Bib_mod_IEEE}

\end{document}